\RequirePackage[svgnames]{xcolor}
\RequirePackage[log-declarations=false]{xparse}

\documentclass[a4paper,USenglish,cleveref,autoref,numberwithinsect]{lipics-v2019}

\usepackage{preamble}
\NewDocumentCommand\CiteExtended{}{the appendices}

\title{Cubical Syntax for Reflection-Free Extensional Equality}

\author{Jonathan Sterling}{%
  Carnegie Mellon University
  \and
  \url{http://cs.cmu.edu/~jmsterli}%
}{%
  jmsterli@cs.cmu.edu%
}{%
  https://orcid.org/0000-0002-0585-5564%
}{}

\author{Carlo Angiuli}{%
  Carnegie Mellon University
  \and
  \url{http://cs.cmu.edu/~cangiuli}%
}{%
  cangiuli@cs.cmu.edu%
}{%
  https://orcid.org/0000-0002-9590-3303%
}{}

\author{Daniel Gratzer}{%
  Aarhus University
  \and
  \url{http://jozefg.github.io}
}{%
  gratzer@cs.au.dk%
}{%
  https://orcid.org/0000-0003-1944-0789%
}{}

\authorrunning{J. Sterling, C. Angiuli, and D. Gratzer}
\Copyright{Jonathan Sterling, Carlo Angiuli, and Daniel Gratzer}



\category{}

\supplement{}

\funding{%
  The authors gratefully acknowledge the support of the Air Force Office of
  Scientific Research through MURI grant FA9550-15-1-0053. Any opinions,
  findings and conclusions or recommendations expressed in this material are
  those of the authors and do not necessarily reflect the views of the AFOSR.%
}

\acknowledgements{%
  We thank Steve Awodey, Lars Birkedal, Evan Cavallo, David Thrane Christiansen, Thierry Coquand, Kuen-Bang Hou
  (Favonia), Marcelo Fiore, Jonas Frey, Krzysztof Kapulkin, Andr\'as Kov\'acs, Dan Licata,
  Conor McBride, Darin Morrison, Anders M\"ortberg, Michael Shulman, Bas Spitters, and Thomas
  Streicher for helpful conversations about extensional equality, algebraic type theory,
  and categorical gluing.
  We thank our anonymous reviewers for their insightful comments, and especially thank
  Robert Harper for valuable conversations throughout the development of this work.
  We also thank Paul Taylor for his \texttt{diagrams} package, which we have used to
  typeset the commutative diagrams in this paper.
}

\nolinenumbers
\hideLIPIcs

\EventEditors{Herman Geuvers}
\EventNoEds{1}
\EventLongTitle{4th International Conference on Formal Structures for Computation and Deduction (FSCD 2019)}
\EventShortTitle{FSCD 2019}
\EventAcronym{FSCD}
\EventYear{2019}
\EventDate{June 24--30, 2019}
\EventLocation{Dortmund, Germany}
\EventLogo{}
\SeriesVolume{131}
\ArticleNo{32}

\begin{document}

\maketitle

\begin{abstract} We contribute \XTT, a cubical reconstruction of Observational
  Type Theory~\Cite{altenkirch-mcbride-swierstra:2007} which extends
  Martin-L\"of's intensional type theory with a \emph{dependent equality type}
  that enjoys function extensionality and a judgmental version of the
  \emph{unicity of identity proofs} principle (UIP): any two elements of the
  same equality type are judgmentally equal.
  Moreover, we conjecture that the typing relation can be decided in a
  practical way. In this paper, we establish an algebraic canonicity theorem
  using a novel cubical extension (independently proposed by Awodey) of the
  \emph{logical families} or \emph{categorical gluing} argument inspired by
  \Cite[Coquand and Shulman]{coquand:2018,shulman:2015}: every closed element
  of boolean type is derivably equal to either $\true$ or $\false$.

\end{abstract}

\section{Introduction}

The past fifty years of constructive type theory can be summed up as the search
for a scientific understanding of \emph{equality}, punctuated by moments of
qualitative change in our perception of the boundary between semantics (actual
construction) and syntax (proof theory) from a type-theoretic point of view.
Computation is critical to both the semantics and syntax of type theory---from
Martin-L\"of's meaning explanations \Cite{martin-lof:1984}, supplying type
theory with its direct semantics and intuitionistic grounding, to syntactic
properties such as closed and open \emph{canonicity} which establish
computation as the indispensible method for deriving equations.

For too long, a limiting perspective on extensional type theory has
prevailed, casting it as a particular syntactic artifact (for instance, the
formalism obtained by stripping of their meaning the rules which
incidentally appear in~\Cite[Martin-L\"of's monograph]{martin-lof:1984}), a formal system which enjoys
precious few desirable syntactic properties and is distinguished primarily by its
\emph{equality reflection} rule.

We insist on the contrary that the importance of extensional type
theory lies not in the specific choice of syntactic presentation (historically,
via equality reflection), but rather in the \emph{semantic} characteristics of
its equality connective, which are invariant under choice of syntax.
The specifics of how such an equality construct is presented syntactically are
entirely negotiable (the internal language of a doctrine is determined only up
to equivalence), and therefore has an empirical component.

\subsection{Internalizing equality: from judgments to types}

Equality in type theory begins with a form of judgment $\Gamma\vdash A = B\
\mathit{type}$, which expresses that $A$ and $B$ are exactly the same type;
because types can depend on terms, one also includes a form of judgment
$\Gamma\vdash M = N : A$ to express that $M$ and $N$ are exactly the same
element of $A$. This kind of equality, called \emph{judgmental equality}, is
silent in the sense that if $\Gamma\vdash A = B\ \mathit{type}$ holds and
$\Gamma\vdash M : A$ holds, then $\Gamma\vdash M : B$ without further ado.

Judgmental equality in type theory is a completely top-level affair: it cannot
be assumed or negated. On the other hand, both programming and mathematics
require one to establish equations under the assumption of other
equations (for instance, as part of an induction hypothesis). For this reason,
it is necessary to \emph{internalize} the judgmental equality $M=N:A$ as a type
$\Eq{A}{M}{N}$ which can be assumed, negated, or inhabited by induction.

The simplest way to internalize judgmental equality as a type is to provide
introduction and elimination rules which make the existence of a proof of
$\Eq{A}{M}{N}$ equivalent to the judgment $M=N:A$:
\begin{mathpar}
  \inferrule[introduction]{
    \Gamma\vdash M = N : A
  }{
    \Gamma\vdash \mathsf{refl} : \Eq{A}{M}{N}
  }
  \and
  \inferrule[elimination]{
    \Gamma\vdash{}P : \Eq{A}{M}{N}
  }{
    \Gamma\vdash{}M=N:A
  }
\end{mathpar}

The \textsc{elimination} rule above is usually called \emph{equality
reflection}, and is characteristic of \emph{extensional} versions of
Martin-L\"of's type theory.
This presentation of the equality type is very strong, and broadens the reach of
judgmental equality into assertions of higher-level.

A consequence of the equality reflection rule is that judgmental equality is no
longer decidable, a pragmatic concern which affects implementation and
usability. On the other hand, equality reflection implies numerous critical
reasoning principles, including function extensionality (if two functions agree
on all inputs, then they are equal), a judgmental version of the famous
\emph{unicity of identity proofs} (UIP) principle (any two elements of the
equality type are equal), and perhaps the most crucial consequence of
internalized equality, \emph{coercion} (if $M:P(a_0)$ and $P:\Eq{A}{a_0}{a_1}$,
then there is some term $P^*(M) : P(a_1)$; in this case, $P^*(M) = M$).

\subsection{Extensional equality via equality reflection}

The earliest type-theoretic proof assistants employed the equality reflection
rule (or equivalent formulations) in order to internalize the judgmental
equality, a method most famously represented by \Nuprl{} \Cite{constable:1986}
and its descendents, including \RedPRL{}~\Cite{redprl:2018:lfmtp}. The
\Nuprl-style formalisms act as a ``window on the truth'' for a \emph{single}
intended semantics inspired by Martin-L\"of's computational meaning
explanations~\Cite{allen:1987:thesis}; semantic justification in the
computational ontology is the \emph{only} consideration when extending the
\Nuprl{} formalism with a new rule, in contrast to other traditions in which
global properties (e.g.\ admissibility of structural rules, decidability of
typing, interpretability in multiple models, etc.) are treated as definitive.

Rather than supporting \emph{type checking}, proof assistants in this style
rely heavily on interactive development of typing derivations using tactics and
partial decision procedures. A notable aspect of the \Nuprl{} family is that
their formal sequents range not over typed terms (proofs), but over untyped raw
terms (realizers); a consequence is that during the proof process, one must
repeatedly establish numerous \emph{type functionality} subgoals, which restore
the information that is lost when passing from a proof to a realizer.  To
mitigate the corresponding blow-up in proof size, \Nuprl{} relies heavily on
untyped computational reasoning via pointwise functionality, a non-standard
semantics for dependently typed sequents which has some surprising
consequences, such as refuting the principle of \emph{dependent
cut}~\Cite{kopylov:2004}.

Another approach to implementing type theory with equality reflection is
exemplified in the experimental \Andromeda{} proof
assistant~\Cite{bauer-gilbert-haselwarter-pretnar-stone:2016}, in which
proofs are also built interactively using tactics, but judgments range over
abstract proof derivations rather than realizers. This approach mitigates to
some degree the practical problems caused by erasing information prematurely,
and also enables interpretation into a broad class of semantic models.

Although \Nuprl/\RedPRL{} and \Andromeda{} illustrate that techniques beyond
mere type checking are profitable to explore, the authors' experiences building
and using \RedPRL{} for concrete formalization of mathematics underscored the
benefits of having a practical algorithm to check types, particularly in the
setting of cubical type theory (\cref{sec:cubical}), whose higher-dimensional
structure significantly reduces the applicability of \Nuprl-style untyped
reasoning.

In particular, whereas it is possible to treat \emph{all}
$\beta$-rules and many $\eta$-rules in non-cubical type theory as untyped
rewrites, such an approach is unsound for the cubical account of higher
inductive types and univalence~\cite{angiuli-favonia-harper:2018}; consequently,
in \RedPRL{} many $\beta$/$\eta$ rewrites must emit auxiliary proof obligations.
Synthesizing these experiences and challenges led to the creation of the
\redtt{} proof assistant for Cartesian cubical type theory~\Cite{redtt:2018:dagstuhl}.

\subsection{Equality in intensional type theory}

Martin-L\"of's Intensional Type Theory (\ITT)
\Cite{martin-lof:itt:1975,nordstrom-peterson-smith:1990} represents another
extremal point in the internalization of judgmental equality. \ITT{}
underapproximates the equality judgment via its \emph{identity type},
characterized by rules like the following:
\begin{mathpar}
  \inferrule[formation]{
    \Gamma\vdash{}A\ \mathit{type}
    \\
    \Gamma\vdash M:A
    \\
    \Gamma\vdash N:A
  }{
    \Gamma\vdash\mathsf{Id}_A(M,N)\ \mathit{type}
  }
  \and
  \inferrule[introduction]{
    \Gamma\vdash{}M:A
  }{
    \Gamma\vdash{}\mathsf{refl}_A(M):\mathsf{Id}_A(M,M)
  }
  \and
  \inferrule[elimination]{
    \Gamma,x:A,y:A,z:\mathsf{Id}_A(x,y)\vdash C(x,y,z)\ \mathsf{type}
    \\
    \Gamma\vdash P : \mathsf{Id}_A(M,N)
    \\
    \Gamma,x:A\vdash Q : C(x,x,\mathsf{refl}_A(x))
  }{
    \Gamma\vdash{}\mathsf{J}_{x,y,z.C}(P; x.Q) : C(M,N,P)
  }
  \and
  \cdots
\end{mathpar}

Symmetry, transitivity and coercion follow from the elimination rule of the
identity type. Other properties which follow directly from equality reflection,
such as the unicity of identity proofs and function extensionality, are not
validated by \ITT{}; indeed, there are sufficiently intensional models of the
identity type to refute both properties
\Cite{streicher:1994,hofmann-streicher:1998}. While
the desirability of the unicity principle is perhaps up for debate, especially
in light of recent developments in Homotopy Type Theory~\Cite{hottbook},
theorists and practitioners alike generally agree that
function extensionality is desirable.

A significant selling-point for \ITT{} is that, by avoiding equality reflection,
it presents a theory which can be implemented using type checking and
normalization. Consequently, $\beta$ and $\eta$ rules are totally automatic and
never require intervention from the user---in contrast to systems like
\RedPRL{}, whose users are accustomed to establishing $\beta/\eta$ equivalences
by hand at times when heuristical tactics prove inadequate. The downsides
of pure \ITT{}, however, are manifold: function extensionality is absolutely
critical in practice.

\subsection{Setoids and internal model constructions}\label{sec:setoids}

A standard technique for avoiding the deficiencies of the identity type in
\ITT{} is the \emph{setoid construction}~\Cite{hofmann:1995},
an exact completion which glues an equivalence relation ${=_A}$ onto each type
$\verts{A}$ in the spirit of Bishop~\Cite{bishop:1967}. When using setoids, a
function $A\to B$ consists of a type-theoretic function
$f:\verts{A}\to\verts{B}$ together with a proof that it preserves the
equivalence relation, $f_= : \picl{x,y}{\verts{A}}{x \mathrel{=_A} y \to f(x)
\mathrel{=_B} f(y)}$; a \emph{dependent setoid} (family of setoids) is a
type-theoretic family equipped with a coherent coercion operator.

Setoids are a discipline for expressing internally precisely the
extrinsic properties required for constructions to be
extensional (compatible with equality); these extra proof
obligations must be satisfied in parallel with constructions at every turn. The
state of affairs for setoids is essentially analogous to that of proof
assistants with equality reflection, in which type functionality subgoals play
a similar role to the auxiliary paperwork generated by setoids.

Paradoxically, however, every construction in ordinary \ITT{}
is automatically extensional in this sense.
A solution to the problem of equality in type theory should, unlike setoids,
take advantage of the fact that type theory is already restricted to
extensional constructions, adding to it only enough language to refer to
equality internally. This is the approach taken by both Observational Type
Theory and \XTT.

\subsection{Observational Type Theory}\label{sec:ott}

The first systematic solution to the problem of syntax for extensional equality
without equality reflection was Observational Type Theory
(\OTT)~\Cite{altenkirch-mcbride:2006,altenkirch-mcbride-swierstra:2007}, which
built on \Cite[early work by Altenkirch and McBride]{altenkirch:1999, mcbride:1999}. The central idea of
\OTT{} is to work with a \emph{closed} universe of types, defining by
recursion for each pair of types $A,B$ a type $\mathsf{Eq}(A,B)$ of proofs that $A$ and $B$
are equal, and for each pair of elements $M:A$ and $N:B$, a
type of proofs $\mathsf{Eq}_{A,B}(M,N)$ that $M$ and $N$ are (heterogeneously) equal. Finally, one
defines ``generic programs'' by recursion on type structure which calculate
coercions and coherences along proofs of equality.

One can think of \OTT{} as equipping the semantic setoid construction with a
direct-style type-theoretic language, and adding to it closed,
inductively defined universes of types. The heterogeneous equality of \OTT{},
initially a simplifying measure adopted from \Cite[McBride's thesis]{mcbride:1999}, is an early
precursor of the \emph{dependent paths} which appear in Homotopy Type
Theory~\Cite{hottbook}, Cubical Type
Theory~\Cite{cchm:2017,abcfhl:2019,angiuli-favonia-harper:2018}, and \XTT.

Recently, McBride and his collaborators have made progress toward a cubical
version of \OTT{}, using a different cube category and coercion structure, in
which one coerces only from $0$ to $1$, and obtains fillers using an affine
rescaling operation~\Cite{chapman-forsberg-mcbride:2018}.

\subsection{Cubical Type Theory}\label{sec:cubical}

In a rather different line of research, Voevodsky showed that Intensional Type
Theory is compatible with a \emph{univalence axiom} yielding an element of
$\mathsf{Id}_\mathcal{U}(A,B)$ for every equivalence (coherent isomorphism)
between types $A,B$~\Cite{kapulkin-lumsdaine:2016,hottbook}.
A univalent universe \emph{classifies} types under a certain size cut-off in the sense of higher topos theory~\Cite{lurie:2009}. However, Intensional Type Theory extended
with univalence lacks \emph{canonicity}, because identity elimination
computes only on $\mathsf{refl}$ and not on proofs constructed by univalence.

Since then, \emph{cubical type theories} have been developed to validate
univalence without disrupting
canonicity~\Cite{cchm:2017,angiuli-favonia-harper:2018}. These type theories
extend Martin-L\"of's type theory with an abstract interval, maps out of which
represent paths, a higher-dimensional analogue to equality;
the interval has abstract elements, represented by a new
sort of \emph{dimension} variable $i$, and constant endpoints $0,1$. Coercions
arise as an instance of \emph{Kan structure} governed directly by the structure of
paths between types, which are nothing more than types dependent on an
additional dimension variable.

There are currently two major formulations of cubical type theory. De Morgan cubical
type theory \Cite{cchm:2017} equips the interval with negation and binary
connection (minimum and maximum) operations. Cartesian cubical type theory
\Cite{abcfhl:2019,angiuli-favonia-harper:2018}, the closest relative of \XTT{},
has no additional structure on the interval, but equips types with a
much stronger notion of coercion generalizing the one described in
\cref{sec:coercion}.

\subsection{Our contribution: \texorpdfstring{\XTT}{XTT}}

We contribute \XTT{} (\cref{sec:xtt-rules}), a new type theory that supports
extensional equality without equality reflection, using ideas from cubical type theory
\Cite{cchm:2017,abcfhl:2019,angiuli-favonia-harper:2018}. In particular, we
obtain a compositional account of propositional equality satisfying function
extensionality and a \emph{judgmental} version of the unicity of identity
proofs---when $P,Q : \Eq{A}{M}{N}$, we have $P=Q$ judgmentally---enabling us to
substantially simplify our Kan operations (\cref{sec:composition}). Moreover,
\XTT{} is closed under a
cumulative hierarchy\footnote{As in previous work \Cite{sterling:2018:gat}, we employ an
\emph{algebraic} version of cumulativity which does not require subtyping.} of
closed universes \`a la Russell. We hope to integrate \XTT{} into the \redtt{} cubical
proof assistant~\Cite{redtt:2018:dagstuhl} as an implementation
of extensional equality in the style of two-level type
theory~\cite{angiuli-favonia-harper:2018}.

A common thread that runs through the \XTT{} formalism is the decomposition of
constructs from \OTT{} into more modular, \emph{judgmental} principles. For
instance, rather than defining equality separately at every type and entangling
the connectives, we define equality once and for all using the interval.
Likewise, rather than ensuring that equality proofs are unique through brute
force, we obtain unicity using a structural rule which does not mention the
equality type.

By first developing the model theory of \XTT{} in an algebraic way
(\cref{sec:model-theory}), we then prove a canonicity theorem for the
\emph{initial} model of \XTT{} (\cref{sec:canonicity-theorem}): any closed term
of boolean type is equal to either $\true$ or $\false$.  This result is
obtained using a novel extension of the categorical gluing technique described
by \Cite[Coquand and Shulman]{coquand:2018,shulman:2015}, in which one glues
along a cubical nerve functor from \XTT{}'s syntactic category into cubical
sets. We learned after completion of this paper that the idea of proving
canonicity for cubical type theories by gluing along a cubical nerve functor
was circulated informally by Awodey some years prior, and is being
independently developed by Awodey and Fiore.
Canonicity expresses
a form of ``computational adequacy''---in essence, that the equational theory
of \XTT{} suffices to derive any equation which ought to hold by (closed)
computation---and is one of many syntactical considerations that experience has
shown to be correlated to usability.

\section{Programming and proving in \texorpdfstring{\XTT}{XTT}}

Like other cubical type theories, the \XTT{} language extends Martin-L\"of's
type theory with a new sort of variable $i$ ranging over an abstract interval
with global elements $0$ and $1$; we call an element $r$ of the interval a
\emph{dimension}, and we write $\e$ to range over a constant dimension $0$ or
$1$. Cubical type theories like \XTT{} also use a special kind of hypothesis to
constrain the values of dimensions: when $r$ and $s$ are dimensions, then $r=s$
is a \emph{constraint}. In \XTT{}, a single context $\Psi$ accounts
for both dimension variables ($\Psi,i$) and constraints ($\Psi,r=s$). We will
write $\IsDim{r}$ for when a dimension $r$ is valid in a dimension context $\Psi$.
The judgment $\EqDim{r}{s}$ holds when $r$ and $s$ are equal as dimensions with respect
to the constraints in $\Psi$. Dimensions can be substituted for dimension variables,
an operation written $\dsubst{M}{r}{i}$.

Finally, ordinary type-theoretic assumptions $x:A$ are kept in a context
$\Gamma$ that depends on $\Psi$. In \XTT{}, a full context is therefore
written $\Psi\mid\Gamma$. The meaning of a judgment at context
$\parens{\Psi,i=r}$ is completely determined by its instance under the
substitution $r/i$. Under the false constraint $0=1$, all judgments hold; the
resulting collapse of the typing judgment and the judgmental equality does not
disrupt any important metatheoretic properties, because the theory of
dimensions is decidable.

\begin{figure}
  \noindent\makebox[\textwidth]{%
    \begin{minipage}{\paperwidth}
      \begin{grammar}
        cubes & \Psi,\Phi & \cdot \GrmSep \Psi, i \GrmSep \Psi, \xi
        \\
        contexts & \Gamma,\Delta & \cdot \GrmSep \Gamma, x:A
        \\
        dimensions & r, s & i\GrmSep \e
        \\
        constant dims. & \e & 0 \GrmSep 1
        \\
        constraints & \xi & r=r'
        \\
        universe levels & k,l & n\quad (n\in\mathbb{N})
        \\
        types & A,B & M\GrmSep \picl{x}{A}{B}\GrmSep \sigmacl{x}{A}{B} \GrmSep \Eq{i.A}{M}{N} \GrmSep \TyLift{k}{l}{A} \GrmSep \Univ[k] \GrmSep \bool
        \\
        terms & M, N & x \GrmSep A\GrmSep \lam{x}{M} \GrmSep \app[x:A.B]{M}{N} \GrmSep \pair{M}{N} \GrmSep \fst[x:A.B]{M} \GrmSep \snd[x:A.B]{M}\GrmSep\\
        \GrmContinue &
          \lam{i}{M} \GrmSep \papp[i.A]{M}{r}\GrmSep
          \true \GrmSep \false \GrmSep \ifb{x.A}{M}{N_0}{N_1}\GrmSep
          \\
        \GrmContinue &
          \coe{i.A}{r}{r'}{M}\GrmSep
          \hcom{A}{r}{r'}{M}{
            \sys{s}{j.N_0}{j.N_1}
          }
      \end{grammar}
    \end{minipage}
  }
  \caption{A summary of the raw syntax of \XTT. As a matter of top-level
  notation, we freely omit annotations that can be inferred from context, writing
  $\app{M}{N}$ for $\app[x:A.B]{M}{N}$. The annotations chosen in the raw syntax are the minimal ones required to establish a coherent interpretation into the initial \XTT-algebra; for instance, it is unnecessary to include an annotation on the $\lambda$-abstraction.}\label{fig:grammar}
\end{figure}

The general typehood judgment $\IsTy<\Psi>[\Gamma]{A}[k]$ means that $A$ is a
type of universe level $k$ in context $\Gamma$ over the cube $\Psi$; note that
this judgment presupposes the well-formedness of $\Psi,\Gamma$.
Likewise, the element typing judgment $\IsTm<\Psi>[\Gamma]{M}{A}$ means
that $M$ is an element of the type $A$ in $\Gamma$ over $\Psi$ as above; this form of judgment presupposes the
well-formedness of $A$ and thence $\Psi,\Gamma$. We also have typed
judgmental equality $\EqTy<\Psi>[\Gamma]{A}{B}$ and
$\EqTm<\Psi>[\Gamma]{M}{N}{A}$, which presuppose the well-formedness of
\emph{all} their constituents.

\paragraph*{Dependent equality types}

\XTT{} extends Martin-L\"of type theory with \emph{dependent equality types}
$\Eq{i.A}{N_0}{N_1}$ when $\IsTy<\Psi,i>{A}$ and
$\IsTm<\Psi>{N_0}{\dsubst{A}{0}{i}}$ and $\IsTm<\Psi>{N_1}{\dsubst{A}{1}{i}}$.
Geometrically, elements of this type are \emph{lines} or \emph{paths}
in the type $A$ ranging over dimension $i$, with left endpoint $N_0$ and
right endpoint $N_1$.\footnote{Our dependent equality types are
locally the same as dependent path types $\mathsf{Path}_{i.A}({N_0},{N_1})$
from cubical type theories; however, we have arranged in \XTT{} for them to
satisfy a unicity principle by which they earn the name ``equality'' rather than ``path''.}
This type captures internally the
equality of $N_0$ and $N_1$; dependency of $A$ on the dimension $i$ is in essence
a cubical reconstruction of heterogeneous equality, albeit with different
properties from the version invented by McBride in \Cite[his thesis]{mcbride:1999}.

An element of the equality type $\Eq{i.A}{N_0}{N_1}$ is formed by the dimension
$\lambda$-abstraction $\lam{i}{M}$, requiring that $M$ is an element of $A$ in the extended context, and
that $N_0,N_1$ are the left and right sides of $M$ respectively. Proofs $P$ of
equality are eliminated by dimension application, $\papp{P}{r}$, and are subject to
$\beta,\eta,\xi$ rules analogous to those for function types. Finally, we have
$P(\e) = N_\e$ always, extending Gentzen's principle of inversion to the side
condition that we placed on $M$.
More formally:
\begin{mathparpagebreakable}
  \inferrule{
    \IsTm<\Psi,i>{M}{A}
    \\\\
    \etc{
      \EqTm<\Psi,i=\e>{M}{N_\e}{A}
    }
  }{
    \IsTm<\Psi>{\lam{i}{M}}{\Eq{i.A}{N_0}{N_1}}
  }
  \and
  \inferrule{
    \IsDim{r}
    \\\\
    \IsTm{M}{\Eq{i.A}{N_0}{N_1}}
  }{
    \IsTm{M(r)}{\dsubst{A}{r}{i}}
  }
  \and
  \inferrule{
    \IsTm{M}{\Eq{i.A}{N_0}{N_1}}
  }{
    \EqTm{M(\e)}{N_\e}{\dsubst{A}{\e}{i}}
  }
  \and
  \inferrule{
    \IsTm{M}{\Eq{i.A}{N_0}{N_1}}
  }{
    \EqTm{M}{\lam{i}{M(i)}}{\Eq{i.A}{N_0}{N_1}}
  }
  \and
  \inferrule{
    \IsTm<\Psi,i>{M}{A}
  }{
    \EqTm{(\lam{i}{M})(r)}{\dsubst{M}{r}{i}}{\dsubst{A}{r}{i}}
  }
\end{mathparpagebreakable}

\paragraph*{Function extensionality}

A benefit of the cubical formulation of equality types is that the principle of
function extensionality is trivially derivable in a computationally
well-behaved way. Suppose that $f, g : \picl{x}{A}{B}$ and we have a family of
equalities $h : \picl{x}{A}{\Eq{\_. B}{f(x)}{g(x)}}$; then, we obtain a proof
that $f$ equals $g$ by abstraction and application:
\begin{mathpar}
  \lam{i}{\lam{x}{h(x)(i)}} :
  \Eq{\_.\picl{x}{A}{B}}{f}{g}
\end{mathpar}

In semantics of type theory, the structure of equality on a type usually
mirrors the structure of the \emph{elements} of that type in a straightforward
way: for instance, a function of equations is used to equate two functions, and
a pair of equations is used to equate two pairs. The benefit of the cubical
approach is that this observation, at first purely empirical, is systematized
by \emph{defining} equality in every type in terms of the elements of that type
in a context extended by a dimension.

\paragraph*{Judgmental unicity of equality: boundary separation}

In keeping with our desire to provide \emph{convenient} syntax for working with
extensional equality, we want proofs $P,Q:\Eq{i.A}{N_0}{N_1}$ of the same
equation to be judgmentally equal. Rather than adding a rule to that effect,
whose justification in the presence of the elimination rules for equality types
would be unclear, we instead impose a more primitive \emph{boundary separation}
principle at the judgmental level: every term is completely determined by its
boundary.\footnote{We call this principle ``boundary separation'' because it
turns out to be exactly the fact that the collections of types and elements,
when arranged into presheaves on the category of contexts, are \emph{separated}
with respect to a certain coverage on this category. We develop this
perspective in \CiteExtended.}
\begin{mathpar}
  \inferrule{
    \IsDim<\Psi>{r}
    \\
    \etc{
      \EqTm<\Psi,r=\e>{M}{N}{A}
    }
  }{
    \EqTm<\Psi>{M}{N}{A}
  }
\end{mathpar}
In this rule we have abbreviated $\EqTm<\Psi,r=0>{M}{N}{A}$ and $\EqTm<\Psi,r=1>{M}{N}{A}$ as
$\etc{\EqTm<\Psi,r=\e>{M}{N}{A}}$. We shall make use of this notation throughout the paper.

We can now \emph{derive} a rule that (judgmentally) equates all
$P,Q:\Eq{i.A}{N_0}{N_1}$.

\begin{proof}

  If $P,Q : \Eq{i.A}{M}{N}$, then to show that $P=Q$, it suffices to show that
  $\lam{i}{P(i)} = \lam{i}{Q(i)}$; by the congruence rule for equality
  abstraction, it suffices to show that $P(i) = Q(i)$ in the extended context.
  But by boundary separation, we may pivot on the boundary of $i$, and it suffices
  to show that $P(0)=Q(0)$ and $P(1)=Q(1)$. But these are automatic, because
  $P$ and $Q$ are both proofs of $\Eq{i.A}{N_0}{N_1}$, and therefore
  $P(\e)=Q(\e)=N_\e$. \qedhere

\end{proof}

In an unpublished note from 2017, Thierry Coquand
identifies a class of cubical sets equivalent to our separated types, calling
them ``Bishop sets''~\Cite{coquand:2017:bish}.

\subsection{Kan operations: coercion and composition}

How does one \emph{use} a proof of equality? We must have at least a coercion operation
which, given a proof $Q:\Eq{\_.\Univ}{A}{B}$, coherently transforms elements $M:A$ to
elements of $B$.

\subsubsection{Generalized coercion}\label{sec:coercion}

In \XTT{}, coercion and its coherence are obtained as instances of one general
operation: for any two dimensions $r,r'$ and a \emph{line} of types $i.C$, if
$M$ is an element of $\dsubst{C}{r}{i}$, then $\coe{i.C}{r}{r'}{M}$ is an
element of $\dsubst{C}{r'}{i}$.
\begin{mathpar}
  \inferrule{
    \IsDim<\Psi>{r,r'}
    \\
    \IsTy<\Psi,i>{C}
    \\
    \IsTm<\Psi>{M}{\dsubst{C}{r}{i}}
  }{
    \IsTm<\Psi>{\coe{i.C}{r}{r'}{M}}{\dsubst{C}{r'}{i}}
  }
\end{mathpar}

In the case of a proof $Q:\Eq{\_.\Univ}{A}{B}$ of equality between types, we
coerce $M:A$ to the type
$B$ using the instance $\coe{i.Q(i)}{0}{1}{M}$. But how does $M$ relate
to its coercion? \emph{Coherence} of coercion demands their
equality, although such an equation must relate terms of (formally) different
types; this heterogeneous equality is stated in \XTT{} using a
\emph{dependent} equality type $\Eq{i.Q(i)}{M}{\coe{i.Q(i)}{0}{1}{M}}$. To
construct an element of this equality type, we use the same coercion operator
but with a different choice of $r,r'$; we construct this \emph{filler} by
coercing from $0$ to a fresh dimension, obtaining
$\lam{j}{\coe{i.Q(i)}{0}{j}{M}} : \Eq{i.Q(i)}{M}{\coe{i.Q(i)}{0}{1}{M}}$:
\begin{diagram}
  j.Q(j)\quad \ni\qquad
  M
  & \rLine^{\quad j.\coe{i.Q(i)}{0}{j}{M}\quad}
  & \coe{i.Q(i)}{0}{1}{M}
\end{diagram}

To see that the filler $\coe{i.Q(i)}{0}{j}{M}$ has the correct
boundary with respect to $j$, we inspect its instances under the substitutions
$0/j, 1/j$. First, we observe that the right-hand side
$\dsubst{\parens{\coe{i.Q(i)}{0}{j}{M}}}{1}{j}$ is exactly
$\coe{i.Q(i)}{0}{1}{M}$; second, we must see that
$\dsubst{\parens{\coe{i.Q(i)}{0}{j}{M}}}{0}{j}$ is $M$, bringing us to an
important equation that we must impose generally:
\begin{mathpar}
  \inferrule{}{
    \EqTm<\Psi>{\coe{i.C}{r}{r}{M}}{M}{\dsubst{C}{r}{i}}
  }
\end{mathpar}

\paragraph*{How do coercions compute?}

In order to ensure that proofs in \XTT{} can be computed to a canonical form, we
need to explain generalized coercion in each type in terms of the elements of
that type. To warm up, we explain how coercion must compute in a non-dependent
function type:
\begin{mathpar}
  \coe{i.\arr{A}{B}}{r}{r'}{M}
  =
  \lam{x}{
    \coe{i.B}{r}{r'}{
      \parens[\big]{
        \app{M}{
          \coe{i.A}{r'}{r}{x}
        }
      }
    }
  }
\end{mathpar}

That is, we abstract a variable $x:\dsubst{A}{r'}{i}$ and need to
obtain an element of type $\dsubst{B}{r'}{i}$. By \emph{reverse} coercion, we
obtain $\coe{i.A}{r'}{r}{x}:\dsubst{A}{r}{i}$; by applying $M$ to this, we
obtain an element of type $\dsubst{B}{r}{i}$. Finally, we coerce from $r$
to $r'$. The version for dependent function types is not much harder, but
requires a filler:
\begin{mathpar}
  \inferrule{
    \widetilde{x}\triangleq
    \lam{j}{
      \coe{i.A}{r'}{j}{x}
    }
  }{
    \coe{i.\picl{x}{A}{B}}{r}{r'}{M}
    =
    \lam{x}{
      \coe{
        i. \subst{B}{\widetilde{x}(i)}{x}
      }{r}{r'}{
        \app{M}{
          \widetilde{x}(r)
        }
      }
    }
  }
\end{mathpar}

The case for dependent pair types is similar, but without the contravariance:
\begin{mathpar}
  \inferrule{
    \widetilde{M_0}\triangleq
    \lam{j}{\coe{i.A}{r}{j}{\fst{M}}}
  }{
    \coe{i.\sigmacl{x}{A}{B}}{r}{r'}{M}
    =
    \pair{
      \widetilde{M_0}(r')
    }{
      \coe{i.
        \subst{B}{
          \widetilde{M_0}(i)
        }{x}
      }{r}{r'}{\snd{M}}
    }
  }
\end{mathpar}

Coercions for base types (like $\bool$) are uniformly determined by
\emph{regularity}, a rule of \XTT{} stating that if $A$ is a type which doesn't
vary in the dimension $i$, then $\coe{i.A}{r}{r'}{M}$ is just $M$. Regularity
makes type sense because $\dsubst{A}{r}{i} = A = \dsubst{A}{r'}{i}$;
semantically, it is more difficult to justify in the presence of standard
universes, and is not known to be compatible with principles like
univalence.\footnote{Regularity is proved by Swan to be incompatible with
univalent universes assuming that certain standard techniques are used~\Cite{swan:2018}; however,
it is still possible that there is a different way to model univalent universes
with regularity. Awodey constructs a model of intensional type
theory \emph{without} universes in regular Kan cubical sets~\Cite{awodey:2018}, using the term \emph{normality} for what we have called regularity.}
But \XTT{} is specifically designed to provide a theory of \emph{equality}
rather than \emph{paths}, so we do not expect or desire to justify univalence
at this level.\footnote{Indeed, unicity of identity proofs is also
incompatible with univalence. \XTT{} is, however, compatible with a formulation
in which it is just one level of a two-level type theory, along the lines of
Voevodsky's Homotopy Type System, in which the other level would have a
univalent notion of path that coexists in harmony with our notion of
equality~\Cite{voevodsky:2013:hts,angiuli-favonia-harper:2018}.}

The only difficult case is to define coercion for equality types; at first,
we might try to define $\coe{i.\Eq{j.A}{N_0}{N_1}}{r}{r'}{P}$ as
$\lam{j}{\coe{i.A}{r}{r'}{P(j)}}$, but this does not make
type-sense: we need to see that
$\dsubst{\parens[\big]{\coe{i.A}{r}{r'}{P(j)}}}{\e}{j} = N_\e$, but we
only obtain $\dsubst{\parens[\big]{\coe{i.A}{r}{r'}{P(j)}}}{\e}{j} =
\coe{i.A}{r}{r'}{N_\e}$, which is ``off by'' a coercion. Intuitively, we can
solve this problem by specifying what values a coercion takes under certain
substitutions: in this case, $N_0$ under $0/j$, and $N_1$ under $1/j$. We call
the resulting operation \emph{generalized composition}.

\subsubsection{Generalized composition}\label{sec:composition}

For any dimensions $r,r',s$ and a line of types $i.C$, if $M$ is an element of
$\dsubst{C}{r}{i}$ and $i.N_0, i.N_1$ are lines of elements of $C$ defined
respectively on the subcubes $\parens{s=0},\parens{s=1}$ such that
$\dsubst{N_{\e}}{r}{i} = M$, then $\com{i.C}{r}{r'}{M}{\sys{s}{i.N_0}{i.N_1}}$ is an
element of $\dsubst{C}{r'}{i}$. This is called the \emph{composite} of $M$ with
$N_0,N_1$ from $r$ to $r'$, schematically abbreviated
$\com{i.C}{r}{r'}{M}{\etcsys[\e]{s}{i.N_\e}}$.
As with coercion, when $r=r'$, we have $\com{i.C}{r}{r'}{M}{\etcsys{s}{i.N_\e}}
= M$, and moreover, if $s=\e$, we have $\com{i.C}{r}{r'}{M}{\etcsys{s}{i.N_\e}}
= \dsubst{N_\e}{r'}{i}$.

Returning to coercion for equality types, we now have exactly what we need:
\begin{mathpar}
  \coe{i.\Eq{j.C}{N_0}{N_1}}{r}{r'}{P}
  =
  \lam{j}{(
    \com{i.C}{r}{r'}{P(j)}{
      \etcsys[\e]{j}{\_.N_\e}
    })
  }
\end{mathpar}


Next we must explain how the generalized composition operation
computes at each type; in previous works~\Cite{angiuli-favonia-harper:2018}, we have
seen that it is simpler to instead \emph{define} generalized composition in terms of a simpler
\emph{homogeneous} version, in which one composes in a type $C$ rather than a line of types $i.C$;
we write $\hcom{C}{r}{r'}{M}{\etcsys{s}{j.N_\e}}$ for this homogeneous composition, defining the
generalized composition in terms of it as follows:
\begin{mathpar}
  \com{i.C}{r}{r'}{M}{\etcsys[\e]{s}{i.N_\e}} =
  \hcom{
    \dsubst{C}{r'}{i}
  }{r}{r'}{
    \parens{\coe{i.C}{r}{r'}{M}}
  }{
    \etcsys[\e]{s}{
      i.\coe{i.C}{i}{r'}{N_\e}
    }
  }
\end{mathpar}

Surprisingly, in \XTT{} we do not need to build in any computation rules for
homogeneous composition, because they are completely
determined by judgmental boundary separation. For instance, we can derive a
computation rule already for homogeneous composition in the dependent function
type, by observing that the equands have the same boundary with respect to the
dimension $j$:
\begin{mathpar}
  \hcom{\picl{x}{A}{B}}{r}{r'}{M}{\etcsys{j}{i.N_\e}} =
  \lam{x}{
    \hcom{B}{r}{r'}{M(x)}{
      \etcsys{j}{i.N_\e}
    }
  }
\end{mathpar}

From homogeneous composition, we obtain \emph{symmetry and transitivity} for the
equality types. Given $P:\Eq{_.A}{M}{N}$, we obtain an element of type
$\Eq{\_.A}{N}{M}$ as follows:
\begin{mathpar}
  \lam{i}{\hcom{A}{0}{1}{P(0)}{\sys{i}{j.P(j)}{\_.P(0)}}}
\end{mathpar}
Furthermore, given $Q:\Eq{\_.A}{N}{O}$, we obtain an element of type
$\Eq{\_.A}{M}{O}$ as follows:
\begin{mathpar}
  \lam{i}{
    \hcom{A}{0}{1}{P(i)}{
      \sys{i}{\_.P(0)}{j.Q(j)}
    }
  }
\end{mathpar}

\begin{example}[Identity type]
  It is possible to \emph{define} Martin-L\"of's identity type and its eliminator,
  albeit with a much stronger computation rule than is customary.
  \begin{mathpar}
    \mathsf{Id}_A(M,N) \triangleq \Eq{\_.A}{M}{N}
    \and
    \mathsf{refl}_A(M) \triangleq \lam{\_}{M}
    \and
    \inferrule{
      \widetilde{P}\triangleq
      \lam{j}{\parens{
        \hcom{A}{0}{j}{P(0)}{
          \sys{i}{\_.P(0)}{k.P(k)}
        }
      }}
    }{
      \mathsf{J}_{x,y,p.C(x,y,p)}(P;x.Q(x)) \triangleq
      \coe{
        i.
        C\parens{
          P(0),P(i),\widetilde{P}
        }
      }{0}{1}{Q(P(0))}
    }
  \end{mathpar}
  This particular definition of $\mathsf{J}$ relies on \XTT{}'s \emph{boundary
  separation} rule, but one could instead define it in a more complicated way
  without boundary separation. However, that this construction of the identity
  type models the computation rule relies crucially on \emph{regularity}, which
  does not hold in other cubical type theories whose path types validate
  univalence. In the absence of regularity, one can define an operator with the
  same type as $\mathsf{J}$ but which satisfies its computation rule only up to
  a path.
\end{example}

\subsection{Closed universes and type-case}\label{sec:type-case}

In \cref{sec:coercion}, we showed how to calculate coercions $\coe{i.C}{r}{r'}{M}$
in each type former $C$. In previous cubical type theories
\Cite{cchm:2017,angiuli-favonia-harper:2018}, one could ``uncover'' all the
things that a coercion must be equal to by reducing according to the rules
which inspect the interior of the type line $i.C$. While this strategy can be
used to establish canonicity for closed terms, it fails to uncover certain
reductions for open terms, a prerequisite for algorithmic type checking.

Specifically, given a variable $q:\Eq{\_.\Univ}{A_0\to{}B_0}{A_1\to{}B_1}$, the
coercion $\coe{i.q(i)}{r}{r'}{M}$ is \emph{not} necessarily stuck, unlike in
other cubical type theories. Suppose that we can find further proofs
$Q_A:\Eq{\_.\Univ}{A_0}{A_1}$ and $Q_B:\Eq{\_.\Univ}{B_0}{B_1}$; in this case,
$\lam{i}{Q_A(i)\to Q_B(i)}$ is \emph{also} a proof of
$\Eq{\_.\Univ}{A_0\to{}B_0}{A_1\to{}B_1}$, so by boundary separation it must be
equal to $q$, and therefore $\coe{i.q(i)}{r}{r'}{M}$ must be equal to
$\coe{i.Q_A(i)\to Q_B(i)}{r}{r'}{M}$. But the type-directed reduction rule for
coercion applies only to the latter! Generally, to see how to reduce the first
coercion, it seems that we need to be able to ``dream up'' proofs $Q_A,Q_B$ out
of thin air, or determine that they can't exist, an impossible task.

In \XTT{}, we cut this Gordian knot by ensuring that $Q_A,Q_B$ \emph{always}
exist, following the approach employed in \OTT{}.  To invert the equation $q$
into $Q_A$ and $Q_B$, we add an intensional \emph{type-case} operator to
\XTT{}, committing to a closed and inductive notion of universe by allowing
pattern-matching on types~\cite{nordstrom-peterson-smith:1990}. It is also
possible to extend \XTT{} with open and/or univalent universes which themselves
lack boundary separation, as in two-level type theories.

For illustrative purposes, consider coercion along an equality between dependent
function types. Given $q:\Eq{\_.\Univ}{\picl{x}{A_0}{B_0}}{\picl{x}{A_1}{B_1}}$,
we define by type-case the following:
\begin{align*}
  Q_A &\triangleq
  \lam{i}{
    \UCase{q(i)}{
      \UCaseBranch{\Pi_A B}{A}
      \mid
      \UCaseBranch{\_}{\bool}
    }
  }
  : \Eq{\_.\Univ}{A_0}{A_1}
  \\
  Q_B &\triangleq
  \lam{i}{
    \UCase{q(i)}{
      \UCaseBranch{\Pi_A B}{B}
      \mid
      \UCaseBranch{\_}{\lam{\_}{\bool}}
    }
  }
  : \Eq{i.Q_A(i)\to\Univ}{\lam{x}{B_0}}{\lam{x}{B_1}}
\end{align*}
Because of $q$'s boundary, we are concerned only with the $\Pi$ branch of the
above expressions, and are free to emit a ``dummy'' answer in other
branches.
With $Q_A,Q_B$ in hand, we note that $q(i) = \picl{x}{Q_A(i)}{Q_B(i)(x)}$ using boundary separation;
therefore, we are free to calculate $\coe{i.q(i)}{r}{r'}{M}$ as follows:
\begin{mathpar}
  \inferrule{
    \widetilde{x} \triangleq \lam{j}{\coe{i.Q_A(i)}{r'}{j}{x}}
  }{
    \coe{i.q(i)}{r}{r'}{M} =
    \lam{x}{
      \coe{
        i. \app{Q_B(i)}{\widetilde{x}(i)}
      }{r}{r'}{
        \app{M}{
          \widetilde{x}(r)
        }
      }
    }
  }
\end{mathpar}

This \emph{lazy} style of computing with proofs of equality means, in
particular, that coercing along an equation cannot tell the difference between
a postulated axiom and a canonical proof of equality, making \XTT{} compatible
with extension by consistent equational axioms.

\begin{remark}

  One might wonder whether it is possible to tame the use of type-case above to
  something compatible with a \emph{parametric} understanding of types, in which
  (as in \OTT) one cannot branch on whether or not $C$ is a function type or a
  pair type, etc. It is likely that this can be done, but we stress that the
  fundamental difficulty is not resolved: whether or not we allow general
  type-case, we have not escaped the need for type constructors to be disjoint
  and injective, which contradicts the role of universes in mathematics as
  (weak) classifiers of small families. Future work on \XTT{} and its successors
  must focus on resolving this issue, quite apart from any considerations of
  parametricity.

\end{remark}

\subsection{Future extensions}

\paragraph*{Universe of propositions}
\XTT{} currently lacks one of the hallmarks of \OTT{}, an extensional universe
of proof-irrelevant propositions. In future work, we intend to extend \XTT{}
with a reflective subuniverse $\Omega$ of propositions closed under equality and
universal and existential quantification over arbitrary types, satisfying:

\begin{itemize}
  \item \emph{Proof irrelevance.} For each proposition $\IsTm{p}{\Omega}$, we have $\EqTm{M}{N}{p}$ for all $\IsTm{M,N}{p}$.
  \item \emph{Extensionality (univalence).} For all $\IsTm{p,q}{\Omega}$, we have an element of $\Eq{\_.\Omega}{p}{q}$ whenever there are functions $p\to q$ and $q\to p$.
\end{itemize}

The \emph{reflection} of the propositional subuniverse will take a type
$\IsTy{A}$ to a proof-irrelevant proposition $\IsTm{\Verts{A}}{\Omega}$, acting
as a strict truncation or squash
type~\cite{constable:1986,pfenning:2001,awodey-bauer:2004}. The addition of
$\Omega$ will allow \XTT{} to be used as a syntax for topos-theoretic
constructions, with $\Omega$ playing the role of the subobject classifier.

\paragraph*{(Indexed) Quotient Inductive Types}

Another natural extension of \XTT{} is the addition of \emph{quotient types};
already considered as an extension to \OTT{} by the Epigram
Team~\cite{brady-chapman-ped-gundry-mcbride-morris-norell-oury:2011} and more
recently by Atkey~\cite{atkey:sott}, quotient types are essential when using
type theory for either programming or mathematics. One of the ideas of
Homotopy Type Theory and cubical type theories in particular is to reconstruct
the notion of quotienting by an equivalence relation as a special case of
\emph{higher inductive type} (HITs), a generalization of ordinary inductive types
which allows constructors to target higher dimensions with a specified partial
boundary. When working purely at the level of sets, as in \XTT{}, these higher
inductive types are called \emph{quotient inductive types} (QITs)~\cite{altenkirch-kaposi:2016}.

We intend to adapt the work of Cavallo and Harper~\cite{cavallo-harper:2019} to
a general schema for \emph{indexed quotient inductive types} as an extension of
\XTT{}. The resulting system would support ordinary quotients by equivalence
relations \emph{en passant}, and when these equivalence relations are valued in
$\Omega$, one can show that they are effective. Quotient inductive types also
enable the construction of free algebras for infinitary algebraic theories,
usually obtained in classical set theory from the non-constructive axiom of
choice~\cite{blass:1983,lumsdaine-shulman:2017}.
Another
application of quotient inductive types is the definition of a
\emph{localization} functor with respect to a class of maps, enabling users of
the extended \XTT{} to work internally with sheaf subtoposes.

The extension of \XTT{} with quotient inductive types means that we must
account for \emph{formal homogeneous composites} in QITs which are
canonical forms~\cite{cavallo-harper:2019,chm:2018}. Ordinarily, this
introduces a severe complicating factor to a canonicity proof, because the
notion of canonical form ceases to be stable under all dimension
substitutions~\cite{angiuli-favonia-harper:2018,huber:2018}, but we expect the
\emph{proof-relevant} cubical logical families technique that we introduce in
\cref{sec:model-theory} to scale directly to the case of quotient inductive
types without significant change, in contrast with classical approaches based
on partial equivalence relations.

\section{Algebraic model theory and canonicity}\label{sec:model-theory}

We have been careful to formulate the \XTT{} language in a \emph{(generalized)
algebraic} way, obtaining automatically a category of algebras and homomorphisms
which is equipped with an initial
object~\Cite{cartmell:1978,cartmell:1986,kaposi-kovacs-altenkirch:2019}. That this initial object
is isomorphic to the model of \XTT{} obtained by constraining and quotienting
its raw syntax under judgmental equality (i.e. the Lindenbaum--Tarski algebra)
is an instance of Voevodsky's famous Initiality Conjecture~\cite{voevodsky:2016:templeton}, and we do not
attempt to prove it here; we merely observe that this result has been established
for several simpler type
theories~\Cite{streicher:1991,castellan-clairambault-dybjer:2017}.

Working within the category of \XTT-algebras enables us to formulate and prove
results like canonicity and normalization for the \emph{initial} \XTT-algebra
in an economical manner, avoiding the usual bureaucratic overhead of reduction
relations and partial equivalence relations, which were the state of the art
for type-theoretic metatheory prior to the work of Shulman~\Cite{shulman:2015},
Altenkirch and Kaposi~\Cite{altenkirch-kaposi:2016:nbe}, and
Coquand~\Cite{coquand:2018}.

Because our algebraic techniques involve defining families over \emph{only}
well-typed terms already quotiented by judgmental equality,
we avoid many of the technical difficulties arising from working
with the raw terms of cubical type theories, including the closure under
``coherent expansion'' which is critical to earlier cubical metatheories
\Cite{angiuli-favonia-harper:2018, huber:2018}. Our abstract gluing-based
approach therefore represents a methodological advance in metatheory for cubical type
theories.

\begin{theorem}[Canonicity]\label{thm:canonicity}
  In the initial \XTT-algebra, if $\IsTm<\cdot>[\cdot]{M}{\bool}$, then either
  $\EqTm<\cdot>[\cdot]{M}{\true}{\bool}$ or $\EqTm<\cdot>[\cdot]{M}{\false}{\bool}$.
\end{theorem}

Following previous work \Cite{sterling:2018:gat}, we employ for our semantics a
variant of \emph{categories with families} (cwf) \Cite{dybjer:1996} which
supports a predicative hierarchy of universes \`a la Russell. A cwf in our
sense begins with a \emph{category of contexts} $\Cx$, and a presheaf of types
$\CwfTy : \Psh{\Cx\times\Lev}$; here $\Lev$ is the category of \emph{universe
levels}, with objects the natural numbers and unique arrows $l \rTo k$ if and
only if $k \leq l$.\footnote{Observe that $\Lev = \OpCat{\omega}$; reversing
arrows allows us to move types from smaller universes to larger ones.}
The fiber of the presheaf of types $\CwfTy : \Psh{\Cx\times\Lev}$ at
$(\Gamma,k)$ is written $\CwfTy[k]{\Gamma}$, and contains the types in context
$\Gamma$ of universe level $k$. Reindexing implements simultaneous substitution
$\ReIx{\gamma}{A}$ and universe level shifting $\TyLift{k}{l}{A}$.
In our metatheory, we assume the Grothendieck Universe Axiom, and consequently obtain a
transfinite ordinal-indexed hierarchy of meta-level universes
$\SemUniv[k]$. We impose the requirement that each collection of types
$\CwfTy[k]{\Gamma}$ is $k$-small, i.e.\ $\CwfTy[k]{\Gamma} \in \SemUniv[k]$.

Next, we require a \emph{dependent} presheaf of elements $\CwfEl:\Psh{\int
\CwfTy}$, whose fibers $\CwfEl((\Gamma,k),A)$ we write $\CwfEl{\Gamma}{A}$; to
interpret the actions of level lifting on terms properly, we require the
functorial actions
$\CwfEl{\Gamma}{A}\rTo\CwfEl{\Gamma}{\TyLift{k}{l}{A}}$ to be identities,
strictly equating the fibers $\CwfEl{\Gamma}{A}$ and
$\CwfEl{\Gamma}{\TyLift{k}{l}{A}}$.
The remaining data of a basic cwf is a \emph{context comprehension}, which for
every context $\Gamma$ and type $A\in\CwfTy[k]{\Gamma}$ determines an extended
context $\Gamma.A$ with a weakening substitution $\Gamma.A\rTo^{\Proj}\Gamma$
and a variable term $\Var\in\CwfEl{\Gamma.A}{\ReIx{\Proj}{A}}$.

Next, we specify what further structure is required to make such a cwf into an
\XTT-algebra. To represent contexts $\Psi$ semantically, we use the
\emph{augmented Cartesian cube category} $\AugCube$, which adjoins to the Cartesian
cube category $\Cube$ an initial object; from this, we obtain equalizers $0=1$
in addition to the equalizers $i=r$ which exist in $\Cube$. We then require a
\emph{split fibration} $\Cx\rProjto^{\CubeFib}\AugCube$ with a terminal object,
which implements the dependency of contexts $\Gamma$ on cubes $\Psi$ and forces
appropriate dimension restrictions to exist for contexts, types and elements.
The split fibration induces all the structure necessary to implement dimension
operations; we refer the reader to \CiteExtended{} for details. In the following
discussion, we limit ourselves to a few simpler consequences.
First, we can apply dimension substitutions in terms and types, writing
$\ReIxCleave{\psi}[\Gamma]{A}$ to apply $\psi$ in a type $A$ in context $\Gamma$. We can
also apply dimension substitutions to contexts, written $\ReIx{\psi}{\Gamma}$.
We write $\Wk{\imath}$ for the dimension substitution which weakens by a
dimension variable $i$.
Finally, we write $\CwfDim{\Gamma}$ for the set of valid dimensions
expressions generated from $\CubeFib(\Gamma)$.

\begin{requirement*}[Boundary separation in models]
  In order to enforce boundary separation in \XTT-algebras we require that types and elements over
  them satisfy a separation property. In \CiteExtended{} we phrase the full
  condition as a separation requirement with respect to a particular
  Grothendieck topology on the category of contexts. A specific
  consequence is the familiar boundary separation principle for types:
  given two types $A, B \in \CwfTy{\Gamma}$ and a dimension $i\in\CubeFib(\Gamma)$, if
  $\ReIxCleave*{\epsilon/i}{A} = \ReIxCleave*{\epsilon/i}{B}$ for each $\epsilon \in \braces{0,1}$ then $A = B$.
\end{requirement*}

\begin{requirement*}[Coercion in models]
  An \XTT-algebra must also come with a \emph{coercion structure}, specifying how generalized
  coercion is interpreted in each type. For every type
  $A\in\CwfTy[n]{\ReIx{\Wk{\imath}}{\Gamma}}$ over $\Psi,i$, dimensions $r,r'\in\CwfDim{\Gamma}$,
  and element
  $M\in\CwfEl{\Gamma}{\ReIxCleave*{r/i}[\ReIx{\Wk{\imath}}{\Gamma}]{A}}$, we require an element
  $\CwfCoe{i.A}{r}{r'}{M}\in\CwfEl{\Gamma}{\ReIxCleave*{r'/i}[\ReIx{\Wk{\imath}}{\Gamma}]{A}}$
  with the following properties (in addition to naturality requirements):
  \begin{itemize}
  \item \emph{Adjacency.} If $r=r'$ then $\CwfCoe{i.A}{r}{r'}{M} = M$.
  \item \emph{Regularity.} If $A=\ReIxCleave{\Wk{\imath}}[\Gamma]{A'}$ for some
    $A'\in\CwfTy[n]{\Gamma}$, then $\CwfCoe{i.A}{r}{r'}{M}=M$.
  \end{itemize}
  Additional equations in later requirements specify that generalized coercion computes properly in each connective.
  Similarly, a model must be equipped with a \emph{composition structure} which specifies the
  interpretation of the composition operator.
\end{requirement*}

Finally, we specify algebraically the data with which such a cwf must be
equipped in order to model all the connectives of \XTT{} (again, details are contained in \CiteExtended{}); to distinguish the abstract (De Bruijn) syntax of
the cwf from the raw syntax of \XTT{} we use boldface, writing $\CwfPi{A}{B}$,
$\CwfPathApp[i.A]{M}{r}$ and $\CwfUniv[k]$ to correspond to $\picl{x}{A}{B}$,
$\app[i.A]{M}{r}$ and $\Univ[k]$ respectively, etc. We take a moment to specify how some of
the primitives of \XTT{} are translated into requirements on a model.

\begin{requirement*}[Dependent equality types in models]
  An \XTT-algebra must model \emph{dependent equality types}, which is to say that the following structure is exhibited:
  \begin{itemize}

    \item \emph{Formation.} For each type $A\in\CwfTy[n]{\ReIx{\Wk{\imath}}{\Gamma}}$ and elements
      $\etc{N_\e\in\CwfEl{\Gamma}{\ReIxCleave*{\e/i}{A}}}$, a type
      $\CwfEq{i.A}{N_0}{N_1}\in\CwfTy[n]{\Gamma}$.

    \item \emph{Introduction.} For each $M\in\CwfEl{\ReIx{\Wk{\imath}}{\Gamma}}{A}$, an element
      $\CwfPathLam[i.A]{M}\in\CwfEl{\Gamma}{\CwfEq{i.A}{\ReIxCleave*{0/i}{M}}{\ReIxCleave*{1/i}{M}}}$.

    \item \emph{Elimination.} For each
      $M\in\CwfEl{\Gamma}{\CwfEq{i.A}{N_0}{N_1}}$ and $r\in\CwfDim{\Gamma}$, an
      element $\CwfPathApp[i.A]{M}{r}\in\CwfEl{\Gamma}{\ReIxCleave*{r/i}{A}}$ satisfying
      the equations $\etc{\CwfPathApp[i.A]{M}{\epsilon}=N_\epsilon}$.

    \item \emph{Computation.} For $M\in\CwfEl{\ReIx{\Wk{\imath}}{\Gamma}}{A}$ and
      $r\in\CwfDim{\Gamma}$, the equation:
      \begin{mathpar}
        \CwfPathApp[i.A]{\CwfPathLam[i.A]{i.M}}{r} = \ReIxCleave*{r/i}{M}
      \end{mathpar}

    \item \emph{Unicity.} For $M\in\CwfEl{\Gamma}{\CwfEq{i.A}{N_0}{N_1}}$,
      $M=\CwfPathLam[i.A]{j.\CwfPathApp[i.\ReIxCleave{\Wk{\jmath}}{A}]{\ReIxCleave{\Wk{\jmath}}{M}}{j}}$.

    \item \emph{Level restriction.} The following equations:
      \begin{mathpar}
        \TyLift{k}{l}{\CwfEq{i.A}{N_0}{N_1}} = \CwfEq{i.\TyLift{k}{l}{A}}{N_0}{N_1}
        \and
        \CwfPathLam[i.\TyLift{k}{l}{A}]{M}{r} = \CwfPathLam[i.A]{M}{r}
        \and
        \CwfPathApp[i.\TyLift{k}{l}{A}]{M}{r} = \CwfPathApp[i.A]{M}{r}
      \end{mathpar}

    \item \emph{Naturality.} For $\Delta\rTo^\gamma\Gamma$, the following naturality equations:
      \begin{mathpar}
        \ReIx{\gamma}{\CwfEq{i.A}{N_0}{N_1}} = \CwfEq{i.\ReIx*{\Lift{\Wk{\imath}}{\gamma}}{A}}{\ReIx{\gamma}{N_0}}{\ReIx{\gamma}{N_1}}
        \and
        \ReIx{\gamma}{\CwfPathLam[i.A]{i.M}} = \CwfPathLam[i.\ReIx*{\Lift{\Wk{\imath}}{\gamma}}{A}]{i.\ReIx*{\Lift{\Wk{\imath}}{\gamma}}{M}}
        \and
        \ReIx{\gamma}{\CwfPathApp[i.A]{M}{r}} = \CwfPathApp[i.\ReIx*{\Lift{\Wk{\imath}}{\gamma}}{A}]{\ReIx{\gamma}{M}}{\ReIx{\gamma}{r}}
      \end{mathpar}

    \item \emph{Coercion.} When $\Gamma \rProjMapsto^\CubeFib \Psi,j$ and
      $M\in\CwfEl{\ReIx*{r/j}{\Gamma}}{\ReIxCleave*{r/j}{\CwfEq{i.A}{N_0}{N_1}}}$
      where $\IsDim<\Psi>{r,r'}$, we require that
      $\CwfCoe{j.\CwfEq{i.A}{N_0}{N_1}}{r}{r'}{M}$ equals the following abstraction:
      \[
        \CwfPathLam[
          i.\ReIxCleave*{r'/j}{A}
        ]{
          i.
          \CwfCom{j.A}{r}{r'}{
            \CwfPathApp[i.\ReIxCleave*{r/j}{A}]{
              \ReIxCleave{\Wk{\imath}}{M}
            }{i}
          }{
            \etcsys[\e]{i}{
              j.\ReIxCleave{\Wk{\jmath}}{N_\e}
            }
          }
        }
      \]
  \end{itemize}

\end{requirement*}

Any model of extensional type theory can be used to construct a model of
\XTT{}, so long as it is equipped with a cumulative, inductively defined
hierarchy of universes closed under dependent function types, dependent pair
types, extensional equality types and booleans. (Meaning explanations in the
style of Martin-L\"of \Cite{martin-lof:1984} are one such model.)
The interpretation of \XTT{} into extensional models involves erasing
dimensions, coercions, and compositions; the only subtlety, easily managed, is
to ensure that all judgments under absurd constraints hold.

\subsection{The cubical logical families construction}

Any \XTT-algebra $\Cx$ extends to a category $\CPred$ of
\emph{proof-relevant logical predicates}, which we call \emph{logical
families} by analogy. The proof-relevant character of the construction enables a simpler
proof of canonicity than is obtained with proof-irrelevant techniques, such as
partial equivalence relations. Logical families are a type-theoretic version of
the \emph{categorical gluing} construction, in which a very rich semantic
category (such as sets) is cut down to include just the morphisms which track
definable morphisms in $\Cx$;\footnote{The gluing construction is similar to
realizability; the main difference is that in gluing, one considers collections
of ``realizers'' which are \emph{not} all drawn from a single computational
domain.} one then uses the rich structure of the semantic category to obtain
metatheoretic results about syntax (choosing $\Cx$ to be the initial model)
without considering raw terms at any point in the process.

Usually, to prove canonicity one glues the initial model $\Cx$ together with
$\SET$ along the global sections functor; this equips each context $\Gamma$
with a family of sets $\Pred{\Gamma}$ indexed in the \emph{closing
substitutions} for $\Gamma$. In order to prove canonicity for a cubical language like \XTT, we will
need a more sophisticated version of this construction, in which the global
sections functor is replaced with something that determines substitutions which
are closed with respect to term variables, but open with respect to dimension
variables.

The split fibration $\Cx\rProjto^{\CubeFib}\AugCube$ induces a functor
$\AugCube\rTo^{\CubeCx}\Cx$ which takes every cube $\Psi$ to the empty variable
context over $\Psi$.
This functor in turn induces a \emph{nerve} construction
$\Cx\rTo^{\Nerve}\AugCSet$, taking $\Gamma$ to the cubical set
$\Hom[\Cx]{\CubeCx{-}}{\Gamma}$.\footnote{This construction is also called the
\emph{relative hom functor} by \Cite[Fiore]{fiore:2002}; its use in logic
originates in the study of definability for $\lambda$-calculus, characterizing
the domains of discourse for Kripke logical predicates of varying
arity~\Cite{jung-tiuryn:1993}. We learned the connection to the abstract nerve
construction in conversations with M.\ Fiore about his unpublished joint work with
S.\ Awodey.} Intuitively, this is the presheaf of substitutions which are
closed with respect to term variables, but open with respect to dimension
variables; when wearing $\AugCSet$-tinted glasses, these appear to be the
closed substitutions.

This nerve construction extends to the presheaves of types and elements;
we define the fiber of $\NerveTy[k]:\AugCSet$ at $\Psi$ to be the set
$\CwfTy[k]{\CubeCx{\Psi}}$; likewise, we define the fiber of
$\NerveEl[k]:\Psh{\int\NerveTy[k]}$ at $(\Psi, A)$ to be the set
$\CwfEl{\CubeCx{\Psi}}{A}$. Internally to $\AugCSet$, we regard $\NerveEl[k]$
as a dependent type over $\NerveTy[k]$. We will then (abusively) write
$\NerveEl{A}$ for the fiber of $\NerveEl[k]$ determined by $A:\NerveTy[k]$.

\paragraph*{Category of cubical logical families} Gluing $\Cx$ together with
$\AugCSet$ along $\Nerve$ gives us a category of \emph{cubical logical
families} $\CPred$ whose objects are pairs
$\Gl{\Gamma}=(\Gamma,\Pred{\Gamma})$, with $\Gamma:\Cx$ and $\Pred{\Gamma}$ a
\emph{dependent cubical set} over the cubical set $\Nerve{\Gamma}$. In other
words, $\Pred{\Gamma}$ is a ``Kripke logical family'' on the
substitutions $\CubeCx{\Psi}\rTo\Gamma$ which commutes with dimension
substitutions $\Psi'\rTo\Psi$. A morphism $\Gl{\Delta}\rTo\Gl{\Gamma}$ is a
substitution $\Delta\rTo^\gamma\Gamma$ together with a proof that $\gamma$
preserves the logical family: that is, a closed element $\Pred{\gamma}$ of the
type $\prod_{\delta:\Nerve{\Delta}} \Pred{\Delta}\parens{\delta} \to
\Pred{\Gamma}\parens{\ReIx{\gamma}{\delta}}$ in the internal type theory of
$\AugCSet$. We write $\Gl{\gamma}$ for the pair $(\gamma,\Pred{\gamma})$. We
have a fibration $\CPred\rProjto^{\PiSyn}\Cx$ which merely projects $\Gamma$
from $\Gl{\Gamma}=(\Gamma,\Pred{\Gamma})$.

\paragraph*{Glued type structure}

Recall from \cref{sec:type-case} that we must model \emph{closed} universes.
Therefore, the standard presheaf universes which lift $\SemUniv[k]$ to (weakly) classify all
$k$-small presheaves are insufficient in our case; instead, we must equip each type
with a \emph{code} so that type-case is definable. Accordingly, we define for each
$n\in\mathbb{N}$ an inductive cubical set $\Pred{\ClUni{n}}A:\SemUniv[n+1]$ indexed
over $A:\NerveTy[n]$; internally to $\AugCSet$, the cubical set
$\Pred{\ClUni{n}}{A}$ is the collection of \emph{realizers} for the $\Cx$-type
$A$. An imprecise but helpful analogy is to think of a realizer
$\FmtCode{A}:\Pred{\ClUni{n}}A$ as something like a whnf of $A$, with the
caveat that $\FmtCode{A}$ is an element of this inductively defined set, not a $\Cx$-type.
Simultaneously, for each $\FmtCode{A}:\Pred{\ClUni{n}}A$, we define a cubical family
$\UPred{\FmtCode{A}}:\NerveEl{A}\to\SemUniv[n]$ of realizers of elements of $A$,
with each $\UPred{\FmtCode{A}}$ being the \emph{logical family} of the
$\Cx$-type $A$; finally, we also define realizers for coercion and composition
by recursion on the realizers for types.\footnote{It is important to note that
we do \emph{not} use large induction-recursion in $\AugCSet$ (to our knowledge,
the construction of inductive-recursive definitions has not yet been lifted to
presheaf toposes); instead, we model $n$ object universes using the
meta-universe $\SemUniv[n+1]$. This is an instance of \emph{small
induction-recursion}, which can be translated into indexed inductive
definitions which exist in every presheaf
topos~\Cite{hancock-ghani-malatesta-altenkirch:2013,moerdijk:2000}.} A fragment
of this definition is summarized in \cref{fig:type-codes}. In the definition of
$\UPred{\FmtCode{A}}$ we freely make use of the internal type theory of
$\AugCSet$. This not only exposes the underlying \emph{logical relations} flavor of
these definitions but simplifies a number of proofs (see \CiteExtended{}).

\begin{figure}[t!]
  \begin{mathpar}
    \inferrule{
      (j<n)
    }{
      \CodeUni{j}:\Pred{\ClUni{n}}\CwfUniv[j]
    }
    \and
    \inferrule{
    }{
      \CodeBool:\Pred{\ClUni{n}}\CwfBool
    }
    \\
    \inferrule{
      \FmtCode{A}:\Pred{\ClUni{n}}A
      \\
      \FmtCode{B}:
      \textstyle
      \prod_{M : \NerveEl{A}}
      \UPred{\FmtCode{A}} M
      \to
      \Pred{\ClUni{n}}\parens*{
        \ReIx{\Snoc{\Id}{M}}{B}
      }
    }{
      \CodePi{\FmtCode{A}}{\FmtCode{B}} :\Pred{\ClUni{n}}\CwfPi{A}{B}
      \\
      \CodeSg{\FmtCode{A}}{\FmtCode{B}} :\Pred{\ClUni{n}}\CwfSg{A}{B}
    }
    \and
    \inferrule{
      \textstyle
      \FmtCode{A} : \prod_{i:\Dim} \Pred{\ClUni{n}}{A_i}
      \\
      \etc{\Pred{N_\e} : \UPred{\FmtCode{A}(\e)}N_\e}
    }{
      \CodeEq{\FmtCode{A}}{\Pred{N_0}}{\Pred{N_1}}
      :
      \Pred{\ClUni{n}}\CwfEq{i.A_i}{N_0}{N_1}
    }
  \end{mathpar}
  \LightRule
  \begin{align*}
    \UPred{\CodeUni{n}}A &= \Pred{\ClUni{n}}A
    \\
    \UPred{\CodeBool}M &=
    \parens*{M = \CwfTrue}
    +
    \parens*{M = \CwfFalse}
    \\
    \UPred{
      \CodePi{\FmtCode{A}}{\FmtCode{B}}
    }M &=
    \textstyle
    \prod_{N:\NerveEl{A}}
    \prod_{\Pred{N}:\Pred{\FmtCode{A}}N}
    \UPred*{\FmtCode{B}N\Pred{N}}
    \CwfApp[A][B]{M}{N}
    \\
    \UPred{
      \CodeSg{\FmtCode{A}}{\FmtCode{B}}
    }M &=
    \textstyle
    \sum_{
      \Pred{M_0} : \UPred{\FmtCode{A}}\CwfFst[A][B]{M}
    }
    \UPred*{\FmtCode{B}\parens*{\CwfFst[A][B]{M}}\Pred{M_0}}\CwfSnd[A][B]{M}
    \\
    \UPred{
      \CodeEq{\FmtCode{A}}{\Pred{N_0}}{\Pred{N_1}}
    }M
    &=
    \textstyle
    \braces*{
      \Pred{M} :
      \prod_{i:\Dim}
      \UPred{\FmtCode{A}(i)}\CwfPathApp[i.A]{M}{i}
      \mid
      \etc{
        \Pred{M}(\e) = \Pred{N_\e}
      }
    }
  \end{align*}
  \LightRule
  \begin{align*}
    \coe{i.\CodeBool}{r}{r'}{\Pred{M}} &= \Pred{M}
    \\
    \coe{i.\CodePi{\FmtCode{A}}{\FmtCode{B}}}{r}{r'}{\Pred{M}}
    &=
    \lambda \Pred{N}.\,
    \coe{
      i.\FmtCode{B}\parens{
        \coe{i.\FmtCode{A}}{r'}{i}{\Pred{N}}
      }
    }{r}{r'}{
      \Pred{M}\parens[\big]{
        \coe{i.\FmtCode{A}}{r'}{r}{\Pred{N}}
      }
    }
    \\
    \coe{
      i.\CodeEq{\FmtCode{A}}{\Pred{N_0}}{\Pred{N_1}}
    }{r}{r'}{\Pred{M}}
    &=
    \lambda k.\,
    \com{i.\FmtCode{A} k}{r}{r'}{
      \Pred{M}k
    }{
      \etcsys{k}{\_.\Pred{N_\e}}
    }
    \\
    \hcom{
      \CodePi{\FmtCode{A}}{\FmtCode{B}}
    }{r}{r'}{\Pred{M}}{
      \etcsys{s}{i.\Pred{M'}i}
    }
    &=
    \lambda \Pred{N}.\,
    \hcom{
      \FmtCode{B}\Pred{N}
    }{r}{r'}{
      \Pred{M}\Pred{N}
    }{
      \etcsys{s}{i.\Pred{M'}i N \Pred{N}}
    }
    \\
    \hcom{
      \CodeEq{\FmtCode{A}}{\Pred{N_0}}{\Pred{N_1}}
    }{r}{r'}{\Pred{M}}{
      \etcsys{s}{i.\Pred{M'}i}
    } &=
    \lambda j.\,
    \hcom{
      \FmtCode{A} j
    }{r}{r'}{\Pred{M} j}{
      \etcsys{s}{i.\Pred{M'}i j}
    }
    \\
    & \mathrel{\makebox[\widthof{=}]{\vdots}}
  \end{align*}

  \caption{The inductive definition of realizers
  $\Pred{\ClUni{n}}A:\SemUniv[n+1]$ for types $A:\NerveTy[n]$ in $\AugCSet$; we also include a fragment of the realizers for Kan
  operations, which are also defined by recursion on the realizers for types. We
  write $\Dim$ for the representable presheaf $\Yo*{i}$.}
  \label{fig:type-codes}
\end{figure}

From all this, we can define the cwf structure on $\CPred$. We obtain a presheaf of types
$\CwfTy<\CPred>: \Psh{\CPred\times\Lev}$ by taking $\CwfTy<\CPred>[k]{\Gl{\Gamma}}$ to be the
set of pairs $\Gl{A}=(A, \Pred{A})$ where $A\in\CwfTy[k]{\Gamma}$ and
$\Pred{A}$ is an element of the type
$\prod_{\gamma:\Nerve{\Gamma}}\prod_{\Pred{\gamma}:\Pred{\Gamma}(\gamma)}\Pred{\ClUni{k}}\parens{\ReIx{\gamma}{A}}$
in the internal type theory of $\AugCSet$. To define the dependent presheaf
of elements, we take $\CwfEl<\CPred>{\Gl{\Gamma}}{\Gl{A}}$ to be the set of
pairs $\Gl{M}=(M,\Pred{M})$ where $M\in\CwfEl{\Gamma}{A}$ and $\Pred{M}$ is an
element of the type
$\prod_{\gamma:\Nerve{\Gamma}}\prod_{\Pred{\gamma}:\Pred{\Gamma}(\gamma)}\UPred*{\Pred{A}\gamma\Pred{\gamma}}\parens{\ReIx{\gamma}{M}}$
in the internal type theory of $\AugCSet$.
In this model, the context comprehension operation $\Gl{\Gamma}.\Gl{A}$ is defined as the pair
$(\Gamma.A, \Pred{\parens{\Gl{\Gamma}.\Gl{A}}})$ where
$\Pred{\parens{\Gl{\Gamma}.\Gl{A}}}\Snoc{\gamma}{M}$ is the cubical set
$\sum_{\Pred{\gamma}:\Pred{\Gamma}(\gamma)}\UPred*{\Pred{A}\gamma\Pred{\gamma}}\parens{\ReIx{\gamma}{M}}$;
it is easy to see that we obtain realizers for the weakening substitution and
the variable term.

\begin{construction}[Dependent equality types in $\CPred$]
  Recall that we required a model of \XTT{} to have sufficient structure to interpret dependent
  equality types. Here, we discuss how to obtain the formation rule; the full
  construction can be found in \CiteExtended{}. Suppose
  $\Gl{A}\in\CwfTy<\CPred>[n]{\ReIx{\Wk{\imath}}{\Gl{\Gamma}}}$ and elements $\Gl{N_0}$ and
  $\Gl{N_1}$ with $\Gl{N}_\e\in\CwfEl<\CPred>{\Gl{\Gamma}}{\ReIxCleave*{\e/i}{\Gl{A}}}$. We wish to
  construct a type in $\CwfTy<\CPred>[n]{\ReIx{\Wk{\imath}}{\Gl{\Gamma}}}$.

  In $\CPred$, such a type is a pair of a type
  $E \in \CwfTy[n]{\ReIx{\Wk{\imath}}{\Gamma}}$ from $\Cx$ with
  an element witnessing the logical family
  $\prod_{\gamma:\Nerve{\Gamma}}\prod_{\Pred{\gamma}:\Pred{\Gamma}(\gamma)}\Pred{\ClUni{k}}\parens{\ReIx{\gamma}{E}}$.
  We will set the first component to the dependent equality type from $\Cx$ itself, namely
  $E = \CwfEq{i.A}{N_0}{N_1}$. For the second component, we wish to construct an element of
  $\prod_{\gamma:\Nerve{\Gamma}}\prod_{\Pred{\gamma}:\Pred{\Gamma}(\gamma)}\Pred{\ClUni{k}}\parens{\ReIx{\gamma}{\CwfEq{i.A}{N_0}{N_1}}}$.
  Inspecting the rules for $\Pred{\ClUni{k}}$ from \cref{fig:type-codes}, there is only one choice:
  $\Pred{E} = \lam{\gamma}{\lam{\Pred{\gamma}}{\CodeEq{\Pred{A}\gamma\Pred{\gamma}}{\Pred{N_0}\gamma\Pred{\gamma}}{\Pred{N_1}\gamma\Pred{\gamma}}}}$.
\end{construction}
\begin{construction}[Coercion in $\CPred$]
  The coercion structure on $\CPred$ is constructed from the coercion structures on $\Cx$ and the
  coercion operator for codes from \cref{fig:type-codes}.

  Given a type
  $\Gl{A}\in\CwfTy<\Gl{\Gamma}>[n]{\ReIx{\Wk{\imath}}{\Gl{\Gamma}}}$
  over $\Psi,i$, dimensions $r,r'\in\CwfDim{\Gl{\Gamma}}$, and an element
  $\Gl{M}\in\CwfEl{\Gl{\Gamma}}{\ReIxCleave*{r/i}[\ReIx{\Wk{\imath}}{\Gl{\Gamma}}]{\Gl{A}}}$,
  we must construct an element of
  $\CwfEl{\Gl{\Gamma}}{\ReIxCleave*{r'/i}[\ReIx{\Wk{\imath}}{\Gl{\Gamma}}]{\Gl{A}}}$. This element
  must be a pair of $N \in \CwfEl{\Gamma}{\ReIxCleave*{r'/i}[\ReIx{\Wk{\imath}}{\Gamma}]{A}}$ and a
  term
  $\Pred{N} : \prod_{\gamma:\Nerve{\Gamma}}\prod_{\Pred{\gamma}:\Pred{\Gamma}(\gamma)}\UPred*{\Pred{A}\gamma\Pred{\gamma}}\parens{\ReIx{\gamma}{N}}$.
  For the former, we rely on the coercion structure for $\Cx$ and pick $N = \CwfCoe{i.A}{r}{r'}{M}$.
  For the latter, we use the coercion operation on codes defined in \cref{fig:type-codes} and choose
  $\Pred{N} = \lam{\gamma}{\lam{\Pred{\gamma}}{\coe{i.\Pred{A}\gamma \Pred{\gamma}}{r}{r'}{\Pred{M}\gamma\Pred{\gamma}}}}$.

  It is routine to check that this coercion structure enjoys adjacency,
  regularity, and naturality once the corresponding properties are checked for
  the coercion operator on codes.
\end{construction}

\begin{theorem}
  $\CPred$ is an \XTT-algebra, and moreover, $\CPred\rProjto^{\PiSyn}\Cx$ is a homomorphism of \XTT-algebras.
\end{theorem}

\subsection{Canonicity theorem}\label{sec:canonicity-theorem}

Because $\CPred$ is an \XTT-algebra, we are now equipped to prove a
canonicity theorem for the initial \XTT-algebra $\Cx$: if $M$ is an element
of type $\CwfBool$ in the empty context, then either $M=\CwfTrue$ or
$M=\CwfFalse$, and not both.

\begin{proof}

  We have $M\in\CwfEl{\cdot}{\CwfBool}$, and therefore
  $\Interp{M}\in\CwfEl<\CPred>{\cdot}{\Gl{\CwfBool}}$. From this we
  obtain $N:\CwfEl{\cdot}{\CwfBool}$ where $N=\PiSyn\Interp{M}$, and
  $\Pred{N}\in\UPred{\CodeBool}_\cdot(N)$; by definition, $\Pred{N}$ is either a
  proof that $N=\CwfTrue$ or a proof that $N=\CwfFalse$ (see \cref{fig:type-codes}).
  Therefore, it suffices to observe that $\PiSyn\Interp{M}=M$; but this follows
  from the universal property of the initial \XTT-algebra and the fact that
  $\CPred\rProjto^\PiSyn\Cx$ is an \XTT-homomorphism. Moreover, because the
  interpretation of $\CwfBool$ in $\CPred$ is disjoint, $M$ cannot equal both
  $\CwfTrue$ and $\CwfFalse$.  \qedhere

\end{proof}

\bibliography{../references/refs-bibtex}

\clearpage
\appendix

\section{The \texorpdfstring{\XTT}{XTT} language}\label{sec:calculus}

\paragraph*{Annotated syntax}

The raw syntax of \XTT{} includes typing annotations on function application
$\app[x:A.B]{M}{N}$ and pair projections $\fst[x:A.B]{M}$ and $\snd[x:A.B]{M}$,
in order to ensure that the raw syntax could (in theory) be organized into an
\emph{initial} model of \XTT{}, in the sense of \cref{sec:cwf-structure}. A
version of the syntax with fewer annotations would be justified by a normalization
result for \XTT{}, which we do \emph{not} establish here.

Because these annotations can visually obscure the meaning of a term, we adopt
the notational convention that when a term is already known to be well-typed,
we omit the annotation and write $\app{M}{N}$ for
$\app[x:A.B]{M}{N}$, and likewise $\fst{M}$ for $\fst[x:A.B]{M}$, etc.

\paragraph*{Heterogeneous composition}
Following previous work \Cite{angiuli-favonia-harper:2018}, we take coercion and
homogeneous composition as primitive operations, and define heterogeneous composition in
terms of it:
\[
  \com{i.A}{r}{r'}{M}{\etcsys[\e]{s}{j.N_\e}} \triangleq
  \hcom{\dsubst{A}{r'}{j}}{r}{r'}{
    \parens*{\coe{j.A}{r}{r'}{M}}
  }{
    \etcsys[\e]{s}{j. \coe{j.A}{j}{r'}{N_\e}}
  }
\]

Other versions of cubical type theory, such as De Morgan cubical type theory \Cite{cchm:2017}, take
heterogeneous composition as primitive and derive both coercion and homogeneous
composition as a special case. In our setting, it is especially advantageous to
take coercion and homogeneous composition as a primitive, because in \XTT{} it
is only necessary to provide $\beta$-rules for coercion; in
\cref{sec:rules:derivable}, we observe that all the $\beta$-rules for
homogeneous composition are in fact already derivable, by exploiting the path
unicity rule in \cref{sec:rules:structural}.

\begin{convention}[Presupposition]\label{convention:presupposition}
  The \XTT{} language involves many forms of judgment, each of which is defined conditionally on a
  \emph{presupposition}; in type-theoretic formal systems, a judgment expresses
  the well-formedness of a raw term (the ``subject'') relative to some
  parameters. The parameters themselves are rarely raw terms, but rather terms
  that are already known to be well-formed according to certain judgments (called
  ``presuppositions'').

  We indicate this situation schematically
  for a form of judgment $\mathcal{J}$ in the following way, where
  $\etc{\mathfrak{p}_i}$ are parameters and $\mathfrak{q}$ is a subject:
  
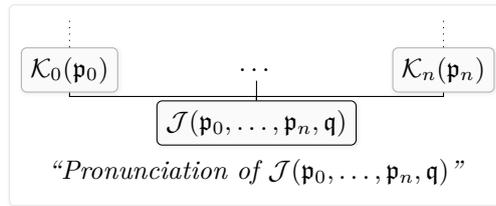
\begin{mathpar}
    \tikz[framed,edge from parent fork up, level distance=2em, sibling distance = 7em]{
      \draw
        node[jdg] {$\mathcal{J}\parens*{\mathfrak{p}_0,\ldots,\mathfrak{p}_n,\mathfrak{q}}$} [grow'=up]
        child {
          node[presup] {$\mathcal{K}_0\parens*{\mathfrak{p}_0}$}
          child[dotted]
        }
        child{ node {\ldots} }
        child{
          node[presup] {$\mathcal{K}_n\parens*{\mathfrak{p}_n}$}
          child[dotted]
        }
        node[below=1em] {\textit{``Pronunciation of $\mathcal{J}\parens{\mathfrak{p}_0,\ldots,\mathfrak{p}_n,\mathfrak{q}}$''}}
    }
  \end{mathpar}
\end{convention}

\subsection{The judgments of \texorpdfstring{\XTT}{XTT}}
\label{sec:judgments}

The judgments of \XTT{} are summarized below using \cref{convention:presupposition}.

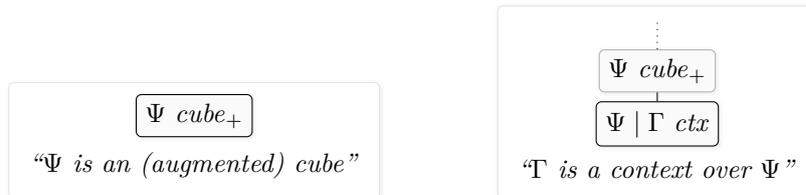
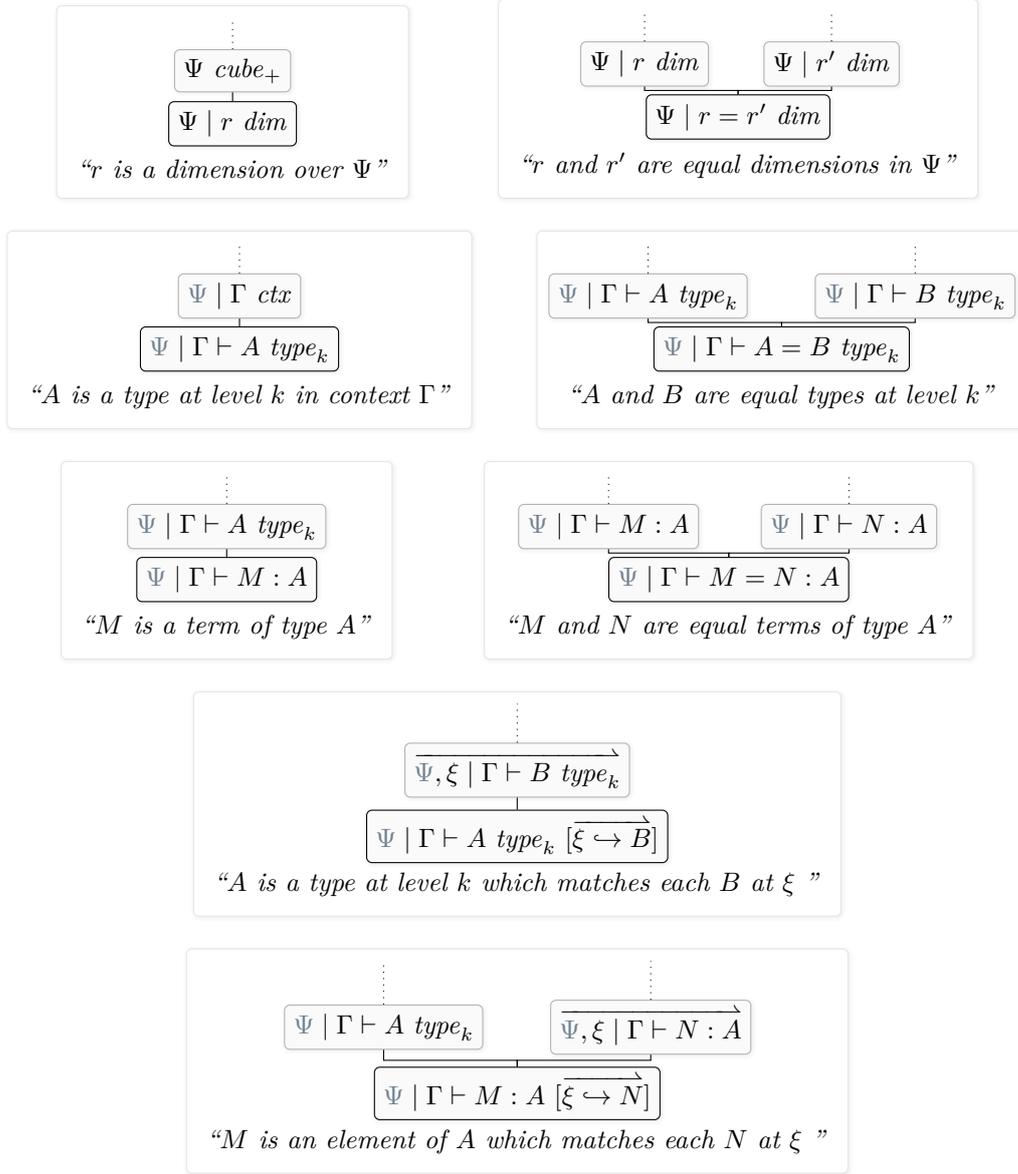
\begin{mathparpagebreakable}
  \tikz[framed,edge from parent fork up, level distance=2em, sibling distance = 7em]{
    \draw
      node[jdg] {$\IsDimCx{\Psi}$} [grow'=up]
      node[below=1em] {\textit{``$\Psi$ is an (augmented) cube''}}
  }
  \and
  \tikz[framed,edge from parent fork up, level distance=2em, sibling distance = 7em]{
    \draw
      node[jdg] {$\IsCx{\Gamma}$} [grow'=up]
      child {
        node[presup] { $\IsDimCx{\Psi}$ }
        child[dotted] {}
      }
      node[below=1em] {\textit{``$\Gamma$ is a context over $\Psi$''}}
  }
  \and
  \tikz[framed,edge from parent fork up, level distance=2em, sibling distance = 7em]{
    \draw
      node[jdg] {$\IsDim<\Psi>{r}$} [grow'=up]
      child {
        node[presup] { $\IsDimCx{\Psi}$ }
        child[dotted] {}
      }
      node[below=1em] {\textit{``$r$ is a dimension over $\Psi$''}}
  }
  \and
  \tikz[framed,edge from parent fork up, level distance=2em, sibling distance = 7em]{
    \draw
      node[jdg] {$\EqDim<\Psi>{r}{r'}$} [grow'=up]
      child {
        node[presup] { $\IsDim<\Psi>{r}$ }
        child[dotted] {}
      }
      child {
        node[presup] { $\IsDim<\Psi>{r'}$ }
        child[dotted] {}
      }
      node[below=1em] {\textit{``$r$ and $r'$ are equal dimensions in $\Psi$''}}
  }
  \and
  \tikz[framed,edge from parent fork up, level distance=2em, sibling distance = 7em]{
    \draw
      node[jdg] {$\IsTy<\Mute{\Psi}>[\Gamma]{A}$} [grow'=up]
      child {
        node[presup] { $\IsCx<\Mute{\Psi}>{\Gamma}$ }
        child[dotted] {}
      }
      node[below=1em] {\textit{``$A$ is a type at level $k$ in context $\Gamma$''}}
  }
  \and
  \tikz[framed,edge from parent fork up, level distance=2em, sibling distance = 10em]{
    \draw
      node[jdg] {$\EqTy<\Mute{\Psi}>{A}{B}$} [grow'=up]
      child {
        node[presup] { $\IsTy<\Mute{\Psi}>{A}$ }
        child[dotted] {}
      }
      child {
        node[presup] { $\IsTy<\Mute{\Psi}>{B}$ }
        child[dotted] {}
      }
      node[below=1em] {\textit{``$A$ and $B$ are equal types at level $k$''}}
  }
  \and
  \tikz[framed,edge from parent fork up, level distance=2em, sibling distance = 7em]{
    \draw
      node[jdg] {$\IsTm<\Mute{\Psi}>{M}{A}$} [grow'=up]
      child {
        node[presup] { $\IsTy<\Mute{\Psi}>{A}$ }
        child[dotted] {}
      }
      node[below=1em] {\textit{``$M$ is a term of type $A$''}}
  }
  \and
  \tikz[framed,edge from parent fork up, level distance=2em, sibling distance = 9em]{
    \draw
      node[jdg] {$\EqTm<\Mute{\Psi}>{M}{N}{A}$} [grow'=up]
      child {
        node[presup] { $\IsTm<\Mute{\Psi}>{M}{A}$ }
        child[dotted] {}
      }
      child {
        node[presup] { $\IsTm<\Mute{\Psi}>{N}{A}$ }
        child[dotted] {}
      }
      node[below=1em] {\textit{``$M$ and $N$ are equal terms of type $A$''}}
  }
  \and
  \tikz[framed,edge from parent fork up, level distance=2.5em, sibling distance = 9em]{
    \draw
      node[jdg] {$\IsTy<\Mute{\Psi}>{A}{\etc{\tube{\xi}{B}}}$} [grow'=up]
      child {
        node[presup] { $\etc{\IsTy<\Mute{\Psi},\xi>{B}}$ }
        child[dotted] {}
      }
      node[below=1em] {\textit{``$A$ is a type at level $k$ which matches each $B$ at $\xi$ ''}}
  }
  \and
  \tikz[framed,edge from parent fork up, level distance=2.5em, sibling distance = 10em]{
    \draw
      node[jdg] {$\IsTm<\Mute{\Psi}>{M}{A}{\etc{\tube{\xi}{N}}}$} [grow'=up]
      child {
        node[presup] { $\IsTy<\Mute{\Psi}>{A}$ }
        child[dotted] {}
      }
      child {
        node[presup] { $\etc{\IsTm<\Mute{\Psi},\xi>{N}{A}}$ }
        child[dotted] {}
      }
      node[below=1em] {\textit{``$M$ is an element of $A$ which matches each $N$ at $\xi$ ''}}
  }
\end{mathparpagebreakable}

\subsection{The rules of \texorpdfstring{\XTT}{XTT}}\label{sec:xtt-rules}

In the following sections, we summarize the rules of \XTT{}; we systematically
omit obvious premises to equational rules and all congruence rules for
judgmental equality, because these can be mechanically obtained from the typing
rules.

\subsubsection{Cubes}

\begin{ruleblock}
  \Rule[emp]{
  }{
    \IsDimCx{\cdot}
  }
  \and
  \Rule[snoc/dim]{
    \IsDimCx{\Psi}
  }{
    \IsDimCx{\Psi,i}
  }
  \and
  \Rule[snoc/constr]{
    \IsDimCx{\Psi}\\
    \IsDim<\Psi>{r,r'}
  }{
    \IsDimCx{\Psi,r = r'}
  }
\end{ruleblock}

\subsubsection{Contexts}
\begin{ruleblock}
  \Rule[emp]{
  }{
    \IsCx{\cdot}
  }
  \and
  \Rule[snoc]{
    \IsCx{\Gamma}
    \\
    \IsTy[\Gamma]{A}
  }{
    \IsCx{\Gamma,x:A}
  }
\end{ruleblock}

\subsubsection{Dimensions}
\begin{ruleblock}
  \Rule[constant]{
  }{
    \IsDim{\e}
  }
  \and
  \Rule[variable]{
    i\in \Psi
  }{
    \IsDim<\Psi>{i}
  }
  \and
  \Rule[reflexivity]{
  }{
    \EqDim{r}{r}
  }
  \and
  \Rule[symmetry]{
    \EqDim{r}{r'}
  }{
    \EqDim{r'}{r}
  }
  \and
  \Rule[transitivity]{
    \EqDim{r_0}{r_1}\\
    \EqDim{r_1}{r_2}
  }{
    \EqDim{r_0}{r_2}
  }
  \and
  \Rule[hyp]{
    \Psi \ni r = r'
  }{
    \EqDim<\Psi>{r}{r'}
  }
\end{ruleblock}

\subsubsection{Structural}\label{sec:rules:structural}

\begin{ruleblock}
  \Rule[variable]{
    \Gamma\ni{}x:A
  }{
    \IsTm[\Gamma]{x}{A}
  }
  \and
  \Rule[false constraint]{
    \EqDim{0}{1}
  }{
    \Jdg
  }
  \and
  \Rule[conversion]{
    \EqTy{A_0}{A_1}
    \\
    \IsTm{M}{A_0}
  }{
    \IsTm{M}{A_1}
  }
  \\
  \Rule[boundary separation (types)]{
    \IsDim{r}
    \\
    \etc{
      \EqTy<\Psi,r=\e>{A}{B}
    }
  }{
    \EqTy<\Psi>{A}{B}
  }
  \and
  \Rule[boundary separation (terms)]{
    \IsDim<\Psi>{r}
    \\
    \etc{
      \EqTm<\Psi,r=\e>{M}{N}{A}
    }
  }{
    \EqTm<\Psi>{M}{N}{A}
  }
\end{ruleblock}

The following rules are admissible:
\begin{ruleblock}
  \Rule[constraint cut]{
    \EqDim<\Psi>{r}{r'}
    \\
    \Jdg<\Psi,r=r'>
  }{
    \Jdg<\Psi>
  }
  \and
  \Rule[constraint weakening]{
    \Jdg<\Psi>
  }{
    \Jdg<\Psi,\xi>
  }
\end{ruleblock}

\subsubsection{Coercion}

\begin{ruleblock}
  \Rule[coercion]{
    \IsDim<\Psi>{r,r'}
    \\
    \IsTy<\Psi,i>{A}
    \\
    \IsTm<\Psi>{M}{\dsubst{A}{r}{i}}
  }{
    \IsTm<\Psi>{\coe{i.A}{r}{r'}{M}}{\dsubst{A}{r'}{i}}
  }
  \and
  \Rule[coercion boundary]{
  }{
    \EqTm<\Psi>{\coe{i.A}{r}{r}{M}}{M}{\dsubst{A}{r}{i}}
  }
  \and
  \Rule[coercion regularity]{
    \EqTy<\Psi,j,j'>{\dsubst{A}{j}{i}}{\dsubst{A}{j'}{i}}
  }{
    \EqTm<\Psi>{\coe{i.A}{r}{r'}{M}}{M}{\dsubst{A}{r'}{i}}
  }
\end{ruleblock}

\subsubsection{Composition}

\begin{ruleblock}
  \Rule[composition]{
    \IsDim<\Psi>{r,r',s}
    \\
    \IsTm<\Psi>{M}{A}
    \\
    \etc{
      \IsTm<\Psi,j,s=\e>{N_\e}{A}{
        \tube{j=r}{M}
      }
    }
  }{
    \IsTm<\Psi>{
      \hcom{A}{r}{r'}{M}{
        \etcsys[\e]{s}{j.N_\e}
      }
    }{A}
  }
  \and
  \Rule[composition boundary]{
  }{
    \EqTm{
      \hcom{A}{r}{r}{M}{
        \etcsys[\e]{s}{j.N_\e}
      }
    }{M}{A}
    \\\\
    \EqTm{
      \hcom{A}{r}{r'}{M}{
        \etcsys[\e']{\e}{j.N_{\e'}}
      }
    }{\dsubst{N_\e}{r'}{j}}{A}
  }
\end{ruleblock}

\subsubsection{Level restrictions}

\begin{ruleblock}
  \Rule[lift formation]{
    \IsTy{A}[k]
    \\
    k\leq l
  }{
    \IsTy{\TyLift{k}{l}{A}}[l]
  }
  \and
  {
    \mprset{fraction={===}}
    \Rule[lift element]{
      \IsTm{M}{A}
    }{
      \IsTm{M}{\TyLift{k}{l}{A}}
    }
    \and
    \Rule[lift hypothesis]{
      \Jdg[\Gamma,x:A]
    }{
      \Jdg[\Gamma,x:\TyLift{k}{l}{A}]
    }
  }
  \and
  \Rule[lift functoriality]{
  }{
    \EqTy{\TyLift{k}{k}{A}}{A}[k]
    \\\\
    \EqTy{\TyLift{l}{m}{\TyLift{k}{l}{A}}}{\TyLift{k}{m}{A}}[k]
  }
  \and
  \Rule[lift coercion]{
  }{
    \EqTm{
      \coe{i.\TyLift{k}{l}{A}}{r}{r'}{M}
    }{
      \coe{i.A}{r}{r'}{M}
    }{
      \TyLift{k}{l}{\dsubst{A}{r'}{i}}
    }
  }
\end{ruleblock}

\subsubsection{Dependent pair types}

\begin{ruleblock}
  \Rule[pair formation]{
    \IsTy[\Gamma]{A}
    \\\\
    \IsTy[\Gamma,x:A]{B}
  }{
    \IsTy[\Gamma]{\sigmacl{x}{A}{B}}
  }
  \and
  \Rule[pair lifting]{
  }{
    \EqTy{
      \TyLift{k}{l}{\sigmacl{x}{A}{B}}
    }{
      \sigmacl{x}{
        \TyLift{k}{l}{A}
      }{
        \TyLift{k}{l}{B}
      }
    }[l]
  }
  \and
  \Rule[pair introduction]{
    \IsTy[\Gamma]{A}
    \\
    \IsTy[\Gamma,x:A]{B}
    \\
    \IsTm{M}{A}
    \\
    \IsTm{N}{\subst{B}{M}{x}}
  }{
    \IsTm{\pair{M}{N}}{\sigmacl{x}{A}{B}}
  }
  \\
  \Rule[pair elimination]{
    \IsTy[\Gamma]{A}
    \\\\
    \IsTy[\Gamma,x:A]{B}
    \\\\
    \IsTm{M}{\sigmacl{x}{A}{B}}
  }{
    \IsTm{\fst[x:A.B]{M}}{A}
    \\\\
    \IsTm{\snd[x:A.B]{M}}{\subst{B}{\fst{M}}{x}}
  }
  \and
  \Rule[pair elimination lifting]{
  }{
    \EqTm{
      \fst[x:\TyLift{k}{l}{A}.\TyLift{k}{l}{B}]{M}
    }{
      \fst[x:A.B]{M}
    }{A}
    \\\\
    \EqTm{
      \snd[x:\TyLift{k}{l}{A}.\TyLift{k}{l}{B}]{M}
    }{
      \snd[x:A.B]{M}
    }{\subst{B}{\fst{M}}{x}}
  }
  \and
  \Rule[pair computation]{
  }{
    \EqTm{\fst{\pair{M}{N}}}{M}{A}
    \\\\
    \EqTm{\snd{\pair{M}{N}}}{N}{\subst{B}{M}{x}}
  }
  \and
  \Rule[pair coercion computation (1)]{
  }{
    \EqTm{
      \fst{
        \coe{i.\sigmacl{x}{A}{B}}{r}{r'}{M}
      }
    }{
      \coe{i.A}{r}{r'}{\fst{M}}
    }{\dsubst{A}{r'}{i}}
  }
  \and
  \Rule[pair coercion computation (2)]{
    \Def{H}{
      \coe{i.\subst{B}{\coe{i.A}{r}{i}{\fst{M}}}{x}}{r}{r'}{\snd{M}}
    }
  }{
    \EqTm{
      \snd{
        \coe{i.\sigmacl{x}{A}{B}}{r}{r'}{M}
      }
    }{H}{
      \subst{\dsubst{B}{r'}{i}}{
        \coe{i.A}{r}{r'}{\fst{M}}
      }{x}
    }
  }
  \and
  \Rule[pair unicity]{
  }{
    \EqTm{M}{\pair{\fst{M}}{\snd{M}}}{\sigmacl{x}{A}{B}}
  }
\end{ruleblock}

\subsubsection{Dependent function types}

\begin{ruleblock}
  \Rule[function formation]{
    \IsTy[\Gamma]{A}
    \\\\
    \IsTy[\Gamma,x:A]{B}
  }{
    \IsTy[\Gamma]{\picl{x}{A}{B}}
  }
  \and
  \Rule[function lifting]{
  }{
    \EqTy{
      \TyLift{k}{l}{\picl{x}{A}{B}}
    }{
      \picl{x}{
        \TyLift{k}{l}{A}
      }{
        \TyLift{k}{l}{B}
      }
    }[l]
  }
  \and
  \Rule[function introduction]{
    \IsTy[\Gamma]{A}
    \\
    \IsTy[\Gamma,x:A]{B}
    \\
    \IsTm[\Gamma,x:A]{M}{B}
  }{
    \IsTm[\Gamma]{\lam{x}{M}}{\picl{x}{A}{B}}
  }
  \and
  \Rule[function elimination]{
    \IsTy[\Gamma]{A}
    \\
    \IsTy[\Gamma,x:A]{B}
    \\\\
    \IsTm{M}{\picl{x}{A}{B}}
    \\
    \IsTm{N}{A}
  }{
    \IsTm{\app[x:A.B]{M}{N}}{\subst{B}{N}{x}}
  }
  \and
  \Rule[function elimination lifting]{
  }{
    \EqTm{
      \app[x:\TyLift{k}{l}{A}.\TyLift{k}{l}{B}]{M}{N}
    }{
      \app[x:A.B]{M}{N}
    }{\subst{x}{N}{B}}
  }
  \and
  \Rule[function computation]{
  }{
    \EqTm[\Gamma]{
      \app{
        \parens{\lam{x}{M}}
      }{N}
    }{
      \subst{M}{N}{x}
    }{
      \subst{B}{N}{x}
    }
  }
  \and
  \Rule[function coercion computation]{
    \Def<\Psi,i>{\widetilde{N}[i]}{
      \coe{i.A}{r'}{i}{N}
    }
  }{
    \EqTm<\Psi>{
      \app{
        \parens{
          \coe{i.\picl{x}{A}{B}}{r}{r'}{M}
        }
      }{N}
    }{
      \coe{i.\subst{B}{\widetilde{N}[i]}{x}}{r}{r'}{M\parens{\widetilde{N}[r]}}
    }{C}
  }
  \and
  \Rule[function unicity]{
  }{
    \EqTm{M}{\lam{x}{M(x)}}{\picl{x}{A}{B}}
  }
\end{ruleblock}

\subsubsection{Dependent equality types}
\begin{ruleblock}
  \Rule[equality formation]{
    \IsTy<\Psi,i>{A}
    \\
    \etc{
      \IsTm<\Psi,i,i=\e>{N_\e}{A}
    }
  }{
    \IsTy<\Psi>{\Eq{i.A}{N_0}{N_1}}
  }
  \and
  \Rule[equality lifting]{
  }{
    \EqTy{
      \TyLift{k}{l}{\Eq{i.A}{N_0}{N_1}}
    }{
      \Eq{i.\TyLift{k}{l}{A}}{N_0}{N_1}
    }[l]
  }
  \and
  \Rule[equality introduction]{
    \IsTm<\Psi,i>{M}{A}{
      \etc{
        \tube{i=\e}{N_\e}
      }
    }
  }{
    \IsTm<\Psi>{\lam{i}{M}}{\Eq{i.A}{N_0}{N_1}}
  }
  \and
  \Rule[equality elimination]{
    \IsDim{r}
    \\
    \IsTy<\Psi,i>{A}
    \\
    \etc{
      \IsTm<\Psi,i,i=\e>{N_\e}{A}
    }
    \\
    \IsTm{M}{\Eq{i.A}{N_0}{N_1}}
  }{
    \IsTm{\papp[i.A]{M}{r}}{\dsubst{A}{r}{i}}
  }
  \and
  \Rule[equality elimination lifting]{
  }{
    \EqTm{\papp[i.\TyLift{k}{l}{A}]{M}{r}}{\papp[i.A]{M}{r}}{\dsubst{A}{r}{i}}
  }
  \and
  \Rule[equality boundary]{
    \IsTm{M}{\Eq{i.A}{N_0}{N_1}}
  }{
    \EqTm{M(\e)}{N_\e}{\dsubst{A}{\e}{i}}
  }
  \and
  \Rule[equality computation]{
  }{
    \EqTm<\Psi>{(\lam{i}{M})(r)}{\dsubst{M}{r}{i}}{\dsubst{A}{r}{i}}
  }
  \and
  \Rule[equality coercion computation]{
  }{
    \EqTm{
      \parens*{\coe{j.\Eq{i.A}{N_0}{N_1}}{r}{r'}{P}}(s)
    }{
      \com{j.\dsubst{A}{s}{i}}{r}{r'}{P(s)}{
        \etcsys[\e]{s}{j.N_\e}
      }
    }{
      \dsubst{A}{r',s}{j,i}
    }
  }
  \and
  \Rule[equality unicity]{
  }{
    \EqTm{M}{\lam{i}{M(i)}}{\Eq{i.A}{N_0}{N_1}}
  }
\end{ruleblock}

\subsubsection{Booleans}

\begin{ruleblock}
  \Rule[boolean formation]{
  }{
    \IsTy{\bool}
  }
  \and
  \Rule[boolean lifting]{
  }{
    \EqTy{\TyLift{k}{l}{\bool}}{\bool}[l]
  }
  \and
  \Rule[boolean introduction]{
  }{
    \IsTm{\true}{\bool}
    \\\\
    \IsTm{\false}{\bool}
  }
  \and
  \Rule[boolean elimination]{
    \IsTy[\Gamma,x:\bool]{C}
    \\
    \IsTm[\Gamma]{M}{\bool}
    \\
    \IsTm[\Gamma]{N_0}{\subst{C}{\true}{x}}
    \\
    \IsTm[\Gamma]{N_1}{\subst{C}{\false}{x}}
  }{
    \IsTm[\Gamma]{
      \ifb{x.C}{M}{N_0}{N_1}
    }{
      \subst{C}{M}{x}
    }
  }
  \and
  \Rule[boolean elimination lifting]{
  }{
    \EqTm{
      \ifb{x.\TyLift{k}{l}{C}}{M}{N_0}{N_1}
    }{
      \ifb{x.C}{M}{N_0}{N_1}
    }{
      \subst{\TyLift{k}{l}{C}}{M}{x}
    }
  }
  \and
  \Rule[boolean computation]{
  }{
    \EqTm{
      \ifb{x.C}{\true}{N_0}{N_1}
    }{
      N_0
    }{
      \subst{C}{\true}{x}
    }
    \\\\
    \EqTm{
      \ifb{x.C}{\false}{N_0}{N_1}
    }{
      N_1
    }{
      \subst{C}{\false}{x}
    }
  }
\end{ruleblock}

\subsubsection{Universes}

\begin{ruleblock}
  \Rule[universe formation]{
    k<l
  }{
    \IsTy{\Univ[k]}[l]
  }
  \and
  \Rule[universe lifting]{
  }{
    \EqTy{\TyLift{l}{m}{\Univ[k]}}{\Univ[k]}[m]
  }
  \and
  {
    \mprset{fraction={===}}
    \Rule[universe elements]{
      \IsTy{A}[k]
    }{
      \IsTm{A}{\Univ[k]}
    }
    \and
    \Rule[universe equality]{
      \EqTy{A_0}{A_1}[k]
    }{
      \EqTm{A_0}{A_1}{\Univ[k]}
    }
  }
  \and
  \Rule[type-case]{
    \IsTy[\Gamma]{C}[l]
    \\\\
    \IsTm[\Gamma,x:\Univ[k],y:x\to\Univ[k]]{M_\Pi}{C}
    \\\\
    \IsTm[\Gamma,x:\Univ[k],y:x\to\Univ[k]]{M_\Sigma}{C}
    \\\\
    \IsTm[\Gamma,x_0:\Univ[k],x_1:\Univ[k],x^=:\Eq{i.\Univ[k]}{x_0}{x_1},y_0:x_0,y_1:x_1]{
      M_{\mathsf{Eq}}
    }{C}
    \\\\
    \IsTm[\Gamma]{M_\bool}{C}
    \\\\
    \IsTm[\Gamma]{M_{\mathcal{U}}}{C}
  }{
    \IsTm[\Gamma]{
      \UCase{X}{
        \UCaseBranch{\Pi_x y}{M_\Pi}
        \mid
        \UCaseBranch{\Sigma_x y}{M_\Sigma}
        \mid
        \UCaseBranch{\Eq{x_0,x_1,x^=}{y_0}{y_1}}{M_\mathsf{Eq}}
        \mid
        \UCaseBranch{\bool}{M_{\bool}}
        \mid
        \UCaseBranch{\mathcal{U}}{M_{\mathcal{U}}}
      }
    }{C}
  }
  \and
  \Rule[type-case computation]{
  }{
    \EqTm{
      \UCase*{\picl{z}{A}{B}}{
        \UCaseBranch{\Pi_x y}{M}
        \mid
        \ldots
      }
    }{
      \subst{M}{
        A, \lam{z}{B}
      }{x,y}
    }{C}
    \\
    \EqTm{
      \UCase*{\sigmacl{z}{A}{B}}{
        \ldots
        \mid
        \UCaseBranch{\Sigma_x y}{M}
        \mid
        \ldots
      }
    }{
      \subst{M}{
        A, \lam{z}{B}
      }{x,y}
    }{C}
    \\
    \EqTm{
      \UCase{\bool}{
        \ldots
        \mid
        \UCaseBranch{\bool}{M}
        \mid
        \ldots
      }
    }{M}{C}
    \\
    \EqTm{
      \UCase{\Univ[k']}{
        \ldots\mid\UCaseBranch{\mathcal{U}}{M}
      }
    }{M}{C}
  }
  \and
  \Rule{
    \Def{H}{
      \subst{M}{
        \dsubst{A}{0}{i},
        \dsubst{A}{1}{i},
        \lam{i}{A},
        N_0,
        N_1
      }{x_0,x_1,x^=,y_0,y_1}
    }
  }{
    \EqTm{
      \UCase*{\Eq{i.A}{N_0}{N_1}}{
        \ldots
        \mid
        \UCaseBranch{
          \Eq{x_0,x_1,x^=}{y_0}{y_1}
        }{M}
        \mid
        \ldots
      }
    }{H}{C}
  }
\end{ruleblock}

\subsubsection{Boundary matching}

\begin{ruleblock}
  \Rule[type boundary]{
    \IsTy<\Psi>{A}
    \and
    \etc{
      \EqTy<\Psi,\xi>{A}{B}
    }
  }{
    \IsTy<\Psi>{A}{
      \etc{
        \tube{\xi}{B}
      }
    }
  }
  \and
  \Rule[term boundary]{
    \IsTm<\Psi>{M}{A}
    \and
    \etc{
      \EqTm<\Psi,\xi>{M}{N}{A}
    }
  }{
    \IsTm<\Psi>{M}{A}{
      \etc{
        \tube{\xi}{N}
      }
    }
  }
\end{ruleblock}

\subsection{Derivable Rules}\label{sec:rules:derivable}

Numerous additional rules about compositions are \emph{derivable} by
exploiting boundary separation. In previous presentations of cubical
type theory (which did not enjoy the unicity of equality proofs), it was
necessary to include $\beta$-rules for compositions explicitly.

\begin{ruleblock}
  \Rule[composition regularity]{
    \etc{
      \EqTm<\Psi,j_0,j_1,i=\e>{
        \dsubst{N_\e}{j_0}{j}
      }{
        \dsubst{N_\e}{j_1}{j}
      }{A}
    }
  }{
    \EqTm<\Psi>{
      \hcom{A}{r}{r'}{M}{
        \etcsys[\e]{i}{j.N_\e}
      }
    }{M}{A}
  }
  \and
  \Rule[heterogeneous composition]{
    \IsDim<\Psi>{r,r',s}
    \\
    \IsTm<\Psi>{M}{\dsubst{A}{r}{j}}
    \\
    \etc{
      \IsTm<\Psi,j,s=\e>{N_\e}{A}{
        \tube{j=r}{M}
      }
    }
  }{
    \IsTm<\Psi>{
      \com{j.A}{r}{r'}{M}{
        \etcsys[\e]{s}{j.N_\e}
      }
    }{
      \dsubst{A}{r'}{j}
    }
  }
  \and
  \Rule[heterogeneous composition boundary]{
  }{
    \EqTm{
      \com{j.A}{r}{r}{M}{
        \etcsys[\e]{s}{j.N_\e}
      }
    }{M}{
      \dsubst{A}{r}{j}
    }
    \\\\
    \EqTm{
      \com{j.A}{r}{r'}{M}{
        \etcsys[\e']{\e}{j.N_{\e'}}
      }
    }{
      \dsubst{N_\e}{r'}{j}
    }{
      \dsubst{A}{r'}{j}
    }
  }
  \and
  \Rule[lift composition]{
  }{
    \EqTm{
      \hcom{\TyLift{k}{l}{A}}{r}{r'}{M}{
        \etcsys[\e]{i}{j.N_\e}
      }
    }{
      \hcom{A}{r}{r'}{M}{
        \etcsys[\e]{i}{j.N_\e}
      }
    }{
      \TyLift{k}{l}{A}
    }
  }
  \and
  \Rule[lift type composition]{
  }{
    \EqTy{
      \hcom{\Univ[l]}{r}{r'}{
        \TyLift{k}{l}{A}
      }{
        \etcsys[\e]{i}{
          j.\TyLift{k}{l}{B_\e}
        }
      }
    }{
      \TyLift{k}{l}{
        \hcom{\Univ[k]}{r}{r'}{A}{
          \etcsys[\e]{i}{
            j.B_\e
          }
        }
      }
    }
  }
  \and
  \Rule[pair composition computation (1)]{
    \Def{H}{
      \hcom{A}{r}{r'}{\fst{M}}{\etcsys[\e]{i}{j.\fst{N_\e}}}
    }
  }{
    \EqTm{
      \fst{
        \hcom{\sigmacl{x}{A}{B}}{r}{r'}{M}{\etcsys[\e]{i}{j.N_\e}}
      }
    }{H}{A}
  }
  \and
  \Rule[pair composition computation (2)]{
    \Def<\Psi,k>{\widetilde{M_1}[k]}{
      \hcom{A}{r}{k}{\fst{M}}{\etcsys[\e]{i}{j.\fst{N_\e}}}
    }
    \\
    \Def<\Psi>{H}{
      \com{k.\subst{B}{\widetilde{M_1}[k]}{x}}{r}{r'}{\snd{M}}{\etcsys[\e]{i}{j.\snd{N_\e}}}
    }
  }{
    \EqTm<\Psi>{
      \snd{
        \hcom{\sigmacl{x}{A}{B}}{r}{r'}{M}{\etcsys[\e]{i}{j.N_\e}}
      }
    }{H}{
      \subst{B}{\widetilde{M_1}[r']}{x}
    }
  }
  \and
  \Rule[pair type composition]{
    \Def<\Psi,k>[\Gamma]{\widetilde{A}[k]}{
      \hcom{\Univ}{r}{k}{A}{
        \etcsys[\e]{i}{j.A_\e}
      }
    }
    \\\\
    \Def<\Psi,j>[\Gamma,x:\widetilde{A}[r']]{\widetilde{x}[j]}{
      \coe{k.\widetilde{A}[k]}{r'}{j}{x}
    }
    \\\\
    \Def<\Psi>[\Gamma,x:\widetilde{A}[r']]{\widetilde{B}}{
      \hcom{\Univ}{r}{r'}{
        \subst{B}{\widetilde{x}[r]}{x}
      }{
        \etcsys[\e]{i}{
          j.\subst{B_\e}{
            \widetilde{x}[j]
          }{x}
        }
      }
    }
  }{
    \EqTy<\Psi>[\Gamma]{
      \hcom{\Univ}{r}{r'}{
        \parens{\sigmacl{x}{A}{B}}
      }{
        \etcsys[\e]{i}{
          j.\sigmacl{x}{A_\e}{B_\e}
        }
      }
    }{
      \sigmacl{x}{\widetilde{A}[r']}{\widetilde{B}}
    }[l]
  }
  \and
  \Rule[function composition computation]{
    \Def{H}{\hcom{\subst{B}{N}{x}}{r}{r'}{\app{M}{N}}{\etcsys[\e]{i}{j.\app{M_\e}{N}}}}
  }{
    \EqTm{
      \app{
        \parens{
          \hcom{\picl{x}{A}{B}}{r}{r'}{M}{\etcsys[\e]{i}{j.M_\e}}
        }
      }{N}
    }{H}{\picl{x}{A}{B}}
  }
  \and
  \Rule[function type composition]{
    \Def<\Psi,k>{\widetilde{A}[k]}{
      \hcom{\Univ}{r}{k}{A}{
        \etcsys[\e]{i}{j.A_\e}
      }
    }
    \\\\
    \Def<\Psi,j>[\Gamma,x:\widetilde{A}[r']]{\widetilde{x}[j]}{
      \coe{k.\widetilde{A}[k]}{r'}{j}{x}
    }
    \\\\
    \Def<\Psi>[\Gamma,x:\widetilde{A}[r']]{\widetilde{B}}{
      \hcom{\Univ}{r}{r'}{
        \subst{B}{\widetilde{x}[r]}{x}
      }{
        \etcsys[\e]{i}{
          j.\subst{B_\e}{
            \widetilde{x}[j]
          }{x}
        }
      }
    }
  }{
    \EqTy<\Psi>[\Gamma]{
      \hcom{\Univ}{r}{r'}{
        \parens{\picl{x}{A}{B}}
      }{
        \etcsys[\e]{i}{
          j.\picl{x}{A_\e}{B_\e}
        }
      }
    }{
      \picl{x}{\widetilde{A}[r']}{\widetilde{B}}
    }[l]
  }
  \and
  \Rule[equality composition computation]{
    \Def{H}{
      \hcom{\dsubst{A}{s}{i}}{r}{r'}{P(s)}{
        \etcsys[\e]{k}{j.Q_\e(s)}
      }
    }
  }{
    \EqTm{
      \parens{\hcom{\Eq{i.A}{N_0}{N_1}}{r}{r'}{P}{\etcsys[\e]{k}{j.Q_\e}}}(s)
    }{H}{
      \dsubst{A}{s}{i}
    }
  }
  \and
  \Rule[equality type composition]{
    \Def<\Psi,j,i>{\widetilde{A}[j,i]}{
      \hcom{\Univ}{r}{j}{A}{
        \etcsys[\e]{k}{j.A_\e}
      }
    }
    \\
    \Def<\Psi>{\widetilde{M}}{
      \com{j.\widetilde{A}[j,r]}{r}{r'}{M}{
        \etcsys[\e]{k}{j.M_\e}
      }
    }
    \\
    \Def<\Psi>{\widetilde{N}}{
      \com{j.\widetilde{A}[j,r']}{r}{r'}{N}{
        \etcsys[\e]{k}{j.N_\e}
      }
    }
  }{
    \EqTy<\Psi>{
      \hcom{\Univ}{r}{r'}{
        \Eq{i.A}{M}{N}
      }{
        \etcsys[\e]{k}{
          j.\Eq{i.A_\e}{M_\e}{N_\e}
        }
      }
    }{
      \Eq{i.\widetilde{A}[r',i]}{\widetilde{M}}{\widetilde{N}}
    }[l]
  }
  \and
  \Rule[boolean type composition]{
  }{
    \EqTy{
      \hcom{\Univ[k]}{r}{r'}{\bool}{
        \etcsys{i}{j.\bool}
      }
    }{
      \bool
    }[l]
  }
  \and
  \Rule[universe type composition]{
  }{
    \EqTy{
      \hcom{\Univ[k']}{r}{r'}{\Univ}{
        \etcsys{i}{j.\Univ}
      }
    }{
      \Univ
    }[l]
  }
\end{ruleblock}

\section{Algebraic model theory}\label{sec:cwf-structure}

We begin by giving a general formulation of a category with families (cwf)
which has the structure of \XTT. We work in a constructive set theory extended
by Grothendieck's Axiom of Universes: every set is contained in some
Grothendieck Universe; this axiom induces a ordinal-indexed hierarchy of
Grothendieck universes $\SemUniv[k]:\SET$. Concretely, we will be using the chain of
inclusions $\SemUniv[0] \in ... \in \SemUniv[n] \in ... \in \SemUniv[\omega]$.


Let $\Cube$ be the Cartesian cube category, the free strictly associative
Cartesian category generated by an interval; concretely, its objects are
dimension contexts $\IsDimCx{\Psi}$, with morphisms given by substitutions
between them. For the sake of clarity, we choose to work with the explicit
syntactic presentation where $\Psi$ is a list of named variables.  Next, let
$\AugCube$ be the \emph{augmented Cartesian cube category}, which freely adjoins an
initial object $\bot$ to $\Cube$. Using the initial object, we can see that
$\AugCube$ has all equalizers, and that the evident functor $\Cube\rTo\AugCube$
is left exact (in other words, the new limits coincide with the old ones where
they existed). We will write $\Psi_0,\Psi,\Psi_1 \rTo^{\Wk{\Psi}} \Psi_0,\Psi_1$
for the obvious weakening projection in $\AugCube$.

\subsection{Algebraic Cumulative Cwfs}
We employ a variation on the notion of
categories with families (cwfs) \Cite{dybjer:1996} suitable for modeling dependent type theory
with a hierarchy of universes \`a la Russell which is cumulative in an
algebraic sense (i.e.\ without subtyping).\footnote{We emphasize
that although the data of a cwf contains the data of a category, we are doing an
\emph{algebraic} model theory for type theory in which (e.g.) the initial
object is determined up to isomorphism rather than up to equivalence.}

\subsubsection{Basic cwf structure: contexts, types, elements}

Here, we develop the basic \emph{judgmental} structure of a model of \XTT{},
prior to requiring the existence of various connectives.

\paragraph{Category of contexts}
\emph{An algebraic cumulative cwf} begins with (the data of) a category $\Cx$
of \emph{contexts} $\Gamma$, with morphisms $\Gamma\rTo\Delta$ interpreting
substitutions. We require there to be a terminal context $\cdot$ such that for
any $\Gamma$ there is a unique substitution $\Gamma\rTo\cdot$.

\begin{notation}[Yoneda isomorphism]
  For a presheaf $X:\Psh{\mathcal{C}}$ on any category $\mathcal{C}$, we will use the following notations for the
  components of the Yoneda isomorphism:
  \begin{diagram}
    X(\Gamma) & \pile{\rTo^{\YoEmb}\\\lTo_{\YoUnemb}} & \Hom[\Psh{\Cx}]{\Yo{\Gamma}}{X}
  \end{diagram}
\end{notation}

\paragraph{Cubical structure}

We furthermore require our category of contexts be equipped with a
\emph{split} fibration $\Cx\rProjto^{\CubeFib}\AugCube$ which preserves the terminal object. To be precise, we equip
$\Cx$ with the data of a functor $\Cx\rTo^{\vert\CubeFib\vert}\AugCube$ and a
splitting $\Cleave$ for $\vert\CubeFib\vert$. It is important to note
that we intend this structure to be preserved on the nose in homomorphisms of
structured cwfs.

\begin{figure}
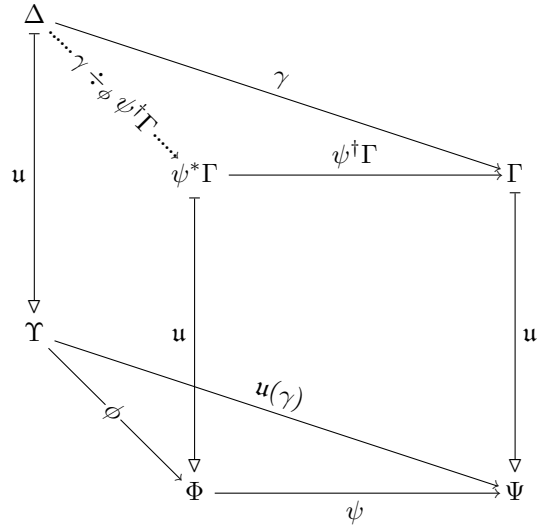


  \begin{diagram}
    \Delta && \\
    &\rdDotsto~{\CartUniMap<\phi>{\gamma}{\Cleave{\psi}{\Gamma}}}\rdTo(6,2)^\gamma\\
    \dProjMapsto^{\CubeFib} && \ReIx{\psi}{\Gamma} &&  \rTo^{\Cleave{\psi}{\Gamma}}  && \Gamma\\
    \\
    \Upsilon&&\dProjMapsto^{\CubeFib}&&&& \dProjMapsto_{\CubeFib} \\
    & \rdTo~\phi\rdTo(6,2)^{\CubeFib{\gamma}} &          &          &            &&       \\
    &        &    \Phi     &          &   \rTo_\psi   &        &   \Psi   \\
  \end{diagram}

  \caption{A schematic illustration of the situation induced by a split
  fibration $\CubeFib$; $\Cleave{\psi}{\Gamma}$ is the Cartesian lifting of
  $\psi$, and the dotted arrow is the one induced by its universal property.}\label{fig:prone-lifting}

\end{figure}

\begin{notation}[Split fibration]\label{notation:splitting}
  We impose the following notations for working with the split fibration
  $\CubeFib$, supplemented by \cref{fig:prone-lifting}.
  \begin{enumerate}

    \item The fiber of $\CubeFib$ over $\Psi$ is written $\Cx_\Psi$;
      explicitly, this is the subcategory of $\Cx$ whose objects are taken to
      $\Psi$ by $\CubeFib$.

    \item The splitting $\Cleave$ takes a dimension substitution
      $\Psi'\rTo^\psi\Psi$ and an object $\Gamma:\Cx_\Psi$ to a morphism
      $\ReIx{\psi}{\Gamma}\rTo^{\Cleave{\psi}{\Gamma}}\Gamma$ such that
      $\CubeFib{\Cleave{\psi}{\Gamma}} = \psi$.

    \item Given a Cartesian morphism $\Delta\rTo^{\gamma_\Delta}\Gamma$ and an
      arrow $\Xi\rTo^{\gamma_\Xi}\Gamma$, along with an arrow in the base
      category $\CubeFib{\Xi}\rTo^\psi\CubeFib{\Gamma}$ such that
      $\CubeFib{\gamma_\Xi} = \CubeFib{\gamma_\Delta}\circ \psi$, the universal
      property of the Cartesian morphism guarantees a unique map
      $\CartUniMap<\psi>{\gamma_\Xi}{\gamma_\Delta}$ which lies over $\psi$,
      such that $\gamma_\Xi = \gamma_\Delta\circ\parens*{\CartUniMap<\psi>{\gamma_\Xi}{\gamma_\Delta}}$.


    \item Given $\Delta\rTo^\gamma\Gamma$ (not necessarily vertical) and $\Psi$
      disjoint from $\CubeFib{\Delta}$ and $\CubeFib{\Gamma}$, we write
      $\ReIx{\Wk{\Psi}}{\Delta}\rTo^{\Lift{\Wk{\Psi}}{\gamma}}\ReIx{\Wk{\Psi}}{\Gamma}$
      for the morphism
      $\CartUniMap<\CubeFib{\gamma}\star{}\Id_\Psi>{\parens*{\gamma\circ\Cleave{\Wk{\Psi}}{\Delta}}}{\Cleave{\Wk{\Psi}}{\Gamma}}$
      where $\CubeFib{\gamma}\star\Id_\Psi$ is the horizontal composite
      $\CubeFib{\Delta},\Psi\to\CubeFib{\Gamma},\Psi$.

%

    \item For a presheaf $\mathcal{F} : \Psh{\Cx}$, we will write
      $\mathcal{F}\parens{\Gamma} \rTo{\ReIxCleave{\psi}[\Gamma]}
      \mathcal{F}\parens{\ReIx{\psi}{\Gamma}}$ for the action of $\mathcal{F}$
      on $\ReIx{\psi}{\Gamma}\rTo^{\Cleave{\psi}{\Gamma}}\Gamma$.

  \end{enumerate}
\end{notation}

\NewDocumentCommand\CxLvl{D<>{\Cx}}{#1_{\Lev}}

\paragraph{Contexts and levels}

$\Cx$ is the category of contexts; but the judgments of \XTT{} are also
parameterized in a universe level. Therefore, we will need to work in
presheaves not on $\Cx$ but rather on $\CxLvl\triangleq\Cx\times\Lev$. Recall
that $\Lev$ is the category of \emph{universe levels} and is defined by
$\Lev = \OpCat{\omega}$.

\paragraph{Types and elements}

Next, we require a presheaf $\CwfTy$ in $\Psh{\CxLvl}$, yielding at each
$(\Gamma,k)$ the set $\CwfTy[k]{\Gamma}$ of types of level $k$ over $\Gamma$;
we require each fiber $\CwfTy[k]{\Gamma}$ to be $(k+1)$-small in the ambient
set theory.  As a matter of notation, we write $\TyLift{k}{l}{A}$ for level
restrictions, mirroring the concrete syntax of \XTT.

Then, we require a \emph{dependent presheaf} $\CwfEl:\Psh{\int\CwfTy}$, writing
$\CwfEl{\Gamma}{A}$ for the fiber over $A\in\CwfTy[k]{\Gamma}$; in order to
justify the rules which make elements of types and their level liftings
definitionally interchangeable, we require that functorial action of maps
$(\Gamma,l,\TyLift{k}{l}A) \rTo (\Gamma,k,A)$ in $\int\CwfTy$ on the dependent
presheaf $\CwfEl$ must be strict identities. Consequently, we have
$\CwfEl{\Gamma}{A} = \CwfEl{\Gamma}{\TyLift{k}{l}{A}}$, following \Cite{sterling:2018:gat}.

By taking a dependent sum, it is possible to regard $\CwfEl$ as an element of
the slice $\Psh{\CxLvl}/\CwfTy$; this perspective will be profitable when defining
the notion of a \emph{context comprehension}.

\paragraph{Context comprehension}

Writing $\CwfEl\rTo^{\pi}\CwfTy:\Psh{\CxLvl}$ for the evident
projection of types from elements, we follow
\Cite{fiore:2012,awodey:2018:natural-models} in requiring that $\pi$ be
equipped with a choice of representable pullbacks along natural transformations
out of representable objects (contexts):
\begin{diagram}
  \DPullback*{
    \Yo*{\Gamma.A, k}
  }{\CwfEl}{\Yo*{\Gamma, k}}{\CwfTy}{\pi}{\YoEmb{A}}{\YoEmb{\Var}}{\Yo{\Proj}}
\end{diagram}
We require the equation $\Gamma.A = \Gamma.\TyLift{k}{l}{A}$ in the objects of $\Cx$.

Given a substitution
$\Delta\rTo^\gamma\Gamma$ and an element
$N\in\CwfEl{\Delta}{\ReIx{\gamma}{A}}$, we can form the extended substitution
$\Snoc{\gamma}{N}$ using the universal property of the pullback:
\begin{diagram}
  \DRow \Yo*{\Delta,k}
  \\
  \DBlock
  \DRow \rdDotsto~{\Yo{\Snoc{\gamma}{N}}}\rdTo(4,2)^{\YoEmb{N}}\rdTo(2,4)_{\Yo{\gamma}}
  \\
  \DBlock
  \DRow \JONpbk{\Yo*{\Gamma.A,k}} & \rTo_{\YoEmb{\Var}} & \CwfEl
  \\
  \DRow \dTo_{\Yo{\Proj}} & & \dTo_{\pi}
  \\
  \DRow \Yo*{\Gamma,k} & \rTo_{\YoEmb{A}} & \CwfTy
  \DDone
\end{diagram}

\paragraph{Constraint comprehension}

Let $\CwfDim:\Psh{\Cx}$ be the presheaf of dimensions, taking a context
$\Gamma:\Cx$ to the set of dimensions $\CubeFib{\Gamma}\rTo^{r}\brackets{i}$.
Then, define $\CwfProp=\CwfDim\times\CwfDim$, writing
$(r=s)\in\CwfProp{\Gamma}$ for the pair of $r,s\in\CwfDim{\Gamma}$; we will
follow the syntax of \XTT{} in using $\xi$ to range over an element of
$\CwfProp{\Gamma}$.

\begin{lemma}
  The split fibration $\CubeFib$ forces the
  diagonal $\CwfDim\rTo^\delta\CwfProp$ to be representable in the same sense as
  above; schematically:
  \begin{diagram}
    \DRow \Yo{\Delta}
    \\
    \DBlock
    \DRow \rdDotsto~{\Yo*{\gamma.\xi}}\rdTo(4,2)^{\YoEmb{s}}\rdTo(2,4)_{\Yo{\gamma}}
    \\
    \DBlock
    \DRow \JONpbk{\Yo*{\Gamma.\xi}} & \rTo & \CwfDim
    \\
    \DRow \dTo_{\Yo{\Wk{\xi}}} & & \dTo_{\delta}
    \\
    \DRow \Yo{\Gamma} & \rTo_{\YoEmb{\xi}} & \CwfProp
    \DDone
  \end{diagram}
\end{lemma}

\begin{proof}
  The context $\Gamma.\xi$ is obtained from the equalizer of $\xi$ in $\AugCube$, using the splitting of the fibration
  $\CubeFib$:
   \begin{diagram}
     \Gamma.\xi & \rTo^{\Cleave{\Wk{\xi}}{\Gamma}} & \Gamma
     \\
     \dProjMapsto_{\CubeFib}
     &&
     \dProjMapsto_{\CubeFib}
     \\
     \CubeFib{\Gamma}.\xi & \rTo^{\Wk{\xi}} & \CubeFib{\Gamma}
     & \pile{\rTo^{\xi_0}\\\rTo_{\xi_1}}
     & \brackets{i}
   \end{diagram}

   We can see that $\Yo*{\Gamma.\xi}$ is indeed the pullback below:
   \begin{diagram}
     \DBlock
     \DRow \JONpbk{\Yo*{\Gamma.\xi}} & \rTo^{\YoEmb{\xi_0\circ\Wk{\xi}}} & \CwfDim
     \\
     \DRow \dTo^{\Yo*{\Cleave{\Wk{\xi}}{\Gamma}}} & & \dTo_{\delta}
     \\
     \DRow \Yo{\Gamma} & \rTo_{\YoEmb{\xi}} & \CwfProp
     \DDone
   \end{diagram}

   To see that the diagram commutes, we just verify that $\YoEmb{\xi_0\circ\Wk{\xi}=\xi_0\circ\Wk{\xi}} =
   \YoEmb{\xi}\circ\Yo*{\Cleave{\Wk{\xi}}{\Gamma}}$, which is the same as to
   say that $\xi_0\circ\Wk{\xi} = \xi_1\circ\Wk{\xi}$; but this is just the
   fact that $\Wk{\xi}$ is the equalizer of $\xi_0,\xi_1$. Next, we check the
   universal property of the pullback; because limits in $\Psh{\Cx}$ are formed
   pointwise (as in all presheaf categories), it suffices to check universality
   at representable objects only.

   Fix $\Yo{\Delta}\rTo^{\Yo{\gamma}}\Yo{\Gamma}$ and
   $\Yo{\Delta}\rTo^{\YoEmb{s}}\CwfDim$ such that $\delta\circ\YoEmb{s} =
   \YoEmb{\xi}\circ\Yo{\gamma}$; we need to choose a unique morphism
   $\Yo{\Delta}\rDotsto^{\Yo{\eta}}\Yo{\Gamma.\xi}$ such that
   $\Yo*{\Cleave{\xi}{\Gamma}}\circ\Yo{\eta} = \Yo{\gamma}$ and
   $\YoEmb{\xi_0}\circ\Yo{\eta} = \YoEmb{s}$. Unraveling the Yoneda paperwork,
   we have assumed that $\xi_0\circ\CubeFib{\gamma}=s=\xi_1\circ\CubeFib{\gamma}$
   and we want to find $\Delta\rDotsto^\eta\Gamma.\xi$ such that the following
   triangles commute in $\Cx$ and $\AugCube$ respectively:
   \begin{mathpar}
     (1) \qquad
     \begin{diagram}
       \Delta &\rTo^{\eta} & \Gamma.\xi
       \\
       &\rdTo_{\gamma} &\dTo_{\Cleave{\Wk{\xi}}{\Gamma}}
       \\
       &&\Gamma
     \end{diagram}
     \and
     \begin{diagram}
       \CubeFib{\Delta}
       \\
       \dTo^{\CubeFib{\eta}} &\rdTo^{s}
       \\
       \CubeFib{\Gamma}.\xi & \rTo_{\xi_0\circ\Wk{\xi}} & \brackets{i}
     \end{diagram}
     \qquad
     (2)
   \end{mathpar}

   First, observe that because
   $\xi_0\circ\CubeFib{\gamma}=\xi_1\circ\CubeFib{\gamma}$, the universal
   property of the equalizer guarantees a unique map
   $\CubeFib{\Delta}\rTo^{\CubeFib{\gamma}.\xi}\CubeFib{\Gamma}.\xi$ with the
   same property:
   \begin{diagram}
     \CubeFib{\Gamma}.\xi & \rTo^{\Wk{\xi}} & \CubeFib{\Gamma} & \pile{\rTo^{\xi_0}\\\rTo_{\xi_1}} & \brackets{i}
     \\
     &\luDotsto~{\CubeFib{\gamma}.\xi} & \uTo_{\CubeFib{\gamma}}
     \\
     && \CubeFib{\Delta}
   \end{diagram}

   Using $\CubeFib{\gamma}.\xi$ from above, we obtain $\eta$ from the universal property of
   the Cartesian lifting $\Gamma.\xi\rTo^{\Cleave{\Wk{\xi}}{\Gamma}}\Gamma$:
   \begin{mathpar}
     \begin{diagram}
       \Delta
       \\
       & \rdDotsto_{\eta} \rdTo(4,2)^{\gamma}
       \\
       &&\Gamma.\xi &\rTo_{\Cleave{\Wk{\xi}}{\Gamma}} &\Gamma
     \end{diagram}
     \mbox{lying over}
     \begin{diagram}
       \CubeFib{\Delta}
       \\
       & \rdTo_{\CubeFib{\gamma}.\xi} \rdTo(4,2)^{\CubeFib{\gamma}}
       \\
       &&\CubeFib{\Gamma}.\xi &\rTo_{\Wk{\xi}} & \CubeFib{\Gamma}
     \end{diagram}
   \end{mathpar}

   We therefore see immediately that triangle (1) commutes; to see that
   triangle (2) commutes, we calculate: $\xi_0\circ\Wk{\xi}\circ\CubeFib{\eta}
   = \xi_0\circ\Wk{\xi}\circ\CubeFib{\gamma}.\xi = \xi_0\circ\CubeFib{\gamma} =
   s$.
\end{proof}

In order to model the collapse of the typing and equality judgments under the
contraint $0=1$ in \XTT{}'s syntax, we will \emph{require} that the contexts
$\Gamma.0=1$ and $\Gamma.1=0$ are \emph{initial} in $\Cx$; this implies
initiality in the fibration $\CubeFib$, because the equalizer
$\CubeFib{\Gamma}.0=1$ is the initial object in $\AugCube$.

\begin{notation}[Constraint weakening]
  Because we will use it frequently, we will often write $\Gamma.\xi
  \rTo^{\Wk{\xi}} \Gamma$ for the Cartesian lifting
  $\Cleave{\Wk{\xi}}{\Gamma}$.
\end{notation}

\begin{notation}[Constraint lifting]
  When $\Delta\rTo^\gamma\Gamma$, we write
  $\Delta.\ReIx{\gamma}{\xi}\rTo^{\Lift{\Wk{\xi}}{\gamma}}\Gamma.\xi$
  for $(\gamma\circ\Wk*{\ReIx{\gamma}{\xi}}).\xi$.
\end{notation}

We implicitly lift everything to do with dimensions and constraints into
$\Psh{\CxLvl}$, by reindexing silently along the projection $\CxLvl\rProjto\Cx$.

\paragraph{Boundary separation}
\NewDocumentCommand\Cov{g}{\mathbf{K}_{\partial}\IfValueT{#1}{\parens{#1}}}
\NewDocumentCommand\Two{}{\mathbf{2}}

To characterize models of \XTT{}, we need to ensure that every type and every
element is totally determined by its boundary with respect to the
dimension context. A simple way to state this requirement is as a
\emph{separation} condition with respect to a particular coverage on the
category of contexts $\Cx$. We define the coverage $\Cov$ on $\Cx$ by taking
the constraint weakenings $\braces{\Wk*{r=\e}}_\e$ to constitute a covering family for each
dimension $r$:
\begin{mathpar}
  \Cov{\Gamma} \ni
  \braces*{
      \Gamma.r=\e \rTo^{\Wk*{r=\e}} \Gamma
  }_{\e\in\Two}
  \quad
  \parens{r\in\CwfDim{\Gamma}}
\end{mathpar}

\begin{lemma}
  The family of sets $\Cov$ is a coverage on $\Cx$.
\end{lemma}

\begin{proof}
  To see that $\Cov$ is in fact a coverage, we fix $\Delta\rTo^\gamma\Gamma$ and
  observe that any covering family $\braces{\Gamma.\xi_\e\rTo^{\Wk*{\xi_\e}}\Gamma}_{\e\in\Two}$ can
  be pulled back to obtain a new covering family $\braces{\Delta.\ReIx{\gamma}{\xi_\e}
  \rTo^{\Wk*{\ReIx{\gamma}{\xi_\e}}}\Delta}_{\e\in\Two}$ such that each composite
  $\gamma\circ\Wk*{\ReIx{\gamma}{\xi_\e}}$ factors through
  $\Gamma.\xi_\e\rTo^{\Wk*{\xi_\e}}\Gamma$:
  \begin{diagram}
    \Delta.\ReIx{\gamma}{\xi_\e}
    & \rDotsto^{{\color{Red}\Lift{\Wk*{\xi_\e}}{\gamma}}} & \Gamma.\xi_\e
    \\
    \dTo^{\Wk*{\ReIx{\gamma}{\xi_\e}}} & & \dTo_{\Wk*{\xi_\e}}
    \\
    \Delta & \rTo_{\gamma} & \Gamma
  \end{diagram}
\end{proof}

This coverage lifts immediately along the projection $\CxLvl\rProjto\Cx$ to a
coverage on $\CxLvl$; because it will not result in ambiguity, we leave this
lifting implicit.

\begin{definition}[Separation]\label{def:separation}

  Given a coverage $\mathbf{K}$ on a category $\mathcal{C}$, a presheaf $F:\Psh{\Cx}$ is
  $\mathbf{K}$-\emph{separated} when, for any elements $a,b\in F(\Gamma)$ and
  covering family $\braces{\Delta_i\rTo^{\gamma_i}\Gamma}_{i\in I} \in\mathbf{K}(\Gamma)$, if we have
  $\ReIx{\gamma_i}{a}=\ReIx{\gamma_i}{b}\in F(\Delta_i)$ for
  each $i\in I$, then $a=b\in F(\Gamma)$.

\end{definition}

\begin{definition}[Boundary separation]\label{def:boundary-separation}

  We say that a cwf has \emph{boundary separation} when the presheaves
  $\CwfTy,\CwfEl:\Psh{\CxLvl}$ are $\Cov$-separated.

\end{definition}

\subsubsection{Kan operations: coercion and composition}

\begin{definition}[Regular coercion structure]\label{def:regular-coercion-structure}

  A cwf has \emph{regular coercion structure} iff for every type
  $A\in\CwfTy[n]{\ReIx{\Wk{\imath}}{\Gamma}}$ over $\Psi,i$ and dimensions
  $r,r'\in\CwfDim{\Gamma}$ and element
  $M\in\CwfEl{\Gamma}{\ReIxCleave*{r/i}[\ReIx{\Wk{\imath}}{\Gamma}]{A}}$, there is
  an element
  $\CwfCoe{i.A}{r}{r'}{M}\in\CwfEl{\Gamma}{\ReIxCleave*{r'/i}[\ReIx{\Wk{\imath}}{\Gamma}]{A}}$
  which has the following properties:

  \begin{itemize}
    \item \emph{Adjacency.} If $r=r'$ then $\CwfCoe{i.A}{r}{r'}{M} = M$.
    \item \emph{Regularity.} If $A=\ReIxCleave{\Wk{\imath}}[\Gamma]{A'}$ for some $A'\in\CwfTy[n]{\Gamma}$, then $\CwfCoe{i.A}{r}{r'}{M}=M$.

    \item \emph{Level restriction.} The equation
      $\CwfCoe{i.\TyLift{k}{l}{A}}{r}{r'}{M} = \CwfCoe{i.A}{r}{r'}{M}$.

    \item \emph{Naturality.} For $\Delta\rTo^\gamma\Gamma$ we have
      $\ReIx{\gamma}{\CwfCoe{i.A}{r}{r'}{M}} =
      \CwfCoe{i.\ReIx*{\Lift{\Wk{\imath}}{\gamma}}{A}}{\ReIx{\gamma}{r}}{\ReIx{\gamma}{r'}}{\ReIx{\gamma}{M}}$.

  \end{itemize}
\end{definition}

\begin{definition}[Regular homogeneous composition structure]\label{def:regular-hcom-structure}

  We say that $\Cx$ has \emph{regular homogeneous composition structure} iff,
  for each $A\in\CwfTy[n]{\Gamma}$ and $r,r',s\in\CwfDim{\Gamma}$ and
  $M\in\CwfEl{\Gamma}{A}$ and
  $\etc{N_\e\in\CwfEl{\ReIx{\Wk{\jmath}}{(\Gamma.s=\e)}}{\ReIxCleave{\Wk{\jmath}}{\ReIx{
    \Wk*{s=\e}}{A}}}}$ for fresh $j$ such that $\etc{\ReIxCleave*{r/j}{N_\e}=M}$, we have an element
  $\CwfHcom{A}{r}{r'}{M}{\etcsys[\e]{s}{j.N_\e}}$ satisfying the following
  conditions:

  \begin{itemize}
    \item \emph{Adjacency.} If $r=r'$ then $\CwfHcom{A}{r}{r'}{M}{\etcsys[\e]{s}{j.N_\e}}=M$; moreover, if $s=\e$, then $\CwfHcom{A}{r}{r'}{M}{\etcsys[\e]{s}{j.N_\e}} = \ReIxCleave*{r'/j}{N_\e}$.
    \item \emph{Regularity.} If we have $\etc{N_\e = \ReIxCleave{\Wk{\jmath}}{N'_\e}}$ for
      some $N'$, then we have the equation $\CwfHcom{A}{r}{r'}{M}{\etcsys[\e]{s}{j.N_\e}} = M$.
    \item \emph{Naturality.} For $\Delta\rTo^\gamma\Gamma$, we require the following naturality conditions:
      \begin{mathpar}
        \ReIx{\gamma}{
          \CwfHcom{A}{r}{r'}{M}{\etcsys[\e]{s}{j.N_\e}}
        } =
        \CwfHcom{
          \ReIx{\gamma}{A}
        }{
          \ReIx{\gamma}{r}
        }{
          \ReIx{\gamma}{r'}
        }{
          \ReIx{\gamma}{M}
        }{
          j.
          \ReIx*{
            \Lift{\Wk{\jmath}}{
              \Lift{\Wk*{s=\e}}{\gamma}
            }
          }{N_\e}
        }
      \end{mathpar}

  \end{itemize}

\end{definition}

\begin{notation}[Heterogeneous composition]

  When a cwf has coercion and homogeneous composition, we write its
  \emph{heterogeneous composition} using the following definitional extension:
  \begin{align*}
    &\CwfCom{i.A}{r}{r'}{M}{
      \etcsys[\e]{s}{i.N_\e}
    }
    \\
    &\triangleq
    \CwfHcom{
      \ReIxCleave*{r'/j}{A}
    }{r}{r'}{
      \parens*{\CwfCoe{i.A}{r}{r'}{M}}
    }{
      \etcsys[\e]{s}{
        i.\CwfCoe{i.A}{i}{r'}{N_\e}
      }
    }
  \end{align*}
\end{notation}

\subsubsection{Closure under type-theoretic connectives}

\begin{definition}[Booleans]\label{def:booleans}
  A cwf has the \emph{booleans} when it is equipped with the following structure:
  \begin{itemize}

    \item \emph{Formation.} Types $\CwfBool\in\CwfTy[n]{\Gamma}$ for all $\Gamma,n$.

    \item \emph{Introduction.} Elements $\CwfTrue\in\CwfEl{\Gamma}{\CwfBool}$ and $\CwfFalse\in\CwfEl{\Gamma}{\CwfBool}$.

    \item \emph{Elimination.} If $C\in\CwfTy[n]{\Gamma.\CwfBool}$ and $M\in\CwfEl{\Gamma}{\CwfBool}$ and $N_0\in\CwfEl{\Gamma}{\ReIx{\Snoc{\Id}{\CwfTrue}}{C}}$ and $N_1\in\CwfEl{\Gamma}{\ReIx{\Snoc{\Id}{\CwfFalse}}{C}}$, an element $\CwfIf{C}{M}{N_0}{N_1}\in\CwfEl{\Gamma}{\ReIx{\Snoc{\Id}{M}}{C}}$.

    \item \emph{Computation.} The following equations:
      \begin{mathpar}
        \CwfIf{C}{\CwfTrue}{N_0}{N_1} = N_0
        \and
        \CwfIf{C}{\CwfFalse}{N_0}{N_1} = N_1
      \end{mathpar}

    \item \emph{Level restriction.} The following two equations:
      \begin{mathpar}
        \TyLift{k}{l}{\CwfBool} = \CwfBool
        \and
        \CwfIf{\TyLift{k}{l}{C}}{M}{N_0}{N_1} =
        \CwfIf{C}{M}{N_0}{N_1}
      \end{mathpar}

    \item \emph{Naturality.} For $\Delta\rTo^\gamma\Gamma$, the following naturality equations:
      \begin{mathpar}
        \ReIx{\gamma}{\CwfBool} = \CwfBool
        \and
        \ReIx{\gamma}{\CwfTrue} = \CwfTrue
        \and
        \ReIx{\gamma}{\CwfFalse} = \CwfFalse
        \and
        \ReIx{\gamma}{\CwfIf{C}{M}{N_0}{N_1}} =
        \CwfIf{
          \ReIx{\CwfLift{\gamma}}{C}
        }{
          \ReIx{\gamma}{M}
        }{
          \ReIx{\gamma}{N_0}
        }{
          \ReIx{\gamma}{N_1}
        }
      \end{mathpar}

  \end{itemize}
\end{definition}

\begin{definition}[Dependent function types]\label{def:dependent-function-types}

  A cwf has \emph{dependent function types} when it is equipped with the following
  structure:
  \begin{itemize}

    \item \emph{Formation.} For each type $A\in\CwfTy[n]{\Gamma}$ and family
      $B\in\CwfTy[n]{\Gamma.A}$, a type $\CwfPi{A}{B}\in\CwfTy[n]{\Gamma}$.

    \item \emph{Introduction.} For each $M\in\CwfEl{\Gamma.A}{B}$, an element
      $\CwfLam[A][B]{M}\in\CwfEl{\Gamma}{\CwfPi{A}{B}}$.

    \item \emph{Elimination.} For each $M\in\CwfEl{\Gamma}{\CwfPi{A}{B}}$ and
      $N\in\CwfEl{\Gamma}{A}$ an element
      $\CwfApp[A][B]{M}{N}\in\CwfEl{\Gamma}{\ReIx{\Snoc{\Id}{N}}{B}}$.

    \item \emph{Computation.} For each $M\in\CwfEl{\Gamma.A}{B}$ and
      $N\in\CwfEl{\Gamma}{A}$, the equation $\CwfApp[A][B]{\CwfLam[A][B]{M}}{N} =
      \ReIx{\Snoc{\Id}{N}}{M}$.

    \item \emph{Unicity.} For each $M\in\CwfEl{\Gamma}{\CwfPi{A}{B}}$, the following equation:
      \[
        M =
        \CwfLam[\ReIx{\Proj}{A}][\ReIx{\CwfLift{\Proj}}{B}]{
          \CwfApp[\ReIx{\Proj}{A}][\ReIx{\CwfLift{\Proj}}{B}]{\ReIx{\Proj}{M}}{\Var}
        }
      \]

    \item \emph{Level restriction.} The following equations:
      \begin{mathpar}
        \TyLift{k}{l}{\CwfPi{A}{B}} = \CwfPi{\TyLift{k}{l}{A}}{\TyLift{k}{l}{B}}
        \and
        \CwfLam[\TyLift{k}{l}{A}][\TyLift{k}{l}{B}]{M} = \CwfLam[A][B]{M}
        \and
        \CwfApp[\TyLift{k}{l}{A}][\TyLift{k}{l}{B}]{M}{N} = \CwfApp[A][B]{M}{N}
      \end{mathpar}

    \item \emph{Naturality.} We have the following naturality conditions for each
      $\Delta\rTo^\gamma\Gamma$:
      \begin{mathpar}
        \ReIx{\gamma}{\CwfPi{A}{B}} = \CwfPi{\ReIx{\gamma}{A}}{\ReIx{\CwfLift{\gamma}}{B}}
        \and
        \ReIx{\gamma}{\CwfLam[A][B]{M}} = \CwfLam[\ReIx{\Gamma}{A}][\ReIx{\CwfLift{\gamma}}{B}]{\ReIx{\CwfLift{\gamma}}{M}}
        \and
        \ReIx{\gamma}{\CwfApp[A][B]{M}{N}} = \CwfApp[\ReIx{\gamma}{A}][\ReIx{\CwfLift{\gamma}}{B}]{\ReIx{\gamma}{M}}{\ReIx{\gamma}{N}}
      \end{mathpar}

    \item \emph{Coercion.} When $\Gamma\rProjMapsto^\CubeFib \Psi,i$ and
      $\IsDim<\Psi>{r,r'}$ and
      $M\in\CwfEl{\ReIx*{r/i}{\Gamma}}{\ReIxCleave*{r/i}{\CwfPi{A}{B}}}$, we require the following
      equation:
      \begin{mathpar}
        \CwfCoe{i.\CwfPi{A}{B}}{r}{r'}{M}
        =
        \CwfLam[
          \ReIxCleave*{r'/i}{A}
        ][
          \ReIxCleave*{r'/i}{B}
        ]{
          \CwfCoe{i.
            \ReIx{
              \Snoc{\Id}{
                \CwfCoe{i.\ReIx{\Proj}{A}}{r'}{i}{\Var}
              }
            }{B}
          }{r}{r'}{
            \CwfApp[\ReIx{\Proj}{\ReIxCleave*{r/i}{A}}][\ReIx{\CwfLift{\Proj}}{\ReIxCleave*{r/i}{B}}]{
              \ReIx{\Proj}{M}
            }{\CwfCoe{i. \ReIx{\Proj}{A}}{r'}{r}{\Var}}
          }
        }
      \end{mathpar}

  \end{itemize}
\end{definition}

\begin{definition}[Dependent pair types]\label{def:dependent-pair-types}

  A cwf has \emph{dependent pair types} when it is equipped with the following
  structure:

  \begin{itemize}

    \item \emph{Formation.} For each type $A\in\CwfTy[n]{\Gamma}$ and family
      $B\in\CwfTy[n]{\Gamma.A}$, a type $\CwfSg{A}{B}\in\CwfTy[n]{\Gamma}$.

    \item \emph{Introduction.} For each $M\in\CwfEl{\Gamma}{A}$ and
      $N\in\CwfEl{\Gamma}{\ReIx{\Snoc{\Id}{M}}{B}}$, an element $\CwfPair[A][B]{M}{N}\in\CwfEl{\Gamma}{\CwfSg{A}{B}}$.

    \item \emph{Elimination.} For each $M\in\CwfEl{\Gamma}{\CwfSg{A}{B}}$,
      elements $\CwfFst[A][B]{M}\in\CwfEl{\Gamma}{A}$ and
      $\CwfSnd[A][B]{M}\in\CwfEl{\Gamma}{\ReIx{[\Id,\CwfFst[A][B]{M}]}{B}}$.

    \item \emph{Computation.} For $M\in\CwfEl{\Gamma}{A}$ and
      $N\in\CwfEl{\Gamma}{\ReIx{[\Id,M]}{B}}$, the following equations:
      \begin{mathpar}
        \CwfFst[A][B]{\CwfPair[A][B]{M}{N}} = M
        \and
        \CwfSnd[A][B]{\CwfPair[A][B]{M}{N}} = N
      \end{mathpar}

    \item \emph{Unicity.} For $M\in\CwfEl{\Gamma}{\CwfSg{A}{B}}$, the equation
      $M=\CwfPair[A][B]{\CwfFst[A][B]{M}}{\CwfSnd[A][B]{M}}$.

    \item \emph{Level restriction.} The following equations:
      \begin{mathpar}
        \TyLift{k}{l}{\CwfSg{A}{B}} = \CwfSg{\TyLift{k}{l}{A}}{\TyLift{k}{l}{B}}
        \and
        \CwfPair[\TyLift{k}{l}{A}][\TyLift{k}{l}{B}]{M}{N} = \CwfPair[A][B]{M}{N}
        \and
        \CwfFst[\TyLift{k}{l}{A}][\TyLift{k}{l}{B}]{M} = \CwfFst[A][B]{M}
        \and
        \CwfSnd[\TyLift{k}{l}{A}][\TyLift{k}{l}{B}]{M} = \CwfSnd[A][B]{M}
      \end{mathpar}

    \item \emph{Naturality.} For substitutions $\Delta\rTo^\gamma\Gamma$ the
      following naturality equations:
      \begin{mathpar}
        \ReIx{\gamma}{\CwfSg{A}{B}} = \CwfSg{\ReIx{\gamma}{A}}{\ReIx{\CwfLift{\gamma}}{B}}
        \and
        \ReIx{\gamma}{\CwfPair[A][B]{M}{N}} =
        \CwfPair[\ReIx{\gamma}{A}][\ReIx{\CwfLift{\gamma}}{B}]{\ReIx{\gamma}{M}}{\ReIx{\gamma}{N}}
        \and
        \ReIx{\gamma}{\CwfFst[A][B]{M}} = \CwfFst[\ReIx{\gamma}{A}][\ReIx{\CwfLift{\gamma}}{B}]{\ReIx{\gamma}{M}}
        \and
        \ReIx{\gamma}{\CwfSnd[A][B]{M}} = \CwfSnd[\ReIx{\gamma}{A}][\ReIx{\CwfLift{\gamma}}{B}]{\ReIx{\gamma}{M}}
      \end{mathpar}

    \item \emph{Coercion.} When $\Gamma\rProjMapsto^\CubeFib \Psi,i$ and
      $\IsDim<\Psi>{r,r'}$ and
      $M\in\CwfEl{\ReIx*{r/i}{\Gamma}}{\ReIxCleave*{r/i}{\CwfSg{A}{B}}}$, we require the following
      equation:
      \begin{mathpar}
        \Rule{
          M_0 \triangleq \CwfCoe{i.A}{r}{r'}{\CwfFst[\ReIxCleave*{r/i}{A}][\ReIxCleave*{r/i}{B}]{M}}
          \\
          M_1 \triangleq
            \CwfCoe{
              i.
              \ReIx{
                \Snoc{\Id}{
                  \CwfCoe{i.A}{r}{i}{\CwfFst[\ReIxCleave*{r/i}{A}][\ReIxCleave*{r/i}{B}]{M}}
                }
              }{B}
            }{r}{r'}{
              \CwfSnd[\ReIxCleave*{r/i}{A}][\ReIxCleave*{r/i}{B}]{M}
            }
        }{
          \CwfCoe{i.\CwfSg{A}{B}}{r}{r'}{M}
          =
          \CwfPair[
            \ReIxCleave*{r'/i}{A}
          ][
            \ReIxCleave*{r'/i}{B}
          ]{M_0}{M_1}
        }
      \end{mathpar}

  \end{itemize}

\end{definition}

\begin{definition}[Dependent path types]\label{def:path-types}

  A cwf has \emph{dependent path types} when it is
  equipped with the following structure:
  \begin{itemize}

    \item \emph{Formation.} For each type $A\in\CwfTy[n]{\ReIx{\Wk{\imath}}{\Gamma}}$ and elements
      $\etc{N_\e\in\CwfEl{\Gamma}{\ReIxCleave*{\e/i}{A}}}$, a type
      $\CwfPath{i.A}{N_0}{N_1}\in\CwfTy[n]{\Gamma}$.

    \item \emph{Introduction.} For each $M\in\CwfEl{\ReIx{\Wk{\imath}}{\Gamma}}{A}$, an element
      $\CwfPathLam[i.A]{M}\in\CwfEl{\Gamma}{\CwfPath{i.A}{\ReIxCleave*{0/i}{M}}{\ReIxCleave*{1/i}{M}}}$.

    \item \emph{Elimination.} For each
      $M\in\CwfEl{\Gamma}{\CwfPath{i.A}{N_0}{N_1}}$ and $r\in\CwfDim{\Gamma}$, an
      element $\CwfPathApp[i.A]{M}{r}\in\CwfEl{\Gamma}{\ReIxCleave*{r/i}{A}}$ satisfying
      the equations $\etc{\CwfPathApp[i.A]{M}{\epsilon}=N_\epsilon}$.

    \item \emph{Computation.} For $M\in\CwfEl{\ReIx{\Wk{\imath}}{\Gamma}}{A}$ and
      $r\in\CwfDim{\Gamma}$, the equation $\CwfPathApp[i.A]{\CwfPathLam[i.A]{i.M}}{r} =
      \ReIxCleave*{r/i}{M}$.

    \item \emph{Unicity.} For $M\in\CwfEl{\Gamma}{\CwfPath{i.A}{N_0}{N_1}}$, the
      equation $M=\CwfPathLam[i.A]{j.\CwfPathApp[i.\ReIxCleave{\Wk{\jmath}}{A}]{\ReIxCleave{\Wk{\jmath}}{M}}{j}}$.

    \item \emph{Level restriction.} The following equations:
      \begin{mathpar}
        \TyLift{k}{l}{\CwfPath{i.A}{N_0}{N_1}} = \CwfPath{i.\TyLift{k}{l}{A}}{N_0}{N_1}
        \and
        \CwfPathLam[i.\TyLift{k}{l}{A}]{M}{r} = \CwfPathLam[i.A]{M}{r}
        \and
        \CwfPathApp[i.\TyLift{k}{l}{A}]{M}{r} = \CwfPathApp[i.A]{M}{r}
      \end{mathpar}

    \item \emph{Naturality.} For $\Delta\rTo^\gamma\Gamma$, the following naturality equations:
      \begin{mathpar}
        \ReIx{\gamma}{\CwfPath{i.A}{N_0}{N_1}} = \CwfPath{i.\ReIx*{\Lift{\Wk{\imath}}{\gamma}}{A}}{\ReIx{\gamma}{N_0}}{\ReIx{\gamma}{N_1}}
        \and
        \ReIx{\gamma}{\CwfPathLam[i.A]{i.M}} = \CwfPathLam[i.\ReIx*{\Lift{\Wk{\imath}}{\gamma}}{A}]{i.\ReIx*{\Lift{\Wk{\imath}}{\gamma}}{M}}
        \and
        \ReIx{\gamma}{\CwfPathApp[i.A]{M}{r}} = \CwfPathApp[i.\ReIx*{\Lift{\Wk{\imath}}{\gamma}}{A}]{\ReIx{\gamma}{M}}{\ReIx{\gamma}{r}}
      \end{mathpar}

    \item \emph{Coercion.} When $\Gamma \rProjMapsto^\CubeFib \Psi,j$ and
      $\IsDim<\Psi>{r,r'}$ and
      $M\in\CwfEl{\ReIx*{r/j}{\Gamma}}{\ReIxCleave*{r/j}{\CwfPath{i.A}{N_0}{N_1}}}$, we require the following
      equation:
      \[
        \CwfCoe{j.\CwfPath{i.A}{N_0}{N_1}}{r}{r'}{M}
        =
        \CwfPathLam[
          i.\ReIxCleave*{r'/j}{A}
        ]{
          i.
          \CwfCom{j.A}{r}{r'}{
            \CwfPathApp[i.\ReIxCleave*{r/j}{A}]{
              \ReIxCleave{\Wk{\imath}}{M}
            }{i}
          }{
            \etcsys[\e]{i}{
              j.\ReIxCleave{\Wk{\jmath}}{N_\e}
            }
          }
        }
      \]

  \end{itemize}
\end{definition}

\begin{definition}[Dependent equality types]

  A cwf which has both boundary separation and dependent path types is said to have
  \emph{dependent equality types}, and we accordingly write $\CwfEq{i.A}{M}{N}$ for
  $\CwfPath{i.A}{M}{N}$.

\end{definition}

\begin{definition}[Universes \`a la Russell]\label{def:universes-russell}

  An algebraic cumulative cwf has universes \`a la Russell iff for all levels
  $k<l$ and context $\Gamma:\Cx$, there is a type
  $\CwfUniv[k]\in\CwfTy[l]{\Gamma}$ such that $\CwfEl{\Gamma}{\CwfUniv[k]} =
  \CwfTy[k]{\Gamma}$. We additionally require the naturality equations
  $\ReIx{\gamma}{\CwfUniv[k]}=\CwfUniv[k]$ and $\TyLift{l}{m}{\CwfUniv[k]} =
  \CwfUniv[k]$.

\end{definition}

\begin{definition}[Type-case]\label{def:type-case}
  An algebraic cumulative cwf has \emph{type-case} iff given the following data,
  \begin{mathpar}
    C\in\CwfTy[l]{\Gamma}
    \and
    X\in\CwfEl{\Gamma}{\CwfUniv[k]}
    \and
    M_\Pi,M_\Sigma\in\CwfEl{\Gamma.\CwfUniv[k].\CwfPi{\Var}{\CwfUniv[k]}}{\ReIx*{\Proj\circ\Proj}{C}}
    \and
    M_{\mathbf{Eq}}\in\CwfEl{
      \Gamma.\CwfUniv[k].\CwfUniv[k].%
      \CwfEq{\_.\CwfUniv[k]}{\ReIx{\Proj}{\Var}}{\Var}.%
      \ReIx*{\Proj\circ\Proj}{\Var}.%
      \ReIx*{\Proj\circ\Proj}{\Var}
    }{
      \ReIx*{
        \Proj\circ\Proj\circ\Proj\circ\Proj\circ\Proj
      }{C}
    }
    \and
    M_{\CwfBool}\in\CwfEl{\Gamma}{C}
    \and
    M_{\mathbf{U}}\in\CwfEl{\Gamma}{C}
  \end{mathpar}
  we have an element
  $\CwfUCase{C}{X}{M_\Pi}{M_\Sigma}{M_{\mathbf{Eq}}}{M_{\CwfBool}}{M_{\mathbf{U}}}\in\CwfEl{\Gamma}{C}$
  such that the following conditions hold:
  \begin{itemize}
    \item \emph{Computation.}
      \begin{mathpar}
        \CwfUCase{C}{
          \CwfPi{A}{B}
        }{M_\Pi}{M_\Sigma}{M_{\mathbf{Eq}}}{M_{\CwfBool}}{M_{\mathbf{U}}}
        =
        \ReIx{
          \Snoc{\Snoc{\Id}{A}}{
            \CwfLam[A][\CwfUniv[k]]{B}
          }
        }{M_\Pi}
        \and
        \CwfUCase{C}{
          \CwfSg{A}{B}
        }{M_\Pi}{M_\Sigma}{M_{\mathbf{Eq}}}{M_{\CwfBool}}{M_{\mathbf{U}}}
        =
        \ReIx{
          \Snoc{\Snoc{\Id}{A}}{
            \CwfLam[A][\CwfUniv[k]]{B}
          }
        }{M_\Sigma}
        \and
        \Rule{
          \eta \triangleq
          \Snoc{
            \Snoc{
              \Snoc{
                \Snoc{
                  \Snoc{\Id}{\ReIxCleave*{0/i}{A}}
                }{
                  \ReIxCleave*{1/i}{A}
                }
              }{
                \CwfPathLam[\_.\CwfUniv[k]]{i.A}
              }
            }{N_0}
          }{N_1}
        }{
          \CwfUCase{C}{
            \CwfEq{i.A}{N_0}{N_1}
          }{M_\Pi}{M_\Sigma}{M_{\mathbf{Eq}}}{M_{\CwfBool}}{M_{\mathbf{U}}}
          =
          \ReIx{\eta}{M_{\mathbf{Eq}}}
        }
        \and
        \CwfUCase{C}{\CwfBool}{M_\Pi}{M_\Sigma}{M_{\mathbf{Eq}}}{M_{\CwfBool}}{M_{\mathbf{U}}}
        = M_{\CwfBool}
        \and
        \CwfUCase{C}{\CwfUniv[l]}{M_\Pi}{M_\Sigma}{M_{\mathbf{Eq}}}{M_{\CwfBool}}{M_{\mathbf{U}}}
        = M_{\mathbf{U}}
      \end{mathpar}

    \item \emph{Level restriction.} The following equation:
      \begin{mathpar}
        \CwfUCase{\TyLift{k'}{l}{C}}{\CwfUniv[l]}{M_\Pi}{M_\Sigma}{M_\mathbf{Eq}}{M_\CwfBool}{M_{\mathbf{U}}}
        =
        \CwfUCase{C}{\CwfUniv[l]}{M_\Pi}{M_\Sigma}{M_\mathbf{Eq}}{M_\CwfBool}{M_{\mathbf{U}}}
      \end{mathpar}

    \item \emph{Naturality.} For $\Delta\rTo^\gamma\Gamma$, the following naturality condition:
      \begin{mathpar}
        \Rule{
          \gamma_{+2} \triangleq
          \Snoc{\Snoc{\gamma\circ \Proj\circ\Proj}{\Var}}{\Var}
          \\
          \gamma_{+5} \triangleq
          \Snoc{\Snoc{\Snoc{\Snoc{\Snoc{\gamma\circ \Proj\circ\Proj\circ\Proj\circ\Proj\circ\Proj}{\Var}}{\Var}}{\Var}}{\Var}}{\Var}
        }{
          \ReIx{\gamma}{\CwfUCase{C}{X}{M_\Pi}{M_\Sigma}{M_{\mathbf{Eq}}}{M_{\CwfBool}}{M_{\mathbf{U}}}}
          =
          \\
          \CwfUCase{
            \ReIx{\gamma}{C}
          }{
            \ReIx{\gamma}{X}
          }{
            \ReIx{\gamma_{+2}}{M_\Pi}
          }{
            \ReIx{\gamma_{+2}}{M_\Sigma}
          }{
            \ReIx{\gamma_{+5}}{M_{\mathbf{Eq}}}
          }{
            \ReIx{\gamma}{M_{\CwfBool}}
          }{
            \ReIx{\gamma}{M_{\mathbf{U}}}
          }
        }
      \end{mathpar}
  \end{itemize}

\end{definition}

\subsection{Syntactic model and initiality}

The cwfs with all the structure described in \cref{sec:cwf-structure} can be
arranged into a category which has an initial object. This is because every
piece of structure that we have defined in \cref{sec:cwf-structure} is
\emph{generalized algebraic} in the sense of
\Cite{cartmell:1978,cartmell:1986}; even the universe structure can be seen to
be generalized algebraic~\Cite{sterling:2018:gat}.
We conjecture (but do not prove) that the syntax of \XTT{} can be used to
construct a cwf (the Lindenbaum-Tarski algebra) which has the universal
property of the initial object.

\paragraph*{Syntactic presentation of augmented cubes}

The syntactic contexts $\Psi$ can be viewed as a particular syntactic
presentation of the category $\AugCube$ of augmented Cartesian cubes, in which
the equalizers are implemented formally by extending $\Psi$ with equations.

\paragraph*{Category of contexts}

The well-typed term contexts $\IsCx<\Psi>{\Gamma}$ can be organized (up to
judgmental equality) into a category with morphisms $(\Psi'\mid\Gamma')
\rTo^{(\psi,\gamma)} (\Psi\mid\Gamma)$ with
$\Psi'\rTo^\psi\Psi$ and
$\Gamma'\rTo^\gamma\ReIx{\psi}{\Gamma}$, i.e.\ a well-typed substitution of
terms in context $\ReIx{\psi}{\Gamma}$ for the variables from
$\Gamma'$.

\paragraph*{Types and terms}

The presheaves of types $\CwfTy[n]$ are given by syntactic types
$\IsTy<\Psi>[\Gamma]{A}[n]$ up to judgmental equality, with action given by
substitutions; well-typed terms $\IsTm<\Psi>[\Gamma]{M}{A}$ taken up to
judgmental equality generate the fibers of a natural transformation
$\CwfEl\rTo^\pi\CwfTy$. The representability of $\pi$ is immediate from the fact
that syntactic contexts $\Psi\mid\Gamma$ can be extended by any type to
yield $\Psi\mid\Gamma,x:A$. This context comes equipped with a projection
$\Psi\mid\Gamma$ and $\IsTm[\Gamma,x:A]{x}{A}$ which implements the desired
pullback square. The initiality of $\Psi,0=1\mid\Gamma$ is ensured by
the \textsc{false constraint} rule (see \cref{sec:rules:structural}).

\paragraph*{Cubical judgmental structure}

The functor $\Cx\rTo^{\vert\CubeFib\vert}\AugCube$ takes a context
$\Psi\mid\Gamma$ to $\Psi$, and projects out the dimension component of a
substitution $(\psi, \gamma)$.  It is equipped with a splitting which, for the
context $\Psi\mid\Gamma$, lifts $\Psi'\rTo^\psi\Psi$ to
$\Psi'\mid\ReIx{\psi}{\Gamma}\rTo^{(\psi,\Id)}\Psi\mid\Gamma$. Semantic
boundary separation is obtained immediately from the boundary separation rules.

\paragraph*{Connectives} We observe that our cwf also has dependent function,
pair, equality, boolean and universe types given by the syntax.

%

\section{Cubical logical families and gluing}\label{sec:computability-cwf}

We fix an \emph{arbitrary} structured cwf $\Cx$; we will write $\Gamma:\Cx$ for
its objects. We will show how to build a cwf $\Pred$ of cubical logical
families over $\Cx$, called the ``computability cwf'', which has some (but not
all) of the structure of a model of \XTT.  Then, in
\cref{sec:closed-computability-model} we will construct a genuine model
$\CPred$ of \XTT{} by restricting $\Pred$ to a closed universe.

\subsection{The cubical nerve}

We have a functor $\AugCube\rTo^\CubeCx\Cx$ which
takes a dimension context to $\Cx$-context with those dimensions but no term
variables:
\begin{align*}
  \CubeCx{\Psi} &= \ReIx{\Wk{\Psi}}{\parens{\cdot}}
  \\
  \CubeCx{\Phi\rTo^\psi\Psi} &=
  \CubeCx{\Phi} \rTo^{\CartUniMap<\psi>{\Cleave{\Wk{\Phi}}{(\cdot)}}{\Cleave{\Wk{\Psi}}{(\cdot)}}} \CubeCx{\Psi}
\end{align*}

This functor $\CubeCx$ induces a \emph{nerve} construction
$\Cx\rTo^{\Nerve}\AugCSet$, taking $\Gamma$ to the presheaf
$\Hom[\Cx]{\CubeCx{-}}{\Gamma}$. Intuitively, this is the presheaf of
substitutions which are closed with respect to term variables, but open with
respect to dimension variables; from the perspective of inside $\AugCSet$,
these are the closed substitutions.

Abusing notation slightly, we now define the cubical set of ``closed types''
$\NerveTy[k]$ as $\CwfTy[k]\circ\OpCat{\CubeCx}$. We furthermore define a
dependent cubical set $\NerveEl[k]$ over $\NerveTy[k]$, taking
$(\Psi,A)\in\int\NerveTy[k]$ to $\CwfEl{\CubeCx{\Psi}}{A}$; internally, we will (abusively)
write $\NerveEl{A}$ for the fiber of $\NerveEl[k]$ over $A:\NerveTy[k]$.
Given $\gamma:\Nerve{\Gamma}$ and $A\in\CwfTy[k]{\Gamma}$, we abuse notation
by writing $\ReIx{\gamma}{A}:\NerveTy[k]$.

Furthermore, we also define a dependent cubical set $\NerveFam[n]{A}$ over
$A:\NerveTy[n]$, which internalizes the \emph{families} of $\Cx$-types indexed
in a given $\Cx$-type. Explicitly, we define a presheaf $\NerveFam[n]$ whose
fibers are $\coprod_{A\in\CwfTy[n]{\CubeCx{\Psi}}}\CwfTy[n]{\CubeCx{\Psi}.A}$ for each $\Psi$,
and then exhibit the obvious projection $\NerveFam[n]\rightarrowtriangle\NerveTy[n]$.

\begin{lemma}\label{lem:el-univ-ty}
  For each level $n$, we have $\NerveEl{\CwfUniv[n]} = \NerveTy[n]$.
\end{lemma}
\begin{proof}
  This follows from the fact that $\Cx$ is a model of universes \`a la Russell. Calculate:
  \begin{align*}
    \NerveEl{\CwfUniv[n]}
    &=
    \CwfEl{\CubeCx{-}}{\CwfUniv[n]}
    \tag{definition}
    \\
    &=
    \CwfTy[n]{\CubeCx{-}}
    \tag{$\Cx$ has universes}
    \\
    &= \NerveTy[n]
    \tag{definition}
  \end{align*}
  \qedhere
\end{proof}

\subsection{Logical families by semantic gluing}

By gluing the family fibration along the nerve functor
$\Cx\rTo^{\Nerve{-}}\AugCSet$, we acquire a category of \emph{cubical
logical families}, which we can used to prove canonicity for closed terms,
instantiating $\Cx$ with the initial structured cwf. Intuitively, the role of
the gluing category is to ``cut down'' the morphisms in cubical sets to those
which are definable in $\Cx$, allowing us to extract non-trivial theorems about
$\Cx$ using the very powerful tools afforded by the topos $\AugCSet$.

We will prefer a more explicit and type-theoretic presentation of the gluing
category, but it is helpful for intuition to view it as a pullback of the
fundamental fibration for $\AugCSet$ along the cubical nerve functor:
\begin{diagram}
  \DPullback*{\Pred}{\ArrCat{\AugCSet}}{\Cx}{\AugCSet}{\CodFib}{\Nerve{-}}{\PiSem}{\PiSyn}
\end{diagram}
Another view of the gluing category comes from the comma construction $\Id_{\AugCSet}\downarrow\Nerve$.

\paragraph*{Diagrammatic construction of $\Pred$}

Explicitly, an object in $\Pred$ is a triple
$\Gl{\Gamma}=(\Gamma, \Pred{\Gamma},\Quo{\Gamma})$ of a context $\Gamma:\Cx$, a cubical set $\Pred{\Gamma}$,
and a natural transformation
$\Pred{\Gamma}\rTo^{\Quo{\Gamma}}\Nerve{\Gamma}$; a morphism
$\Gl{\Delta}\rTo^{\Gl{\gamma}}\Gl{\Gamma}$ is then a pair
$\Gl{\gamma}=(\gamma, \Pred{\gamma})$ together with a commuting square:
\begin{diagram}
  \DRow \Pred{\Delta} & \rTo^{\Pred{\gamma}}  & \Pred{\Gamma}
  \\
  \DRow \dTo^{\Quo{\Gamma}} && \dTo_{\Quo{\Delta}}
  \\
  \DRow \Nerve{\Delta} & \rTo_{\Nerve{\gamma}} & \Nerve{\Gamma}
\end{diagram}

\paragraph*{Type-theoretic construction of $\Pred$}

Following~\Cite{coquand:2018}, we will prefer a \emph{type-theoretic}
presentation of $\Pred$ in terms of the hierarchy of Grothendieck universes
$\SemUniv[n]$, which lift directly into $\AugCSet$ as in~\Cite{hofmann-streicher:1997}. We found that this
type-theoretic style scales more easily to the complex situations involved in
the semantics of dependent type theory than the diagrammatic style above.

According to the type-theoretic presentation, an object of $\Pred$ is a pair
$\Gl{\Gamma} = \parens{\Gamma, \Pred{\Gamma}}$ with $\Gamma:\Cx$ and
$\Pred{\Gamma}$ a family
$\Nerve{\Gamma}\to\SemUniv[n]$ for some $n$.
A morphism $\Gl{\Delta}\rTo^{\Gl{\gamma}}\Gl{\Gamma}$ is a pair $\Gl{\gamma} =
\parens*{\Delta\rTo^\gamma\Gamma,\Pred{\gamma}}$ with $\Pred{\gamma} :
\prod_{\delta:\Nerve{\Delta}} \Pred{\Delta}{\delta} \to
\Pred{\Gamma}\parens{\ReIx{\gamma}{\delta}}$. To be precise, $\Pred{\gamma}$ is
a global element of the dependent function type in $\AugCSet$;
$\Pred{\gamma}$ witnesses the fact that the syntactic substitution $\gamma$
preserves the logical family.

There is a slight mismatch with the earlier diagrammatic intuition: the type-theoretic presentation
only allows for families whose fibers fit into $\SemUniv[n]$ for some $n$. Since we will work
exclusively with the more restrictive type-theoretic presentation from now on this poses no
technical challenges. Those who prefer the intuition provided by the diagrammatic presentation need
merely restrict the pullback construction to certain suitably small cubical sets.

\paragraph*{What's it for?} $\Pred{\Gamma}$ is a \emph{proof-relevant} logical
predicate (``logical family'') on elements of $\Gamma$ which may have free
dimension variables, but which commutes with all substitutions of those
dimension variables. In other words, $\Pred{\Gamma}$ is a (cubical, proof-relevant) predicate on the
elements of $\Gamma$.

\subsection{Cwf structure: types and elements}

A \emph{glued type} of level $l$ in context $\Gl{\Gamma}$ is a pair
$\Gl{A}=(A,\Pred{A})$ with $A\in\CwfTy[l]{\Gamma}$ and $\Pred{A}$ a
global element of the cubical set
$\prod_{\gamma:\Nerve{\Gamma}}\prod_{\Pred{\gamma}:\Pred{\Gamma}\gamma}\NerveEl{\ReIx{\gamma}{A}}\to\SemUniv[l]$.
Level restrictions $\TyLift{k}{l}{\Gl{A}}$ are inherited directly from $\Cx$:
the family part of a glued type can remain unchanged because
$\NerveEl{A} = \NerveEl{\TyLift{k}{l}{A}}$ and $\SemUniv[k] \subseteq \SemUniv[l]$.
A \emph{glued element} in context $\Gl{\Gamma}$ of type
$\Gl{A}\in\CwfTy<\Pred>[n]{\Gl{\Gamma}}$ is a pair $\Gl{M}=(M,\Pred{M})$ with
$M\in\CwfEl{\Gamma}{A}$ and $\Pred{M}$ a global element of the cubical set
$\prod_{\gamma:\Nerve{\Gamma}}\prod_{\Pred{\gamma}:\Pred{\Gamma}\gamma}\Pred{A}\gamma\Pred{\gamma}\parens*{\ReIx{\gamma}{M}}$.

We observe that the induced projection
$\CwfEl<\Pred>\rTo^\pi\CwfTy<\Pred>$ is representable by exhibiting the evident context comprehension
$\Gl{\Gamma}.\Gl{A}$ for $\Gl{\Gamma}:\Pred$ and
$\Gl{A}\in\CwfTy<\Pred>{\Gamma}$, taking $\Gamma.A$ for its syntactic part, and
using the following family for its semantic part:
\begin{align*}
  \Pred{\parens*{\Gl{\Gamma}.\Gl{A}}}
  \Snoc{\gamma}{a}
  &=
  \textstyle
  \sum_{\Pred{\gamma}:\Pred{\Gamma}\gamma}
  \Pred{A}\gamma\Pred{\gamma}a
\end{align*}

We clearly have that the restrictions
$\CwfEl<\Pred>{\Gl{\Gamma}}{\Gl{A}} \rTo
\CwfEl<\Pred>{\Gl{\Gamma}}{\TyLift{k}{l}{\Gl{A}}}$ are identities, and $\Gl{\Gamma}.\Gl{A} =
\Gl{\Gamma}.\TyLift{k}{l}{\Gl{A}}$.
Therefore, $\Pred$ forms an algebraic cumulative cwf; we will call it the
``computability cwf''. The cubical structure is inherited from $\Cx$ by precomposing with the
fibration $\Pred\rProjto^{\PiSyn}\Cx$. We define the glued constraint comprehension $\Gl{\Gamma}.\xi$ by taking
$\Gamma.\xi$ for the syntactic part, and defining its logical family as follows:
\[
  \Pred{\parens*{\Gl{\Gamma}.r=s}}(\gamma.r=s) =
  \braces*{\Pred{\Gamma}\gamma \mid r=s}
\]

This choice of realizers for the constraint comprehension ensures the
initiality of inconsistent contexts in the gluing model.

\section{Canonicity for \texorpdfstring{\XTT}{XTT}: the computability model}\label{sec:closed-computability-model}

We have not succeeded in closing the computability cwf from
\cref{sec:computability-cwf} under Kan universes \`a la Russell of Kan types;
the essential difficulties are the separation property and the regular
coercion structure. Therefore, this cwf does not have the structure of a model
of \XTT.

To rectify this, we will restrict $\Pred$ to a smaller cwf $\CPred$, in which
the types are generated inductively in a way reminiscent of the construction of
closed universes in PER
models~\Cite{allen:1987:thesis}; the
main difference is that, rather than using large induction-recursion (which has
not been shown to exist in presheaf toposes), we model $n$ object universes in
the cubical universe $\SemUniv[n+1]$ (an instance of \emph{small}
induction-recursion, which can be translated to constructs available in every
presheaf topos~\Cite{hancock-ghani-malatesta-altenkirch:2013,moerdijk:2000}).

Finally, using the universal property of the type structure of the restricted
cwf, we will generate the coercion and composition structure recursively,
obtaining a model of \XTT{}.

\
\subsection{Closed universe hierarchy}\label{sec:closed-universes}

We will define a family $\Pred{\ClUni{n}} :
\NerveTy[n]\to\SemUniv[n+1]$ internally to $\AugCSet$, together
with $\UPred*{-} : \prod_{\FmtCode{A}:\Pred{\ClUni{n}}A}
\NerveEl{A}\to\SemUniv[n]$. The former will serve as the computability
predicate for a closed universe, and we will use the latter in order to define
a family of types to decode the closed universe.
\begin{mathpar}
  \Rule{
    (j<n)
  }{
    \CodeUni{j}:\Pred{\ClUni{n}}\CwfUniv[j]
  }
  \and
  \Rule{
  }{
    \CodeBool:\Pred{\ClUni{n}}\CwfBool
  }
  \\
  \Rule{
    \FmtCode{A}:\Pred{\ClUni{n}}a
    \\
    \FmtCode{B}:\Pred{\ClUniFam{n}{\FmtCode{A}}}B
  }{
    \CodePi{\FmtCode{A}}{\FmtCode{B}} :\Pred{\ClUni{n}}\CwfPi{A}{B}
  }
  \and
  \Rule{
    \FmtCode{A}:\Pred{\ClUni{n}}A
    \\
    \FmtCode{B}:\Pred{\ClUniFam{n}{\FmtCode{A}}}B
  }{
    \CodeSg{\FmtCode{A}}{\FmtCode{B}} :\Pred{\ClUni{n}}\CwfSg{A}{B}
  }
  \and
  \Rule{
    \textstyle
    \FmtCode{A} : \prod_{i:\Dim} \Pred{\ClUni{n}}{A_i}
    \\
    \etc{\Pred{N_\e} : \UPred{\FmtCode{A}(\e)}N_\e}
  }{
    \CodeEq{\FmtCode{A}}{\Pred{N_0}}{\Pred{N_1}}
    :
    \Pred{\ClUni{n}}\CwfEq{i.A_i}{N_0}{N_1}
  }
\end{mathpar}

We define an auxiliary family of types to capture family of type-codes:
\begin{align*}
  \Pred{\ClUniFam{n}{-}} &:
  \textstyle
  \prod_{\FmtCode{A}:\Pred{\ClUni{n}}A}
  \NerveFam[n]{A}\to\SemUniv[n+1]
  \\
  \Pred{\ClUniFam{n}{\FmtCode{A}}} B &=
  \textstyle
  \prod_{M : \NerveEl{A}}
  \UPred{\FmtCode{A}} M
  \to
  \Pred{\ClUni{n}}\parens*{
    \ReIx{\Snoc{\Id}{M}}{B}
  }
\end{align*}

The assignment of computability families to type codes is as follows:
\begin{align*}
  \UPred{(-)} &: \textstyle
  \prod_{\FmtCode{A}:\Pred{\ClUni{n}}{A}}
  \NerveEl{A}\to \SemUniv[n]
  \\
  \UPred{\CodeUni{i}} &= \Pred{\ClUni{i}}
  \\
  \UPred{\CodeBool} &=
  \lambda M.\,
  \parens*{M = \CwfTrue}
  +
  \parens*{M = \CwfFalse}
  \\
  \UPred{
    \CodePi{\FmtCode{A}}{\FmtCode{B}}
  } &=
  \textstyle
  \lambda M.\,
  \prod_{N:\NerveEl{A}}
  \prod_{\Pred{N}:\Pred{\FmtCode{A}}N}
  \UPred*{\FmtCode{B}N\Pred{N}}
  \CwfApp[A][B]{M}{N}
  \\
  \UPred{
    \CodeSg{\FmtCode{A}}{\FmtCode{B}}
  } &=
  \textstyle
  \lambda M.\
  \sum_{
    \Pred{M_0} : \UPred{\FmtCode{A}}\CwfFst[A][B]{M}
  }
  \UPred*{\FmtCode{B}\parens*{\CwfFst[A][B]{M}}\Pred{M_0}}\CwfSnd[A][B]{M}
  \\
  \UPred{
    \CodeEq{\FmtCode{A}}{\Pred{N_0}}{\Pred{N_1}}
  }
  &=
  \textstyle
  \lambda M.\,
  \braces*{
    \Pred{M} :
    \prod_{i:\Dim}
    \UPred{\FmtCode{A}(i)}\CwfPathApp[i.A]{M}{i}
    \mid
    \etc{
      \Pred{M}(\e) = \Pred{N_\e}
    }
  }
\end{align*}

\begin{lemma}\label{lem:type-code-lift}
  For any $A:\NerveTy[k]$ and with $k\leq l$, we have $\Pred{\ClUni{k}}{A} = \Pred{\ClUni{l}}\TyLift{k}{l}{A}$.
\end{lemma}
\begin{proof}

  We will show that
  $\Pred{\ClUni{k}}{A}\subseteq\Pred{\ClUni{l}}\TyLift{k}{l}{A}$; the other
  direction is symmetric. Fix $\FmtCode{A}:\Pred{\ClUni{k}}{A}$; we verify that
  $\FmtCode{A}:\Pred{\ClUni{l}}\TyLift{k}{l}{A}$ as well, proceeding by
  induction.
  \begin{description}
    \item[Case.]
      \[
        \Rule{
          (j<k)
        }{
          \CodeUni{j}:\Pred{\ClUni{k}}\CwfUniv[j]
        }
      \]
      We have $\CodeUni{j}:\Pred{\ClUni{l}}\CwfUniv[j]$, and $\CwfUniv[j] = \TyLift{k}{l}{\CwfUniv[j]}$.

    \item[Case.]
      \[
        \Rule{
        }{
          \CodeBool:\Pred{\ClUni{k}}\CwfBool
        }
      \]
      We likewise have $\CodeBool:\Pred{\ClUni{l}}\CwfBool$, and $\CwfBool = \TyLift{k}{l}{\CwfBool}$.

    \item[Case.]
      \[
        \Rule{
          \FmtCode{A}:\Pred{\ClUni{k}}A
          \\
          \FmtCode{B}:\Pred{\ClUniFam{k}{\FmtCode{A}}}B
        }{
          \CodePi{\FmtCode{A}}{\FmtCode{B}} :\Pred{\ClUni{k}}\CwfPi{A}{B}
        }
      \]

      To see that $\CodePi{\FmtCode{A}}{\FmtCode{B}} :
      \Pred{\ClUni{l}}{\TyLift{k}{l}{\CwfPi{A}{B}}}$, by calculation, it
      suffices to show that $\CodePi{\FmtCode{A}}{\FmtCode{B}} :
      \Pred{\ClUni{l}}{\CwfPi{\TyLift{k}{l}{A}}{\TyLift{k}{l}{B}}}$.  By
      induction, we have $\FmtCode{A}:\Pred{\ClUni{l}}\TyLift{k}{l}{A}$; to
      verify that
      $\FmtCode{B}:\Pred{\ClUniFam{l}{\FmtCode{A}}}\TyLift{k}{l}{B}$, we fix
      $M:\NerveEl{A}$ and $\Pred{M}:\UPred{\FmtCode{A}}M$ and need to check that
      $\FmtCode{B}M\Pred{M}:\Pred{\ClUni{l}}{\TyLift{k}{l}{\parens*{\ReIx{\Snoc{\Id}{M}}{B}}}}$;
      but this follows from our second induction hypothesis and the fact that
      level restriction commutes with substitution.

  \end{description}

  The remaining cases are analogous.
\end{proof}

\begin{notation}
  We will write $\Boundary{r}$ for the formula $(r = 0) \lor (r = 1)$, the \emph{boundary} of $r$.
\end{notation}

\begin{lemma}\label{lem:ty-nerve-path-unicity}
  Internally to $\AugCSet$, the following formulas are true:
  \begin{gather*}
    \forall
      r:\Dim,
      n:\mathbb{N},
      A,B:\NerveTy[n].\,
    \parens*{
      \Boundary{r}\implies A = B
    }
    \implies
    A=B
    \tag{types}
    \\
    \forall
      r:\Dim,
      A:\NerveTy[n],
      M,N:\NerveEl{A}.\,
    \parens*{
      \Boundary{r}\implies M = N
    }
    \implies
    A=B
    \tag{elements}
  \end{gather*}
\end{lemma}
\begin{proof}

  This follows from the fact that, as a model of \XTT{}, $\Cx$ has boundary
  separation; therefore, its types and elements are separated with respect to
  $\Boundary$. This implies that all elements are completely defined by their
  boundaries. We prove this in detail for types only, and the case for terms is
  analogous.  It suffices, using the Kripke-Joyal semantics of the topos
  $\AugCSet$, to fix $\Psi:\AugCube$ and show that for all $r \in \Dim{\Psi}$, $n \in
  \mathbb{N}$, and $A,B \in \NerveTy[n](\Psi)$, if
  $\Psi\forces\parens*{\Boundary{r}\implies A = B}$ then $\Psi\forces A = B$.

  We will write $\Psi.\xi_\e\rTo^{\Wk*{\xi_\e}}\Psi$ for the equalizer of $r,\e$ for
  each $\e\in\Two$. Recalling that $\NerveTy[n](\Psi) =
  \CwfTy[n]{\CubeCx{\Psi}}$, by the $\Cov$-separation of $\CwfTy[n]$ (see
  \cref{def:boundary-separation}), to show that $\Psi\forces A = B$ it
  suffices to show that $\Psi.\xi_e \forces \ReIx{\Wk*{\xi_\e}}{A} =
  \ReIx{\Wk*{\xi_\e}}{B}$ for each $\e\in\Two$.

  Unraveling the Kripke-Joyal paperwork of our assumption, we know that for all
  $\Psi'\rTo^{\psi}\Psi$, if $\Psi'\forces\Boundary{r\circ\psi}$ then
  $\Psi'\forces\ReIx{\psi}{A} = \ReIx{\psi}{B}$. Instantiating this hypothesis
  with the equalizer $\Psi.\xi_e\rTo^{\Wk*{\xi_\e}}\Psi$, it remains only to
  check that $\Psi.\xi_\e\forces\parens{r\circ\Wk{\xi_\e}=0}\lor
  \parens{r\circ\Wk{\xi_\e} = 1}$. But we immediately have $\Psi.\xi_0\forces
  r\circ\Wk{\xi_0}=0$ and $\Psi.\xi_1\forces r\circ\Wk{\xi_1}=1$ by the
  property of the equalizer.
\end{proof}

\begin{definition}

  Given $A:\NerveTy[n]$ and $\FmtCode{A}:\Pred{\ClUni{n}}{A}$, we say that
  $\FmtCode{A}$ is \emph{elementwise separated} (or just \emph{separated}) iff the following formula holds
  internally to $\AugCSet$:
  \[
    \forall
      r:\Dim,
      M:\NerveEl{A},
      \Pred{M_0},\Pred{M_1}:\UPred{\FmtCode{A}}M,
    \parens*{
      \Boundary{r}\implies \Pred{M_0}=\Pred{M_1}
    }
    \implies
    \Pred{M_0}=\Pred{M_1}
  \]

  We write $\HasPathUnicity{\FmtCode{A}}$ for the above formula.

\end{definition}

\begin{definition}

  Given $A:\NerveTy[n]$ and $\FmtCode{A}:\Pred{\ClUni{n}}A$, we say that
  $\FmtCode{A}$ is \emph{typewise separated} iff the following formula holds
  internally to $\AugCSet$:
  \[
    \forall
      r:\Dim,
      B:\NerveTy[n],
      \FmtCode{B}:\Pred{\ClUni{n}}B.\,
    \parens*{
      \Boundary{r}\implies\FmtCode{A}=\FmtCode{B}
    }
    \implies
    \FmtCode{A}=\FmtCode{B}
  \]

  We write $\LocallyPathUnique{\FmtCode{A}}$ for the above formula.

\end{definition}

\begin{lemma}\label{lem:inductive-path-unicity}
  Internally to $\AugCSet$, the following formula is true:
  \[
    \forall
      n:\mathbb{N},
      A:\NerveTy[n],
      \FmtCode{A}:\Pred{\ClUni{n}}A.\,
      \LocallyPathUnique{\FmtCode{A}}
      \land
      \HasPathUnicity{\FmtCode{A}}
  \]
\end{lemma}
\begin{proof}
  We begin by strong induction on $n$; then, we proceed by induction on $\FmtCode{A}$.
  \begin{enumerate}

    \item \emph{Case (dependent function type).}
      Fixing
      $\FmtCode{C}:\Pred{\ClUni{n}}C$ and
      $\FmtCode{D}:\Pred{\ClUniFam{n}{\FmtCode{C}}}D$, we have to verify that
      $\CodePi{\FmtCode{C}}{\FmtCode{D}}$ is typewise and elementwise separated.
      \begin{enumerate}
        \item \emph{Typewise separation}.
          We need to show that
          for all $\FmtCode{B}:\Pred{\ClUni{n}}B$, if
          $\Boundary{r}\implies{}\CodePi{\FmtCode{C}}{\FmtCode{D}} =
          \FmtCode{B}$, then $\CodePi{\FmtCode{C}}{\FmtCode{D}} = \FmtCode{B}$.
          First we observe that there exist $\FmtCode{E},\FmtCode{F}$
          such that $\FmtCode{B}=\CodePi{\FmtCode{E}}{\FmtCode{F}}$. This follows by
          inversion on $\FmtCode{B}$: if $\FmtCode{B}$ is a $\FmtCode{pi}$ our goal is immediate,
          otherwise we would have $\Boundary{r} \implies \bot$ and in $\AugCSet$ we have that
          $(\Boundary{r} \implies \bot) \implies \bot$ giving a contradiction.

          Therefore, we have some $\FmtCode{E}$ and $\FmtCode{F}$ such that
          $\Boundary{r}\implies\CodePi{\FmtCode{C}}{\FmtCode{D}}=\CodePi{\FmtCode{E}}{\FmtCode{F}}$.
          By inversion, we have
          $\Boundary{r}\implies\parens{\FmtCode{C}=\FmtCode{E} \land
          \FmtCode{D}=\FmtCode{F}}$. We need to show that
          $\FmtCode{C}=\FmtCode{E}$ and $\FmtCode{D}=\FmtCode{F}$.
          \begin{enumerate}

            \item Instantiating our induction hypothesis for the typewise
              separation of $\FmtCode{C}$ with $\FmtCode{E}$, we obtain
              $\FmtCode{C}=\FmtCode{E}$ from
              $\Boundary{r}\implies\FmtCode{C}=\FmtCode{E}$.

            \item To see that $\FmtCode{D}=\FmtCode{F}$, we fix $M:\NerveEl{C}$
              and $\Pred{M}:\UPred{\FmtCode{C}}M$, and show that
              $\FmtCode{D}M\Pred{M}=\FmtCode{F}M\Pred{M}$. Instantiating our
              induction hypothesis for the typewise separation of
              $\FmtCode{D}M\Pred{M}$, we obtain
              $\FmtCode{D}M\Pred{M}=\FmtCode{F}M\Pred{M}$ from
              $\Boundary{r}\implies\FmtCode{D}=\FmtCode{F}$.

          \end{enumerate}

        \item \emph{Elementwise separation}
          Fixing $M:\NerveEl{\CwfPi{C}{D}}$ and
          $\Pred{M_0},\Pred{M_1}:\UPred{\CodePi{\FmtCode{C}}{\FmtCode{D}}}M$
          such that $\Boundary{r}\implies\Pred{M_0}=\Pred{N_1}$, we need to
          show that $\Pred{M_0}=\Pred{M_1}$. We fix $N:\NerveEl{C}$ and
          $\Pred{N}:\UPred{\FmtCode{C}}N$, to verify that $\Pred{M_0}N\Pred{N}
          = \Pred{M_1}N\Pred{N}$. Using our induction hypothesis for the
          elementwise separation of $\FmtCode{D}N\Pred{N}$, we obtain
          $\Pred{M_0}N\Pred{N}=\Pred{N_1}\Pred{N}$ from
          $\Boundary{r}\implies\Pred{M_0}N\Pred{N}=\Pred{M_1}N\Pred{N}$.
      \end{enumerate}
    \item \emph{Case (dependent pair type).}
      Fixing
      $\FmtCode{C}:\Pred{\ClUni{n}}C$ and
      $\FmtCode{D}:\Pred{\ClUniFam{n}{\FmtCode{C}}}D$, we have to verify that
      $\CodeSg{\FmtCode{C}}{\FmtCode{D}}$ is typewise and elementwise separated.
      \begin{enumerate}
        \item \emph{Typewise separation.} This case is identical to the
          case for dependent function types.

        \item \emph{Elementwise separation.} Fixing $M:\NerveEl{\CwfSg{C}{D}}$ and
          $\Pred{M_0},\Pred{M_1}:\UPred{\CodeSg{\FmtCode{C}}{\FmtCode{D}}}M$ such
          that $\Boundary{r}\implies\Pred{M_0}=\Pred{M_1}$, we need to show
          that $\Pred{M_0}=\Pred{M_1}$. It suffices to show that
          $\pi_1\Pred{M_0} = \pi_2\Pred{M_1}$ and
          $\pi_2\Pred{M_0}=\pi_2\Pred{M_1}$.
          \begin{enumerate}

            \item From our induction hypothesis for the elementwise separation of
              $\FmtCode{C}$, we obtain $\pi_1\Pred{M_0}=\pi_1\Pred{M_1}$ from
              $\Boundary{r}\implies\Pred{M_0}=\Pred{M_1}$.

            \item From our induction hypothesis for the elementwise separation of
              $\FmtCode{D}M\Pred{M_0}$, we obtain
              $\pi_2\Pred{M_0}=\pi_2\Pred{M_1}$ from
              $\Boundary{r}\implies\Pred{M_0}=\Pred{M_1}$.

          \end{enumerate}
      \end{enumerate}
    \item \emph{Case (equality type).}
      Fixing $\FmtCode{C}:\prod_{i}\Pred{\ClUni{n}}{C_i}$ and
      $M_0:\NerveEl{C_0},\Pred{M_0}:\UPred{\FmtCode{C}0}M_0$ and
      $M_1:\NerveEl{C_1},\Pred{M_1}:\UPred{\FmtCode{C}1}M_1$, we have to
      verify that $\CodeEq{\FmtCode{C}}{\Pred{M_0}}{\Pred{M_1}}$ is typewise and elementwise separated.
      \begin{enumerate}
        \item\emph{Typewise separation.}
          We need to show
          that for all $\FmtCode{B}:\Pred{\ClUni{n}}B$, if
          $\Boundary{r}\implies\CodeEq{\FmtCode{C}}{\Pred{M_0}}{\Pred{M_1}}=\FmtCode{B}$,
          then $\CodeEq{\FmtCode{C}}{\Pred{M_0}}{\Pred{M_1}}=\FmtCode{B}$. By
          inversion, we observe that there exist
          $\FmtCode{D}:\prod_i\Pred{\ClUni{n}}{C_i}$ and
          $\Pred{N_0}:\UPred{\FmtCode{D}0}M_0$ and
          $\Pred{N_1}:\UPred{\FmtCode{D}1}{M_1}$ such that
          $\Boundary{r}\implies\FmtCode{B}=\CodeEq{\FmtCode{D}}{\Pred{N_0}}{\Pred{N_1}}$.
          By inversion we have
          $\Boundary{r}\implies\parens{\FmtCode{D}=\FmtCode{C}\land\etc{\Pred{M_\e}=\Pred{N_\e}}}$.
          We need to show that $\FmtCode{D}=\FmtCode{C}$ and
          $\etc{\Pred{M_\e}=\Pred{N_\e}}$.

          \begin{enumerate}

            \item To see that $\FmtCode{D}=\FmtCode{C}$, we fix $i:\Dim$ and verify
              that $\FmtCode{D}i=\FmtCode{C}i$. Instantiating our induction
              hypothesis for the typewise separation of $\FmtCode{D}i$, we obtain $\FmtCode{D}i=\FmtCode{C}i$
              from $\Boundary{r}\implies\FmtCode{D}=\FmtCode{C}$.

            \item To see that $\Pred{M_\e}=\Pred{N_\e}$, we instantiate our
              induction hypothesis for the elementwise separation of $\FmtCode{D}\e$,
              obtaining $\Pred{M_\e}=\Pred{N_\e}$ from
              $\Boundary{r}\implies\Pred{M_\e}=\Pred{N_\e}$.

          \end{enumerate}
        \item\emph{Elementwise separation.}
          Fixing $P:\NerveEl{\CwfEq{i.C_i}{M_0}{M_1}}$ and
          $\Pred{P_0},\Pred{P_1}:\UPred{\CodeEq{\FmtCode{C}}{\Pred{M_0}}{\Pred{M_1}}}P$
          such that $\Boundary{r}\implies\Pred{P_0}=\Pred{P_1}$, we need to
          show that $\Pred{P_0}=\Pred{P_1}$; fixing $i:\Dim$, we verify that
          $\Pred{P_0}i=\Pred{P_1}i$. Using our induction hypothesis for the
          elementwise separation of $\FmtCode{C}i$, we obtain $\Pred{P_0}i=\Pred{P_1}i$
          from $\Boundary{r}\implies\Pred{P_0}i=\Pred{P_1}$.

      \end{enumerate}

    \item \emph{Case (boolean type).} We need to show that $\CodeBool$ is typewise and elementwise separated.
      \begin{enumerate}
        \item \emph{Typewise separation.} We need to show that for all $\FmtCode{B}:\Pred{\ClUni{n}}B$, if
          $\Boundary{r}\implies\CodeBool=\FmtCode{B}$, then
          $\CodeBool=\FmtCode{B}$. But this is immediate by considering the
          restriction maps for $\CodeBool$.

        \item \emph{Elementwise separation.} Fixing $M:\NerveEl{\CwfBool}$ and
          $\Pred{M_0},\Pred{M_1}:\UPred{\CodeBool}M$ such that
          $\Boundary{r}\implies\Pred{M_0}=\Pred{M_1}$, we need to show that
          $\Pred{M_0}=\Pred{M_1}$. We obtain our goal by case on
          $\Pred{M_0},\Pred{M_1}$, observing that the cross-cases
          $\Boundary{r}\implies\Inl{\ldots} = \Inr{\ldots}$ and
          $\Boundary{r}\implies\Inr{\ldots}=\Inl{\ldots}$ are impossible.

      \end{enumerate}

    \item \emph{Case (universe).} We need to show that $\CodeUni{m}$ is
      typewise and elementwise separated for $m<n$.
      \begin{enumerate}
        \item\emph{Typewise separation.}
          We need to show that for all $\FmtCode{B}:\Pred{\ClUni{n}}B$, if
          $\Boundary{r}\implies\CodeUni{m}=\FmtCode{B}$, then
          $\CodeUni{m}=\FmtCode{B}$. But this is immediate by considering the
          restriction maps for $\ClUni{n}$.

        \item\emph{Elementwise separation.} Elementwise separation of
          $\CodeUni{m}$ follows from the typewise separation part of the outer
          induction hypothesis at $m<n$. \qedhere

      \end{enumerate}

  \end{enumerate}

\end{proof}

\subsection{The universe type and its decoding}\label{sec:universe-and-decoding}

Next, we define a hierarchy of closed universes \`a la Tarski
$\Gl{\CwfUniv[n]}\in\CwfTy<\Pred>{\Gl{\Gamma}}$. Note that these are \emph{not}
universes \`a la Russell in the sense of
\cref{def:universes-russell}. For the syntactic part, take
$\CwfUniv[n]\in\CwfTy{\Gamma}$ itself; then, we define the computability family
using $\Pred{\ClUni{n}}$:
\begin{align*}
  \Pred{\CwfUniv[n]} &:
  \textstyle
  \prod_{\gamma:\Nerve{\Gamma}}
  \prod_{\Pred{\gamma}:\Pred{\Gamma}\gamma}
  \NerveEl{\ReIx{\gamma}{\CwfUniv[n]}}\to \SemUniv[n+1]
  \\
  \ldots
  &:
  \textstyle
  \prod_{\gamma:\Nerve{\Gamma}}
  \prod_{\Pred{\gamma}:\Pred{\Gamma}\gamma}
  \NerveTy[n]\to \SemUniv[n+1]
  \\
  \Pred{\CwfUniv[n]}\gamma\Pred{\gamma}A &=
  \Pred{\ClUni{n}}A
\end{align*}

We equip each type with regular coercion and homogeneous composition structure
for these universes, by recursion on the type codes. For readability, we leave
syntactic arguments like $A,M,\ldots$ implicit.

\begin{mathparpagebreakable}
  \mprset{sep=1em}
  \Rule{
    r,r':\Dim
    \\
    \FmtCode{A}: \textstyle\prod_{i:\Dim}\Pred{\ClUni{n}}A_i
    \\
    \Pred{M} : \UPred{\FmtCode{A} r} M
  }{
    \coe{i.\FmtCode{A}i}{r}{r'}{\Pred{M}} :
    \UPred{\FmtCode{A}r'}\parens[\big]{
      \CwfCoe{i.A_i}{r}{r'}{M}
    }
  }
  \and
  \Rule{
    r,r':\Dim
    \\
    \FmtCode{A} : \Pred{\ClUni{n}}{A}
    \\
    \Pred{M}:\UPred{\FmtCode{A}}M
    \\
    \textstyle
    \etc{
      \Pred{N_\e} :
      \prod_{i:\Dim}
      (s=\e)\to
      \braces*{
        x: \Pred{\FmtCode{A}} N_i
        \mid
        i=r\implies x=\Pred{M}
      }
    }
  }{
    \hcom{\FmtCode{A}}{r}{r'}{\Pred{M}}{
      \etcsys[\e]{s}{i.\Pred{N_\e}i}
    } :
    \UPred{\FmtCode{A}}\parens*{
      \CwfHcom{A}{r}{r'}{M}{
        \etcsys[\e]{s}{i.N_i}
      }
    }
  }
\end{mathparpagebreakable}

\bigskip
\LightRule
\bigskip

\begin{align*}
  \coe{i.\CodeBool}{r}{r'}{\Pred{M}} &= \Pred{M}
  \\
  \coe{i.\CodeUni{n}}{r}{r'}{\FmtCode{A}} &= \FmtCode{A}
  \\
  \coe{i.\CodePi{\FmtCode{A}}{\FmtCode{B}}}{r}{r'}{\Pred{M}}
  &=
  \lambda \Pred{N}.\,
  \coe{
    i.\FmtCode{B}\parens{
      \coe{i.\FmtCode{A}}{r'}{i}{\Pred{N}}
    }
  }{r}{r'}{
    \Pred{M}\parens[\big]{
      \coe{i.\FmtCode{A}}{r'}{r}{\Pred{N}}
    }
  }
  \\
  \coe{i.\CodeSg{\FmtCode{A}}{\FmtCode{B}}}{r}{r'}{\parens*{\Pred{M}_0,\Pred{M}_1}}
  &=
  \parens[\big]{
    \coe{i.\FmtCode{A}}{r}{r'}{
      \Pred{M}_0
    },
    \coe{
      i.
      \FmtCode{B}\parens{
        \coe{i.\FmtCode{A}}{r}{i}{\Pred{M}_0}
      }
    }{r}{r'}{\Pred{M}_1}
  }
  \\
  \coe{
    i.\CodeEq{\FmtCode{A}}{\Pred{N_0}}{\Pred{N_1}}
  }{r}{r'}{\Pred{M}}
  &=
  \lambda k.\,
  \com{i.\FmtCode{A} k}{r}{r'}{
    \Pred{M}k
  }{
    \etcsys{k}{\_.\Pred{N_\e}}
  }
\end{align*}

\bigskip
\LightRule
\bigskip

\begin{mathparpagebreakable}
  \Rule{
    \widetilde{\Pred{M}} =
    \coe{i.\FmtCode{A}_i}{r}{r'}{\Pred{M}}
    \\
    \widetilde{\Pred{N\e}} =
    \lambda i.\, \coe{i.\FmtCode{A}_i}{i}{r'}{\Pred{N}_i}
  }{
    \com{i.\FmtCode{A}_i}{r}{r'}{\Pred{M}}{
      \etcsys{s}{i.\Pred{N_\e}i}
    }
    =
    \hcom{\FmtCode{A}_{r'}}{r}{r'}{
      \widetilde{\Pred{M}}
    }{
      \etcsys{s}{i.\Pred{N_\e} i}
    }
  }
\end{mathparpagebreakable}

\bigskip
\LightRule
\bigskip

\begin{mathparpagebreakable}
  \Rule{}{
    \hcom{\CodeUni{n}}{r}{r'}{\CodeBool}{
      \etcsys{s}{\_.\CodeBool}
    } = \CodeBool
  }
  \and
  \Rule{}{
    \hcom{\CodeUni{n}}{r}{r'}{\CodeUni{k}}{
      \etcsys{s}{\_.\CodeUni{k}}
    } = \CodeUni{k}
  }
  \and
  \Rule{
    \widetilde{\FmtCode{A}} =
    \lambda i.\,
    \hcom{\CodeUni{n}}{r}{i}{\FmtCode{A}}{
      \etcsys{s}{i.\FmtCode{A'}i}
    }
    \\
    \widetilde{\FmtCode{B}} =
    \lambda \Pred{N}.\,
    \hcom{\CodeUni{n}}{r}{r'}{
      \FmtCode{B}\parens[\big]{
        \coe{i.\widetilde{\FmtCode{A}}i}{r'}{r}{\Pred{N}}
      }
    }{
      \etcsys{s}{
        i.
        \FmtCode{B'}i\parens[\big]{
          \coe{i.\widetilde{\FmtCode{A}}i}{r'}{i}{\Pred{N}}
        }
      }
    }
  }{
    \hcom{\CodeUni{n}}{r}{r'}{
      \CodePi{\FmtCode{A}}{\FmtCode{B}}
    }{
      \etcsys{s}{i. \CodePi{\FmtCode{A'}i}{\FmtCode{B'}i}}
    } =
    \CodePi{
      \widetilde{\FmtCode{A}}r'
    }{
      \widetilde{\FmtCode{B}}
    }
  }
  \and
  \Rule{
    \widetilde{\FmtCode{A}} =
    \lambda i.\,
    \hcom{\CodeUni{n}}{r}{i}{\FmtCode{A}}{
      \etcsys{s}{i.\FmtCode{A'}i}
    }
    \\
    \widetilde{\FmtCode{B}} =
    \lambda \Pred{N}.\,
    \hcom{\CodeUni{n}}{r}{r'}{
      \FmtCode{B}\parens[\big]{
        \coe{i.\widetilde{\FmtCode{A}}i}{r'}{r}{\Pred{N}}
      }
    }{
      \etcsys{s}{
        i.
        \FmtCode{B'}i\parens[\big]{
          \coe{i.\widetilde{\FmtCode{A}}i}{r'}{i}{\Pred{N}}
        }
      }
    }
  }{
    \hcom{\CodeUni{n}}{r}{r'}{
      \CodeSg{\FmtCode{A}}{\FmtCode{B}}
    }{
      \etcsys{s}{i. \CodeSg{\FmtCode{A'}i}{\FmtCode{B'}i}}
    } =
    \CodeSg{
      \widetilde{\FmtCode{A}}r'
    }{
      \widetilde{\FmtCode{B}}
    }
  }
  \and
  \Rule{
    \widetilde{\FmtCode{A}} =
    \lambda j,i.\,
    \hcom{\CodeUni{n}}{r}{j}{
      \FmtCode{A} i
    }{
      \etcsys{s}{j. \FmtCode{A'} j i}
    }
    \\
    \widetilde{M} =
    \com{j.\widetilde{\FmtCode{A}}j r}{r}{r'}{\Pred{M}}{
      \etcsys{s}{j.\Pred{M'}j}
    }
    \\
    \widetilde{N} =
    \com{j.\widetilde{\FmtCode{A}}j r'}{r}{r'}{\Pred{N}}{
      \etcsys{s}{j.\Pred{N'}j}
    }
  }{
    \hcom{\CodeUni{n}}{r}{r'}{
      \CodeEq{\FmtCode{A}}{\Pred{M}}{\Pred{N}}
    }{
      \etcsys{s}{i.
        \CodeEq{\FmtCode{A}' i}{
          \Pred{M'}i
        }{
          \Pred{N'}i
        }
      }
    } =
    \CodeEq{
      \widetilde{\FmtCode{A}}r'
    }{
      \widetilde{\Pred{M}}
    }{
      \widetilde{\Pred{N}}
    }
  }
  \\
  \Rule{%
  }{
    \hcom{\CodeBool}{r}{r'}{\Inl{\Refl}}{
      \etcsys{s}{\_.\Inl{\Refl}}
    } =
    \Inl{\Refl}
  }
  \and
  \Rule{%
  }{
    \hcom{\CodeBool}{r}{r'}{\Inr{\Refl}}{
      \etcsys{s}{\_.\Inr{\Refl}}
    } =
    \Inr{\Refl}
  }
  \and
  \Rule{%
  }{
    \hcom{
      \CodePi{\FmtCode{A}}{\FmtCode{B}}
    }{r}{r'}{\Pred{M}}{
      \etcsys{s}{i.\Pred{M'}i}
    }
    =
    \lambda \Pred{N}.\,
    \hcom{
      \FmtCode{B}\Pred{N}
    }{r}{r'}{
      \Pred{M}\Pred{N}
    }{
      \etcsys{s}{i.\Pred{M'}i N \Pred{N}}
    }
  }
  \and
  \Rule{
    \widetilde{\Pred{M}} =
    \lambda j.\,
    \hcom{\FmtCode{A}}{r}{j}{\Pred{M}}{
      \etcsys{s}{i.\Pred{M'}i}
    }
    \\
    \widetilde{\Pred{N}} =
    \com{i.
      \FmtCode{B}\parens[\big]{
        \widetilde{\Pred{M}} i
      }
    }{r}{r'}{\Pred{N}}{
      \etcsys{s}{\Pred{N'}i}
    }
  }{
    \hcom{
      \CodeSg{\FmtCode{A}}{\FmtCode{B}}
    }{r}{r'}{
      \parens*{\Pred{M},\Pred{N}}
    }{
      \etcsys{s}{i.
        \parens*{
          \Pred{M'}i,
          \Pred{N'}i
        }
      }
    } =
    \parens[\big]{
      \widetilde{\Pred{M}} r',
      \widetilde{\Pred{N}}
    }
  }
  \and
  \Rule{}{
    \hcom{
      \CodeEq{\FmtCode{A}}{\Pred{N_0}}{\Pred{N_1}}
    }{r}{r'}{\Pred{M}}{
      \etcsys{s}{i.\Pred{M'}i}
    } =
    \lambda j.\,
    \hcom{
      \FmtCode{A} j
    }{r}{r'}{\Pred{M} j}{
      \etcsys{s}{i.\Pred{M'}i j}
    }
  }
\end{mathparpagebreakable}


Next, we show that every element
$\Gl{A}\equiv\parens{A,\FmtCode{A}}\in\CwfEl<\Pred>{\Gl{\Gamma}}{\Gl{\CwfUniv[n]}}$ determines a
type $\Decode{\Gl{A}}\in\CwfTy<\Pred>[n]{\Gl{\Gamma}}$. For the syntactic part,
choose $A\in\CwfEl{\Gamma}{\CwfUniv[n]}$ itself (which is possible because
$\CwfUniv[n]$ is a universe \`a la Russell in $\Cx$). The computability
family is given as follows:
\begin{align*}
  \Pred{\Decode{\Gl{A}}}\gamma\Pred{\gamma} &= \UPred*{\FmtCode{A}\gamma\Pred{\gamma}}
\end{align*}

\subsection{The \emph{closed-universe} computability cwf}\label{sec:closed-computability}

Now, we are equipped to build a \emph{new} cwf $\CPred$, which we will show to
be a model of \XTT{} in \cref{sec:xtt-computability-model}. Let the underlying category of $\CPred$ be the same as
$\Pred$'s; we will choose new presheaves of types and elements, however.
\begin{align*}
  \CwfTy<\CPred>[n]{\Gl{\Gamma}} &= \CwfEl<\Pred>{\Gl{\Gamma}}{\Gl{\CwfUniv[n]}}
  \\
  \CwfEl<\CPred>{\Gl{\Gamma}}{\Gl{A}} &=
  \CwfEl<\Pred>{\Gl{\Gamma}}{\Decode{\Gl{A}}}
\end{align*}

When $\Gl{A}\in\CwfTy<\CPred>[k]{\Gl{\Gamma}}$ and $k\leq l$, we exhibit the
level restriction by taking $\TyLift{k}{l}{A}$ for the syntactic part, and
retaining $\Pred{A}$ for the semantic part, a move justified by the fact that
the type codes are invariant under lifting (\cref{lem:type-code-lift}).

Given $\Gl{\Gamma}:\CPred$ and $\Gl{A}\in\CwfTy<\CPred>[n]{\Gl{\Gamma}}$, we
need to exhibit the context comprehension $\Gl{\Gamma}.\Gl{A}:\CPred$ with the
appropriate projection map $\Gl{\Proj}$ and variable term $\Gl{\Var}$. We
choose the already-existing context comprehension $\Gl{\Gamma}.\Decode{\Gl{A}}$
inherited from $\Pred$; the projection map and variable term are likewise
inherited.
From \cref{lem:type-code-lift} we also immediately obtain
$\CwfEl<\CPred>{\Gl{\Gamma}}{\Gl{A}} =
\CwfEl<\CPred>{\Gl{\Gamma}}{\TyLift{k}{l}{\Gl{A}}}$ and $\Gl{\Gamma}.\Gl{A} =
\Gl{\Gamma}.\TyLift{k}{l}{\Gl{A}}$.

The projection $\PiSyn$ lifts from $\Pred$ to a fibration
$\CPred\rProjto^{\PiSyn}\Cx$, because the underlying categories of $\Pred$ and
$\CPred$ are identical. A cubical structure for $\CPred$ is obtained from the
composite $\CPred\rProjto^{\PiSyn}\Cx\rProjto^{\CubeFib}\AugCube$.

\subsection{A logical families model of \texorpdfstring{\XTT}{XTT}}\label{sec:xtt-computability-model}

In this section, we argue that the cwf $\CPred$ has the structure of a
model of \XTT.

\begin{construction}\label{con:nerve-ext-dim}

  In preparation, we first will observe that \emph{internally to $\AugCSet$}, we
  have the following operations for any $\Gl{\Gamma} : \CPred$ and fresh $i$:
  \begin{align}
    (-.-) &:\Nerve{\Gamma}\times\Dim\to\Nerve{\ReIx{\Wk{\imath}}{\Gamma}}
    \\
    (-.-) &:
    \textstyle
    \prod_{\gamma:\Nerve{\Gamma}}
    \prod_{r:\Dim}
    \Pred{\Gamma}\gamma\to \Pred{\parens*{\ReIx{\Wk{\imath}}{\Gamma}}}\parens*{\gamma.r}
  \end{align}
  \begin{enumerate}

    \item
      Fix $\Psi$ and $\CubeCx{\Psi}\rTo^\gamma\Gamma$ and $\IsDim<\Psi>{r}$, defining:
      \begin{diagram}
        \CubeCx{\Psi} & \rTo^{\CubeCx{r/i}} & \CubeCx{\Psi,i}\equiv\ReIx{\Wk{\imath}}{\CubeCx{\Psi}} &\rTo^{\Lift{\Wk{\imath}}{\gamma}} & \ReIx{\Wk{\imath}}{\Gamma}
      \end{diagram}
      Naturality is just the associativity of composition.

    \item Fix $\Psi$ and $\CubeCx{\Psi}\rTo^\gamma\Gamma$ and $\Pred{\gamma}\in
      \Pred{\Gamma}_\Psi\gamma$ and $\IsDim<\Psi>{r}$. We need to construct some element of
      $\Pred{\parens*{\ReIx{\Wk{\imath}}{\Gamma}}}_{\Psi}\parens{\gamma.r}$.
      Observing that $(\gamma,\Pred{\gamma})$ constitute a map
      $\Gl{\CubeCx{\Psi}}\rTo^{\Gl{\gamma}} \Gl{\Gamma}$ in $\CPred$, we see
      that our goal is in fact to transform $\Gl{\gamma}$ into a map
      $\Gl{\CubeCx{\Psi}}\rTo\ReIx{\Wk{\imath}}{\Gl{\Gamma}}$ which lies over $\gamma.r$. This we obtain in the
      same way as before, this time using the composite fibration
      $\CPred\rProjto^{\PiSyn}\Cx\rProjto^{\CubeFib}\AugCube$:
      \begin{diagram}
        \Gl{\CubeCx{\Psi}} & \rTo^{\Gl{\CubeCx{r/i}}} & \Gl{\CubeCx{\Psi,i}}\equiv\ReIx{\Wk{\imath}}{\Gl{\CubeCx{\Psi}}} &\rTo^{\Lift{\Wk{\imath}}{\Gl{\gamma}}} & \ReIx{\Wk{\imath}}{\Gl{\Gamma}}
      \end{diagram}
      By analogy, we write this map as $\Pred{\gamma}.r$. \qed
  \end{enumerate}
\end{construction}

\begin{construction}\label{con:nerve-dim-proj}
  When $\Gamma:\Cx$ lies over $\Psi$ and $\IsDim<\Psi>{r}$, we obtain a
  natural transformation $\Nerve{\Gamma}\to\Dim$; fixing $\CubeCx{\Phi}\rTo^\gamma\Gamma$, we obtain:
  \begin{diagram}
    \CubeCx{\Phi} & \rTo^\gamma & \Gamma
    \\
    \dProjMapsto^{\CubeFib} && \dProjMapsto_{\CubeFib}
    \\
    \Phi & \rTo^{\CubeFib{\gamma}} & \Psi
    \\
    &\rdDotsto_{\ReIx{\CubeFib{\gamma}}{r}} &\dTo_{(r/i)}
    \\
    && [i]
  \end{diagram}
  We will write $\gamma[r]:\Dim$ given $\gamma:\Nerve{\Gamma}$ and $\IsDim<\Psi>{r}$ when working internally.
  \qed
\end{construction}

\begin{proposition}
  If $\eta:\Nerve{\Gamma.r=s}$, then $\eta[r]=\eta[s]:\Dim$.
\end{proposition}

\begin{lemma}\label{lem:cpred-coe}
  $\CPred$ has regular coercion structure in the sense of \cref{def:regular-coercion-structure}.
\end{lemma}

\begin{proof}

  Fixing $\Gl{A}\in\CwfTy<\CPred>{\ReIx{\Wk{\imath}}{\Gl{\Gamma}}}$ over $\Psi,i$
  and dimensions $\IsDim<\Psi>{r,r'}$ and an element
  $\Gl{M}\in\CwfEl{\Gl{\Gamma}}{\ReIxCleave*{r/i}{A}}$, we must construct an
  element $\CwfCoe{i.\Gl{A}}{r}{r'}{\Gl{M}}$ which satisfies the adjacency,
  regularity and naturality equations. Taking $\CwfCoe{i.A}{r}{r'}{M}$ for the
  syntactic part, we construct its realizer as follows:
  \[
    \Pred{
      \parens*{\CwfCoe{i.\Gl{A}}{r}{r'}{\Gl{M}}}
    }\gamma\Pred{\gamma}
    =
    \coe{
      j.\Pred{A}\parens{\gamma.j}\parens{\Pred{\gamma}.j}
    }{\gamma[r]}{\gamma[r']}{
      \parens*{\Pred{M}\gamma\Pred{\gamma}}
    }
  \]

  \begin{itemize}

    \item \emph{Adjacency.} If $r=r'$ then $\CwfCoe{i.\Gl{A}}{r}{r'}{\Gl{M}} =
      \Gl{M}$. This holds already for the syntactic part, so it remains to see
      that our realizer preserves it; this is proved by induction on the graph
      of the realizer for coercion in \cref{sec:universe-and-decoding}.

    \item \emph{Regularity.} If $\Gl{A}=\ReIxCleave{\Wk{\imath}}{\Gl{A'}}$ for some
      $\Gl{A'}\in\CwfTy<\CPred>{\Gl{\Gamma}}$, then
      $\CwfCoe{i.\Gl{A}}{r}{r'}{\Gl{M}}=\Gl{M}$. As above, we need only observe
      that the realizer exhibits regularity, which is evident by inspecting
      each of the clauses in \cref{sec:universe-and-decoding}.

    \item \emph{Naturality.} For $\Gl{\Delta}\rTo^{\Gl{\gamma}}\Gl{\Gamma}$ we
      have $\ReIx{\Gl{\gamma}}{\CwfCoe{i.\Gl{A}}{r}{r'}{\Gl{M}}} =
      \CwfCoe{i.\ReIx*{\Lift{\Wk{\imath}}{\Gl{\gamma}}}{\Gl{A}}}{\ReIx{\CubeFib{\Gl{\gamma}}}{r}}{\ReIx{\CubeFib{\Gl{\gamma}}}{r'}}{\ReIx{\Gl{\gamma}}{\Gl{M}}}$.
      As above, this holds already of the syntactic part, so we need only to
      verify for each $\delta:\Nerve{\Delta},\Pred{\delta}:\Pred{\Delta}\delta$
      the following:
      \[
        \coe{
          j.
          \Pred{A}
          \parens{(\ReIx{\gamma}{\delta}).j}
          \parens{(\ReIx{\gamma}{\Pred{\delta}}).j}
        }{
          \parens{\ReIx{\gamma}{\delta}}[r]
        }{
          \parens{\ReIx{\gamma}{\delta}}[r']
        }{
          \Pred{M}
          \parens{\ReIx{\gamma}{\delta}}
          \parens{\ReIx{\gamma}{\Pred{\delta}}}
        }
        =
        \coe{
          j.
          \parens{\ReIx*{\Lift{\Wk{\imath}}{\gamma}}{\Pred{A}}}
          \parens{\delta.j}
          \parens{\Pred{\delta}.j}
        }{
          \delta[\ReIx{\Gl{\gamma}}{r}]
        }{
          \delta[\ReIx{\Gl{\gamma}}{r'}]
        }{
          \parens*{\ReIx{\Gl{\gamma}}{\Pred{M}}}\delta\Pred{\delta}
        }
      \]

      The above follows from the naturality of \cref{con:nerve-dim-proj}.
      \qedhere
  \end{itemize}

\end{proof}

\begin{lemma}\label{lem:cpred-hcom}
  $\CPred$ has regular homogeneous composition structure in the sense of \cref{def:regular-hcom-structure}.
\end{lemma}
\begin{proof}

  Fix a type $\Gl{A}\in\CwfTy<\CPred>{\Gamma}$, dimensions
  $r,r',s\in\CwfDim{\Gl{\Gamma}}$, a cap $\Gl{M}\in\CwfEl{\Gl{\Gamma}}{\Gl{A}}$
  and a tube
  $\etc{\Gl{N_\e}\in\CwfEl{\ReIx{\Wk{\jmath}}{(\Gl{\Gamma}.s=\e)}}{\ReIxCleave{\Wk{\jmath}}{\ReIx{
    \Wk*{s=\e}}{\Gl{A}}}}}$ for fresh $j$ such that
  $\etc{\ReIxCleave*{r/j}{\Gl{N_\e}}=\Gl{M}}$. Choosing
  \begin{align*}
    &\Pred{
      \parens[\Big]{\CwfHcom{\Gl{A}}{r}{r'}{\Gl{M}}{\etcsys[\e]{s}{j.\Gl{N_\e}}}}
    }\gamma\Pred{\gamma}
    \\
    &=
    \hcom{\Pred{A}\gamma\Pred{\gamma}}{\gamma[r]}{\gamma[r']}{\Pred{M}\gamma\Pred{\gamma}}{
      \etcsys{s}{j.
        \Pred{N_\e}\parens*{
          \gamma.s=\e.j
        }\parens*{
          \Pred{\gamma}.j
        }
      }
    }
  \end{align*}

  The adjacency and regularity conditions hold by induction on the graph of
  the realizer defined in \cref{sec:universe-and-decoding}. The naturality condition
  lifts directly from $\Cx$ as in \cref{lem:cpred-coe}.

\end{proof}

\begin{lemma}
  $\CPred$ has the boolean type, in the sense of \cref{def:booleans}.
\end{lemma}
\begin{proof}
  \hfill
  \begin{itemize}

    \item \emph{Formation.} To exhibit
      $\Gl{\CwfBool}\in\CwfTy<\CPred>[n]{\Gl{\Gamma}}$, we choose $\CwfBool$
      itself for the syntactic part, and for its realizer we choose
      $\Pred{\CwfBool}\gamma\Pred{\gamma}=\CodeBool$.

    \item \emph{Introduction.} Choosing $\CwfTrue$ and $\CwfFalse$ for the
      syntactic parts, we exhibit realizers as follows:
      \begin{align*}
        \Pred{\CwfTrue}\gamma\Pred{\gamma} &= \Inl{\Refl}
        \\
        \Pred{\CwfFalse}\gamma\Pred{\gamma} &= \Inr{\Refl}
      \end{align*}

    \item \emph{Elimination.}
      Fixing $\Gl{C}\in\CwfTy<\CPred>[n]{\Gl{\Gamma}.\Gl{\CwfBool}}$ and
      $\Gl{M}\in\CwfEl<\CPred>{\Gl{\Gamma}}{\Gl{\CwfBool}}$ and
      $\Gl{N_0}\in\CwfEl<\CPred>{\Gl{\Gamma}}{\ReIx{\Snoc{\Id}{\Gl{\CwfTrue}}}{C}}$
      and
      $\Gl{N_1}\in\CwfEl<\CPred>{\Gl{\Gamma}}{\ReIx{\Snoc{\Id}{\Gl{\CwfFalse}}}{C}}$,
      we choose $\CwfIf{C}{M}{N_0}{N_1}\in\CwfEl{\Gamma}{\ReIx{\Snoc{\Id}{M}}}{C}$
      for the syntactic part, exhibiting its realizer as follows:
      \begin{align*}
        \Pred{
          \CwfIf{\Gl{C}}{\Gl{M}}{\Gl{N_0}}{\Gl{N_1}}
        }\gamma\Pred{\gamma}
        &=
        \begin{cases}
          \Pred{N_0}\gamma\Pred{\gamma}
          &\mbox{if } \Pred{M}\gamma\Pred{\gamma} = \Inl{\ldots}
          \\
          \Pred{N_1}\gamma\Pred{\gamma}
          &\mbox{if } \Pred{M}\gamma\Pred{\gamma} = \Inr{\ldots}
        \end{cases}
      \end{align*}

    \item \emph{Computation, naturality.} The satisfaction of the equational
      conditions is immediate.\qedhere

  \end{itemize}
\end{proof}

\begin{lemma}\label{thm:cpred-dependent-function-types}
  $\CPred$ has dependent function types in the sense of \cref{def:dependent-function-types}.
\end{lemma}
\begin{proof}
  \hfill
  \begin{itemize}

    \item \emph{Formation.}
      Fixing $\Gl{A}\in\CwfTy<\CPred>[n]{\Gl{\Gamma}}$
      and a family $\Gl{B}\in\CwfTy<\CPred>[n]{\Gl{\Gamma}.\Gl{A}}$, we must
      exhibit a type $\CwfPi{\Gl{A}}{\Gl{B}}\in\CwfTy<\CPred>[n]{\Gl{\Gamma}}$.
      Choosing $\CwfPi{A}{B}$ for the syntactic part, we must exhibit a
      realizer for the type:
      \[
        \Pred{\CwfPi{\Gl{A}}{\Gl{B}}}\gamma\Pred{\gamma}
        =
        \CodePi{\Pred{A}\gamma\Pred{\gamma}}{
          \lambda N,\Pred{N}.\,
          \Pred{B}\Snoc{\gamma}{N}(\Pred{\gamma},\Pred{N})
        }
      \]

    \item \emph{Introduction.} Fixing
      $\Gl{M}\in\CwfEl<\CPred>{\Gl{\Gamma}.\Gl{A}}{\Gl{B}}$, we must exhibit an
      element
      $\CwfLam[\Gl{A}][\Gl{B}]{\Gl{M}}\in\CwfEl<\CPred>{\Gl{\Gamma}}{\CwfPi{\Gl{A}}{\Gl{B}}}$.
      Choosing $\CwfLam[A][B]{M}$ for the syntactic part, we code its realizer as
      follows:
      \[
        \Pred{\CwfLam[\Gl{A}][\Gl{B}]{\Gl{M}}}\gamma\Pred{\gamma} N\Pred{N} =
        \Pred{M}\Snoc{\gamma}{N}(\Pred{\gamma},\Pred{N})
      \]

    \item \emph{Elimination.} Fixing
      $\Gl{M}\in\CwfEl<\CPred>{\Gl{\Gamma}}{\CwfPi{\Gl{A}}{\Gl{B}}}$ and
      $\Gl{N}\in\CwfEl<\CPred>{\Gl{\Gamma}}{\Gl{A}}$, we need an element
      $\CwfApp[\Gl{A}][\Gl{B}]{\Gl{M}}{\Gl{N}}\in\CwfEl<\CPred>{\Gl{\Gamma}}{\ReIx{\Snoc{\Id}{\Gl{N}}}{\Gl{B}}}$.
      Choosing $\CwfApp[A][B]{M}{N}$ for the syntactic part, we code its realizer:
      \[
        \Pred{\CwfApp[\Gl{A}][\Gl{B}]{\Gl{M}}{\Gl{N}}}\gamma\Pred{\gamma} =
        \Pred{M}\gamma\Pred{\gamma} N\Pred{N}
      \]

    \item \emph{Computation, unicity, naturality.} These equations follow
      immediately from the fact that they hold of $\Cx$, and the corresponding
      properties of the dependent function types in $\AugCSet$.

    \item \emph{Coercion.} Supposing $\Gl{\Gamma} \rProjMapsto^\CubeFib \Psi,i$ and
      $\IsDim<\Psi>{r,r'}$ and
      $\Gl{M}\in\CwfEl<\CPred>{\ReIx*{r/i}{\Gl{\Gamma}}}{\ReIxCleave*{r/i}{\CwfPi{\Gl{A}}{\Gl{B}}}}$,
      we need to verify the following equation:
      \begin{mathpar}
        \CwfCoe{i.\CwfPi{\Gl{A}}{\Gl{B}}}{r}{r'}{\Gl{M}}
        =
        \CwfLam[
          \ReIxCleave*{r'/i}{\Gl{A}}
        ][
          \ReIxCleave*{r'/i}{\Gl{B}}
        ]{
          \CwfCoe{i.
            \ReIx{
              \Snoc{\Id}{
                \CwfCoe{i.\ReIx{\Proj}{\Gl{A}}}{r'}{i}{\Var}
              }
            }{\Gl{B}}
          }{r}{r'}{
            \CwfApp[\ReIx{\Proj}{\ReIxCleave*{r/i}{\Gl{A}}}][\ReIx{\CwfLift{\Proj}}{\ReIxCleave*{r/i}{\Gl{B}}}]{
              \ReIx{\Proj}{\Gl{M}}
            }{\CwfCoe{i. \ReIx{\Proj}{\Gl{A}}}{r'}{r}{\Var}}
          }
        }
      \end{mathpar}

      As this holds already of the syntactic part, we need only verify that it
      holds of the realizers (determined in \cref{lem:cpred-coe}); but this is immediate by the definition of
      realizer for coercion at $\CodePi{\FmtCode{A}}{\FmtCode{B}}$ in
      \cref{sec:universe-and-decoding}.
      \qedhere
  \end{itemize}

\end{proof}

\begin{lemma}
  $\CPred$ has dependent pair types in the sense of \cref{def:dependent-pair-types}.
\end{lemma}
\begin{proof}
  \hfill
  \begin{itemize}

    \item \emph{Formation.} Fixing a type
      $\Gl{A}\in\CwfTy<\CPred>[n]{\Gl{\Gamma}}$ and a family
      $\Gl{B}\in\CwfTy<\CPred>[n]{\Gl{\Gamma}.\Gl{A}}$, we choose
      $\CwfSg{A}{B}$ for the syntactic part, and exhibit its realizer as
      follows:
      \[
        \Pred{\CwfSg{\Gl{A}}{\Gl{B}}}\gamma\Pred{\gamma} =
        \CodeSg{\Pred{A}\gamma\Pred{\gamma}}{
          \lambda N,\Pred{N}.\,
          \Pred{B}\Snoc{\gamma}{N}(\Pred{\gamma},\Pred{N})
        }
      \]

    \item \emph{Introduction.} Fixing
      $\Gl{M}\in\CwfEl<\CPred>{\Gl{\Gamma}}{\Gl{A}}$ and
      $\Gl{N}\in\CwfEl<\CPred>{\Gl{\Gamma}}{\ReIx{\Snoc{\Id}{\Gl{M}}}{\Gl{B}}}$, we choose
      $\CwfPair[A][B]{M}{N}$ for the syntactic part, coding its realizer as follows:
      \[
        \Pred{\CwfPair[\Gl{A}][\Gl{B}]{\Gl{M}}{\Gl{N}}}\gamma\Pred{\gamma} =
        \parens*{
          \Pred{M}\gamma\Pred{\gamma},\Pred{N}\gamma\Pred{\gamma}
        }
      \]

    \item \emph{Elimination.} Given
      $\Gl{M}\in\CwfEl<\CPred>{\Gl{\Gamma}}{\CwfSg{\Gl{A}}{\Gl{B}}}$, we choose
      $\CwfFst[A][B]{M}\in\CwfEl{\Gamma}{A}$ and $\CwfSnd[A][B]{M}$ for the syntactic parts
      of the elimination forms, and give their realizers as follows:
      \begin{align*}
        \Pred{\CwfFst[\Gl{A}][\Gl{B}]{\Gl{M}}}\gamma\Pred{\gamma} &= \pi_1\parens*{\Pred{M}\gamma\Pred{\gamma}}
        \\
        \Pred{\CwfSnd[\Gl{A}][\Gl{B}]{\Gl{M}}}\gamma\Pred{\gamma} &= \pi_2\parens*{\Pred{M}\gamma\Pred{\gamma}}
      \end{align*}

    \item \emph{Computation, unicity, naturality.} These are all immediate from
      the fact that they hold in $\Cx$, and the fact that analogous principles
      hold for the dependent pair types of $\AugCSet$.

    \item \emph{Coercion.} Analogous to \cref{thm:cpred-dependent-function-types}.
      \qedhere

  \end{itemize}
\end{proof}

\begin{lemma}
  $\CPred$ has dependent path types in the sense of \cref{def:path-types}.
\end{lemma}
\begin{proof}
  Here, we make use of \cref{con:nerve-ext-dim,con:nerve-dim-proj}.
  \begin{itemize}

    \item \emph{Formation.} Fixing
      $\Gl{A}\in\CwfTy<\CPred>[n]{\ReIx{\Wk{\imath}}{\Gl{\Gamma}}}$ over $\Psi,i$
      and elements
      $\etc{\Gl{N_\e}\in\CwfEl{\Gl{\Gamma}}{\ReIxCleave*{\e/i}{\Gl{A}}}}$, we
      choose $\CwfEq{i.A}{N_0}{N_1}$ for the syntactic part, coding its
      realizer as follows:
      \[
        \Pred{
          \CwfEq{i.\Gl{A}}{\Gl{N_0}}{\Gl{N_1}}
        }\gamma\Pred{\gamma}
        =
        \CodeEq{
          \lambda j.\,
          \Pred{A}\parens*{\gamma.j}\parens*{\Pred{\gamma}.j}
        }{
          \Pred{N_0}\gamma\Pred{\gamma}
        }{
          \Pred{N_1}\gamma\Pred{\gamma}
        }
      \]

    \item \emph{Introduction.} Given
      $\Gl{M}\in\CwfEl<\CPred>{\ReIx{\Wk{\imath}}{\Gl{\Gamma}}}{\Gl{A}}$, we choose
      $\CwfPathLam[i.A]{i.M}$ for the syntactic part, and exhibit its realizer as
      follows:
      \[
        \Pred{\CwfPathLam[i.\Gl{A}]{i.\Gl{M}}}\gamma\Pred{\gamma} =
        \lambda j.\,
        \Pred{M}\parens*{\gamma.j}\parens*{\Pred{\gamma}.j}
      \]

    \item \emph{Elimination.} Fixing
      $\Gl{M}\in\CwfEl<\CPred>{\Gl{\Gamma}}{\CwfPath{i.\Gl{A}}{\Gl{N_0}}{\Gl{N_1}}}$ and $\IsDim<\Psi>{r}$, we
      choose $\CwfPathApp[i.A]{M}{r}$ for the syntactic part; its realizer is analogous:
      \[
        \Pred{\CwfPathApp[i.\Gl{A}]{\Gl{M}}{r}}\gamma\Pred{\gamma} =
        \Pred{M}\gamma\Pred{\gamma}\parens*{\gamma[r]}
      \]

    \item \emph{Computation, boundary, unicity and naturality.} Immediate.

    \item \emph{Coercion.} Analogous to \cref{thm:cpred-dependent-function-types}.
      \qedhere

  \end{itemize}

\end{proof}

\begin{lemma}
  $\CPred$ has universes \`a la Russell in the sense of \cref{def:universes-russell}.
\end{lemma}
\begin{proof}
  Fixing $\Gl{\Gamma}:\CPred$ and levels $k<l$, we need a type
  $\Gl{\CwfUniv[k]}\in\CwfTy<\CPred>[l]{\Gl{\Gamma}}$ such that
  $\CwfEl<\CPred>{\Gl{\Gamma}}{\Gl{\CwfUniv[k]}} =
  \CwfTy<\CPred>[k]{\Gl{\Gamma}}$. For the syntactic part, we simply choose
  $\CwfUniv[k]$; for its realizer:
  \[
    \Pred{\CwfUniv[k]}\gamma\Pred{\gamma} =
    \CodeUni{k}
  \]

  To see that the condition is met, we first observe that
  $\CwfEl<\CPred>{\Gl{\Gamma}}{\Gl{\CwfUniv[k]}} =
  \CwfEl<\Pred>{\Gl{\Gamma}}{\Decode{\Gl{\CwfUniv[k]}}}$; and moreover,
  $\CwfTy<\CPred>[k]{\Gl{\Gamma}} = \CwfEl<\Pred>{\Gl{\Gamma}}{\Gl{\CwfUniv[k]}}$. Therefore, it suffices to show that
  $\Decode{\Gl{\CwfUniv[k]}} = \Gl{\CwfUniv[k]}$.
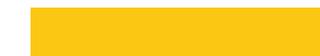
  \begin{align*}
    \Decode{\Gl{\CwfUniv[k]}} &=
    \parens*{\CwfUniv[k],\lambda\gamma\Pred{\gamma}. \UPred*{\Pred{\CwfUniv[k]}\gamma\Pred{\gamma}}}
    \\
    &=
    \parens*{\CwfUniv[k],\lambda\gamma\Pred{\gamma}. \UPred{\CodeUni{k}}}
    \\
    &=
    \parens*{\CwfUniv[k],\lambda\gamma\Pred{\gamma}. \Pred{\ClUni{k}}}
    \\
    &=
    \Gl{\CwfUniv[k]} \qedhere
  \end{align*}

\end{proof}

\begin{lemma}\label{lem:cpred-path-unicity}
  $\CPred$ has boundary separation in the sense of \cref{def:boundary-separation}.
\end{lemma}
\begin{proof}
  Because $\Cx$ has boundary separation, we need only to see that this property lifts
  to the realizers. Therefore, it suffices to show the following:
  \begin{itemize}

    \item \emph{Types.} For all $\Gl{A},\Gl{B}\in\CwfTy<\CPred>{\Gl{\Gamma}}$
      and $r\in\CwfDim<\CPred>{\Gl{\Gamma}}$, we must verify the following implication:
      \[
        \parens{
          \forall\gamma\Pred{\gamma}.\,
          \etc{\Pred{A}(\gamma.\ReIx{\gamma}{r}=\e)\Pred{\gamma}=\Pred{B}(\gamma.\ReIx{\gamma}{}r=\e)\Pred{\gamma}}
        }
        \implies
        \forall\gamma\Pred{\gamma}.\,
        \Pred{A}\gamma\Pred{\gamma}=\Pred{B}\gamma\Pred{\gamma}
      \]
      Fixing $\gamma$ and $\Pred{\gamma}$, it suffices to show:
      \[
        \etc{\Pred{A}(\gamma.\ReIx{\gamma}{r}=\e)\Pred{\gamma}=\Pred{B}(\gamma.\ReIx{\gamma}{r}=\e)\Pred{\gamma}}
        \implies
        \Pred{A}\gamma\Pred{\gamma}=\Pred{B}\gamma\Pred{\gamma}
      \]
      But this is equivalent to the following, which is obtained from the typewise separation of $\Pred{A}\gamma\Pred{\gamma}$, a consequence of
      \cref{lem:inductive-path-unicity}:
      \[
        \parens{
          \Boundary{\ReIx{\gamma}{r}}
          \implies
          \Pred{A}\gamma\Pred{\gamma}=\Pred{B}\gamma\Pred{\gamma}
        }
        \implies
        \Pred{A}\gamma\Pred{\gamma}=\Pred{B}\gamma\Pred{\gamma}
      \]

    \item \emph{Elements.} For all $\Gl{A}\in\CwfTy<\CPred>{\Gl{\Gamma}}$ and
      $M,N\in\CwfTy<\CPred>{\Gl{\Gamma}}{\Gl{A}}$ and
      $r\in\CwfDim<\CPred>{\Gl{\Gamma}}$, we must verify the following
      implication:
      \[
        \parens{
          \forall\gamma\Pred{\gamma}.\,
          \etc{\Pred{M}(\gamma.\ReIx{\gamma}{r}=\e)\Pred{\gamma}=\Pred{N}(\gamma.\ReIx{\gamma}{}r=\e)\Pred{\gamma}}
        }
        \implies
        \forall\gamma\Pred{\gamma}.\,
        \Pred{M}\gamma\Pred{\gamma}=\Pred{N}\gamma\Pred{\gamma}
      \]

      This follows in an analogous way to the above from
      \cref{lem:inductive-path-unicity}, using the elementwise separation of
      $\Pred{A}$. \qedhere

  \end{itemize}

\end{proof}

\begin{lemma}
  $\CPred$ has type-case in the sense of \cref{def:type-case}.
\end{lemma}
\begin{proof}
  We fix the following glued data:
  \begin{mathpar}
    \Gl{C}\in\CwfTy<\CPred>{\Gl{\Gamma}}
    \and
    \Gl{X}\in\CwfEl<\CPred>{\Gl{\Gamma}}{\Gl{\CwfUniv[k]}}
    \and
    \Gl{M_\Pi},\Gl{M_\Sigma}\in\CwfEl<\CPred>{\Gl{\Gamma}.\Gl{\CwfUniv[k]}.\CwfPi{\Var}{\Gl{\CwfUniv[k]}}}{\ReIx*{\Proj\circ\Proj}{\Gl{C}}}
    \and
    \Gl{M_{\mathbf{Eq}}}\in\CwfEl<\CPred>{
      \Gl{\Gamma}.\Gl{\CwfUniv[k]}.\Gl{\CwfUniv[k]}.%
      \CwfEq{\_.\Gl{\CwfUniv[k]}}{\ReIx{\Proj}{\Var}}{\Var}.%
      \ReIx*{\Proj\circ\Proj}{\Var}.%
      \ReIx*{\Proj\circ\Proj}{\Var}
    }{
      \ReIx*{
        \Proj\circ\Proj\circ\Proj\circ\Proj\circ\Proj
      }{\Gl{C}}
    }
    \and
    \Gl{M_{\CwfBool}}\in\CwfEl<\CPred>{\Gl{\Gamma}}{\Gl{C}}
    \and
    \Gl{M_{\mathbf{U}}}\in\CwfEl<\CPred>{\Gl{\Gamma}}{\Gl{C}}
  \end{mathpar}

  We need to exhibit an element
  $\CwfUCase{\Gl{C}}{\Gl{X}}{\Gl{M_\Pi}}{\Gl{M_\Sigma}}{\Gl{M_{\mathbf{Eq}}}}{\Gl{M_{\CwfBool}}}{\Gl{M_{\mathbf{U}}}}\in\CwfEl<\CPred>{\Gl{\Gamma}}{\Gl{C}}$
  with the specified computation and naturality rules. Inheriting the syntactic
  part from $\Cx$ as
  $\CwfUCase{C}{X}{M_\Pi}{M_\Sigma}{M_{\mathbf{Eq}}}{M_{\CwfBool}}{M_{\mathbf{U}}}\in\CwfEl{\Gamma}{C}$,
  it remains to define its realizer:
  \begin{align*}
    &\Pred{
      \parens*{
        \CwfUCase{\Gl{C}}{\Gl{X}}{\Gl{M_\Pi}}{\Gl{M_\Sigma}}{\Gl{M_{\mathbf{Eq}}}}{\Gl{M_{\CwfBool}}}{\Gl{M_{\mathbf{U}}}}
      }
    }
    \gamma\Pred{\gamma}
    =
    \\
    &\begin{cases}
      \Pred{M_\Pi}
      \angles{\gamma,A,\CwfLam[A][\CwfUniv[k]]{B}}
      \parens{
        \parens{\Pred{\gamma},\FmtCode{A}},
        \FmtCode{B}
      }
      & \mbox{if }
      \Pred{X}\gamma\Pred{\gamma} = \CodePi{\FmtCode{A}}{\FmtCode{B}} : \Pred{\ClUni{k}}\CwfPi{A}{B}
      \\
      \Pred{M_\Sigma}
      \angles{\gamma,A,\CwfLam[A][\CwfUniv[k]]{B}}
      \parens{
        \parens{\Pred{\gamma},\FmtCode{A}},
        \FmtCode{B}
      }
      & \mbox{if }
      \Pred{X}\gamma\Pred{\gamma} = \CodeSg{\FmtCode{A}}{\FmtCode{B}} : \Pred{\ClUni{k}}\CwfSg{A}{B}
      \\
      \!\!\!\begin{array}[t]{l}
        \Pred{M_{\mathbf{Eq}}}
        \\
        \begin{array}[t]{l}
          \angles{
            \gamma,\ReIxCleave*{0/i}{A},\ReIxCleave*{1/i}{A},\CwfPathLam[\_.\CwfUniv[k]]{i.A},N_0,N_1
          }
          \\
          \parens{
            \Pred{\gamma},\FmtCode{A}0,\FmtCode{A}1,\FmtCode{A},\Pred{N_0},\Pred{N_1}
          }
        \end{array}
      \end{array}
      & \mbox{if }
      \Pred{X}\gamma\Pred{\gamma} = \CodeEq{\FmtCode{A}}{\Pred{N_0}}{\Pred{N_1}} : \Pred{\ClUni{k}}{\CodeEq{i.A}{N_0}{N_1}}
      \\
      \Pred{M_{\CwfBool}}\gamma\Pred{\gamma}
      & \mbox{if } \Pred{X}\gamma\Pred{\gamma} = \CodeBool
      \\
      \Pred{M_{\mathbf{U}}}\gamma\Pred{\gamma}
      & \mbox{if } \Pred{X}\gamma\Pred{\gamma} = \CodeUni{k'}
    \end{cases}
  \end{align*}

  The required equations follow by calculation.
\end{proof}

\begin{corollary}
  $\CPred$ is a model of \XTT{} and moreover, $\CPred\rProjto^{\PiSyn}\Cx$
  is a homomorphism of \XTT-algebras.
\end{corollary}

\nocite{martin-lof:1979}
\nocite{taylor:1999}
\nocite{kaposi-kovacs-altenkirch:2019}
\nocite{abel-coquand-dybjer:2008}
\nocite{fiore-simpson:1999}

\end{document}